\documentclass[12pt, reqno]{amsart}

\usepackage{microtype}

\usepackage{graphics, stackrel}
\usepackage{amsmath, amssymb, amsthm, amsfonts, mathtools}
\usepackage{verbatim}

\usepackage{graphicx}

\usepackage[flushleft,para]{threeparttable}
\usepackage{array, booktabs, tabularx, dcolumn}
\newcolumntype{d}[1]{D{.}{.}{#1}}
\newcommand{\mcone}[1]{\multicolumn{1}{c}{#1}}

\usepackage{natbib}
\usepackage[T1]{fontenc}
\usepackage{mathpazo}

\usepackage{fancyvrb}
\usepackage{color}

\usepackage{enumitem}
\setlist[enumerate]{itemsep=2pt,topsep=3pt}
\setlist[itemize]{itemsep=2pt,topsep=3pt}
\setlist[enumerate,1]{label={\upshape (\roman*)}}

\usepackage[plain]{fancyref}%
\frefformat{plain}{\fancyreffiglabelprefix}{{\Freffigname}\fancyrefdefaultspacing#1}%
\frefformat{plain}{\fancyreftablabelprefix}{{\Freftabname}\fancyrefdefaultspacing#1}%
\frefformat{plain}{\fancyrefseclabelprefix}{{\Frefsecname}\fancyrefdefaultspacing#1}%
\newcommand*{\fancyrefprlabelprefix}{pr}
\newcommand{\Frefprname}{Proposition}
\frefformat{plain}{\fancyrefprlabelprefix}{{\Frefprname}\fancyrefdefaultspacing#1}%
\Frefformat{plain}{\fancyrefprlabelprefix}{{\Frefprname}\fancyrefdefaultspacing#1}%

\usepackage{xr-hyper}%
\usepackage{hyperref}
\hypersetup{%
  citecolor=blue,%
  colorlinks=true,%
  linkcolor=blue,
  bookmarksopen=true,%
  naturalnames=true,%
  colorlinks=true,%
  pdfauthor={Christophe Gouel, Qingyin Ma, John Stachurski},%
  pdftitle={Interest rate dynamics and commodity prices},%
  pdfstartview={FitH}}%
\usepackage{bookmark}%
\bookmarksetup{%
  open,%
  numbered,%
  depth=5, 
}%

\usepackage{mathrsfs}
\usepackage{bbm}

\usepackage{fullpage}

\usepackage{multirow}

\usepackage{algorithm}

\DeclareMathOperator{\Cov}{Cov}

\DeclareMathOperator{\EE}{\mathbbm E} 
\DeclareMathOperator{\diff}{\mathrm{d}} 

\newcommand{\setntn}[2]{ \{ #1 : #2 \} }

\newcommand{\iidsim}{\stackrel {\textrm{ {\sc iid }}} {\sim} }
\newcommand{\1}{\mathbbm 1}

\renewcommand{\epsilon}{\varepsilon}
\renewcommand{\phi}{\varphi}

\newcommand{\cC}{\mathscr C}

\newcommand{\bB}{\mathscr B}

\newcommand{\XX}{\mathsf X}

\newcommand{\ZZ}{\mathsf Z}
\renewcommand{\SS}{\mathsf S}

\newcommand{\RR}{\mathbbm R}

\newcommand{\NN}{\mathbbm N}

\providecommand{\inner}[1]{\left\langle{#1}\right\rangle}

\theoremstyle{plain}

\newtheorem{theorem}{Theorem}[section]
\newtheorem{corollary}{Corollary}[section]
\newtheorem{lemma}{Lemma}[section]
\newtheorem{proposition}{Proposition}[section]

\theoremstyle{definition}

\newtheorem{example}{Example}[section]
\newtheorem{remark}{Remark}[section]

\newtheorem{assumption}{Assumption}[section]

\DeclareMathOperator{\diag}{diag}
\renewcommand{\underline}[1]{\text{\b{$#1$}}} 
  
\usepackage{caption}

\newcommand{\me}{\mathrm{e}}

\usepackage{dsfont}

\newcommand{\vertiii}[1]{{\left\vert\kern-0.25ex\left\vert\kern-0.25ex\left\vert #1 
    \right\vert\kern-0.25ex\right\vert\kern-0.25ex\right\vert}}
    
\usepackage{subcaption}

\makeatletter
\newcommand*{\rom}[1]{\expandafter\@slowromancap\romannumeral #1@}
\makeatother

\usepackage{stackengine}

\usepackage{rotating}
\usepackage{lscape}

\makeatletter
\g@addto@macro{\endabstract}{\@setabstract}

\makeatletter
\def\@xfootnote[#1]{%
	\protected@xdef\@thefnmark{#1}%
	\@footnotemark\@footnotetext}
\makeatother

\begin{document}
\begin{center}
	\LARGE
	Interest Rate Dynamics and Commodity Prices\footnote[$\dagger$]{We 
		thank the editor Guillermo Ordonez, two anonymous referees, and an 
		anonymous associate editor, as well as Erica Perego, John Rust, Liyan 
		Yang, Changhua Yu, Shenghao Zhu, and seminar audiences at CEPII, 
		CUHK-SZ, Peking University, and XJTU for very helpful comments and 
		suggestions. Qingyin Ma gratefully acknowledges the financial support 
		from Natural Science Foundation of China (No.~72003138) and the Project 
		of Construction and Support for High-Level Innovative Teams of Beijing 
		Municipal Institutions (No.~BPHR20220119). \\
		E-mail addresses: \href{mailto:christophe.gouel@inrae.fr}{christophe.gouel@inrae.fr}, \href{mailto:qingyin.ma@cueb.edu.cn}{qingyin.ma@cueb.edu.cn}, and \href{mailto:john.stachurski@anu.edu.au}{john.stachurski@anu.edu.au}.} 
	\par \bigskip \medskip 
	
	\normalsize
	Christophe Gouel\textsuperscript{$\diamond$},
	Qingyin Ma\textsuperscript{$\ast$},
	and John Stachurski\textsuperscript{$\star$} \par \bigskip
	
	{\footnotesize
	\textsuperscript{$\diamond$}Université Paris-Saclay, INRAE, AgroParisTech, Paris-Saclay Applied Economics, \par 
	and CEPII, 20 avenue de Ségur, Paris, France \par \smallskip
	\textsuperscript{$\ast$}ISEM, Capital University of Economics and Business\par\smallskip
	\textsuperscript{$\star$}Research School of Economics, Australian National University 
	\par \bigskip}	
	\today 
\end{center}

\vspace{-1em}

\begin{abstract}
  \footnotesize
  \setlength\parindent{0pt}%
  \setlength\parskip{1ex  plus  0.6ex  minus  0.4ex}%
  \noindent
  In economic studies and popular media, interest rates are routinely cited
as a major factor behind commodity price fluctuations. At the same time,
the transmission channels are far from transparent, leading to long-running
debates on the sign and magnitude of interest rate effects. Purely empirical
studies struggle to address these issues because of the complex interactions
between interest rates, prices, supply changes, and aggregate demand. To move
this debate to a solid footing, we extend the competitive storage model to
include stochastically evolving interest rates. We establish general conditions
for existence and uniqueness of solutions and provide a systematic
theoretical and quantitative analysis of the interactions between interest rates
and prices. 


  
\end{abstract}


\section{Introduction}\label{sec:introduction}

Commodity prices are major determinants of exchange rates, government revenue,
the balance of payments, output fluctuations, and inflation \citep[see,
e.g.,][]{byrne2013primary, gospodinov2013commodity, eberhardt2021commodity,
  peersman2022international}.  While some commodity price movements are driven
by idiosyncratic shocks, \citet{alquist2020commodity} find that common factors
explain up to 80\% of the variance of commodity prices \citep[see
also][]{byrne2013primary}. Aggregate factors are particularly important when
considering the impact of commodities on inflation and exchange rates because
such factors induce price comovement in all or many commodities.

Historically, the aggregate factor that has generated the most attention is interest
rates. For example, \citet{frankel2008explanation, frankel2008monetary, frankel2018} has long argued that interest rates are a major driver
of comovements in commodity prices, with rising interest rates decreasing
commodity prices and vice versa. The main argument
relates to cost of carry: higher interest rates reduce demand for
inventories, which exerts downward pressure on commodity prices. At the same
time, it is easy to imagine scenarios where interest rates and commodity prices
are \emph{positively} correlated---for example, when high aggregate demand
boosts both commodity prices and the cost of borrowing (through credit markets
and, potentially, the responses of monetary authorities).

Indeed, empirical studies of the sign and magnitude
of interest rate effects on commodity prices face profound challenges because of the
endogeneity and equilibrium nature of the mechanisms in question.
For example, even if we fully
control for changes in output and demand, rising commodity prices might
themselves trigger a tightening of monetary policy, without any change in output
\citep[see, e.g.,][]{cody1991role}. Conversely, pure monetary shocks affect
commodity markets through various channels (e.g., speculation, aggregate demand,
and supply response) that are hard to disentangle empirically.\footnote{For
    example, a decline in the US interest rate can stimulate global demand
    \citep[see, e.g.,][]{Rame16} and firms' incentive to hold inventories
    \citep[see, e.g.,][]{frankel1986expectations, frankel2008,
frankel2014effects}, increasing commodity prices. An increase in
interest rates works in the opposite direction.} 

These challenges demand a structural model built on firm theoretical foundations
that can isolate the direct effect of interest rates on commodity prices through
each of the channels listed above. The obvious candidate to provide the
necessary structure is the competitive storage model developed by
\citet{samuelson1971stochastic}, \citet{newbery1982optimal},
\citet{wright1982economic}, \citet{scheinkman1983simple}, \citet{deaton1992on,
deaton1996competitive}, \citet{chambers1996theory} and \citet{Cafi15a}, among others. In this
model, commodities are assets with intrinsic value, separate from
future cash flows. The standard version of the model features time-varying
production, storage by forward-looking investors, arbitrage constraints, and
non-negative carryover. Within the constraints of the model, there is a clear
relationship between interest rates, storage, and commodity prices.
\citet{Fama87} show that the relationship between the basis, the spread between
the futures and spot prices, and interest rates is consistent with the structure
of this model.

The main obstacle to applying the standard competitive storage model to the
problem at hand is that the discount rate is constant. The
source of this shortcoming is technical: a constant positive interest rate is
central to the traditional proof of the existence and uniqueness of equilibrium
prices and the study of their properties \citep[see, e.g.,][]{deaton1992on,
deaton1996competitive}. In particular, positive constant rates are used to
obtain contraction mappings over a space of candidate price functions, with the
discount factor being the modulus of contraction.  

\begin{figure}[htb!]
	\centering
	\scalebox{0.7}{\includegraphics{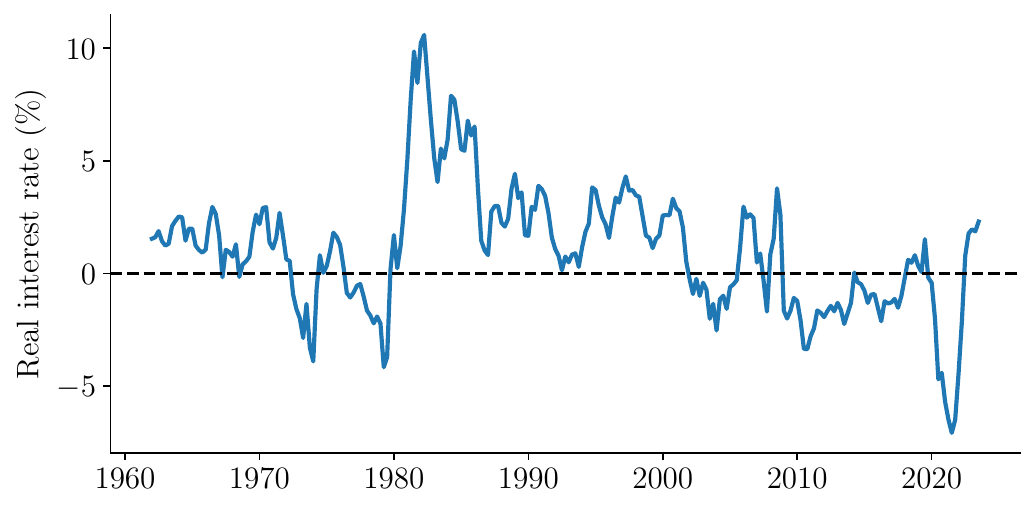}}
	\caption{\label{fig:real_int}The real interest rate over the long run
          (one-year US treasury yield deflated by a measure of expected
          inflation obtained from an autoregressive model estimated on a 30-year
          rolling window). Source: FRED.}
\end{figure}

At the same time, relaxing the assumption of constant discounting is necessary
for analysis of interactions between interest rates and commodity prices.
Without this modification, it is not possible to study how the nature and timing
of shocks to supply, demand, and interest rates affect the sign and
magnitude of changes in commodity prices. Moreover, allowing for state-dependent
discounting brings the model closer to the data, since real interest rates 
exhibit large movements over time, as shown in
\fref{fig:real_int}.

One difficulty with relaxing the assumption of constant interest rates in the
competitive storage model is that negative real interest rates cannot be
ignored, as shown in \fref{fig:real_int}. If interest rates can be
sufficiently negative for sufficiently long periods, then the model will have no
finite equilibrium, due to unbounded demand for inventories. Thus, developing a
model that can handle realistic calibrations requires accommodating negative
yields on risk-free bonds in some states of the world, while providing
conditions on these states and the size of the yields such that the model
retains a well-defined and unique solution. 

In this paper, we extend the competitive storage model to include
state-dependent discounting and establish conditions under which a unique
equilibrium price process exists. These conditions allow for both positive and
negative discount rates, while also providing a link between the asymptotic
return on risk-free assets with long maturity, the depreciation rate of the
commodity in question, and the existence and uniqueness of solutions. Under
these conditions, we show that the equilibrium solution can be computed via a
globally convergent algorithm and characterize the continuity
and monotonicity properties of the equilibrium objects. We also develop an
endogenous grid algorithm for computing equilibrium objects efficiently.

With these results in hand, we examine the effect of interest rates on commodity
prices from a theoretical and quantitative perspective. We show that, in some
settings, interest rates and commodity prices can be positively correlated, such
as when shocks shift up both interest rates and aggregate demand. Nevertheless,
we are able to identify relatively sharp conditions under which a negative
correlation is realized and analyze the impact of interest rates on commodity
prices in depth. These conditions require that the exogenous state follows a
monotone Markov process that is independent across dimensions and has a
non-negative effect on interest rate and commodity output. The independence
restriction cannot be dropped: if different exogenous states are
contemporaneously correlated, then the relationship between interest rates and
commodity prices can be reversed.

On the quantitative side, we study the impulse response functions (IRFs) of
commodity price, inventory, and price volatility in response to an interest rate
shock for suitable structural parameters. We use these IRFs to explore the
speculative and global demand channels. The former examines the role of
speculators in the physical market, whose incentives to hold inventories are
affected by fluctuations in interest rates. The latter studies the impact of
exogenous interest rate shocks on commodity demand through their effects on
economic activity. While both channels have been frequently suggested in the
literature, their analysis remains significantly under-explored to date. To
handle the nonlinearities of the competitive storage model, we follow the
methodology of \citet{Koop96} and treat IRFs as state-and-history-dependent
random variables.

The simulated IRFs show that prices fall immediately after a positive interest
rate shock and slowly converge to their long-run value, with a more pronounced
decline and a slower convergence when the demand channel is active. The behavior
of inventory dynamics is nuanced. While inventories typically decrease
post-shock due to higher cost of carry, they may later rise when the demand
channel is active. This behavior is caused by the lower spot prices resulting
from reduced demand, which creates profit opportunities through
storage. Overall, inventories tend to return to their long-run average more
slowly than prices. Furthermore, price volatility exhibits sensitivity to
inventory dynamics: a larger response in inventory tends to generate an
inversely larger response in price volatility. Finally, the magnitude and
overall pattern of the IRFs depend substantially on the market supply and
interest rate regimes.\footnote{These findings suggest that postulating a
  uniform effect of monetary shocks across different market supply and interest
  rate scenarios could bias empirical analysis.}

Our work has some implications concerning commodity futures. For example,
following \citeauthor{Gard76}'s (\citeyear{Gard76}) methodology, it has become
common to use futures prices as a proxy for expected prices \citep[see,
e.g.,][]{Goue22a}. However, our findings indicate that this substitution is
invalid within the framework outlined in this study. As elucidated by
\citet{Cox81}, forward and futures prices diverge in the presence of stochastic
interest rates if interest rates and commodity prices are correlated, which is
the case here---particularly under the global demand channel.

Regarding existing literature on interest rates and commodity prices, Jeffrey
Frankel has made numerous empirical and theoretical contributions to this topic,
focusing on how commodity prices overshoot their long-run target after a shock
due to their inherent price flexibility \citep{Fran85, frankel1986expectations,
  frankel2008, frankel2014effects}. This literature tends to find a negative
effect of interest rate increases on commodity prices in the short
\citep{Rosa14, scrimgeour2015commodity} and medium run \citep{Anzu13,
  harvey2017long}.\footnote{An exception is \citet{Kili22}, who find no effect
  of real interest rate movements on oil prices.} The macroeconomic studies of
\citet{christiano1999monetary} and \citet{bernanke2005measuring} have also found
a negative relationship between interest rates and commodity prices. Moreover,
interest rates affect not only the level of commodity prices, but also their
cross-correlation and volatility \citep{gruber2018interest}. Compared to
these studies, the methodology developed here allows for a more systematic
analysis of transmission mechanisms.

This work also intersects with other studies that examine the theoretical
relationship between interest rate fluctuations and commodity prices, including
\citet{arseneau2013commodity}, \citet{Basa16}, \citet{Tume16}, and
\citet{Boda21}. However, \citet{Basa16} omit the nonnegativity constraint on
storage, while \citet{Tume16} and \citet{Boda21} overlook the nonlinearity of
storage and impact of large shocks by approximating their model around a steady
state with positive stocks. Although \citet{arseneau2013commodity} admit
nonlinearity, they adopt a limited stochastic structure where commodity
production shocks are the only source of uncertainty. In contrast to these
studies, we establish a comprehensive theory that avoids these simplifications
and, in addition, allows us to jointly handle realistic depreciation rates and
interest rate processes, which are crucial for accurately representing commodity
price dynamics.\footnote{Our framework admits occasionally negative interest
  rates, in line with real-world outcomes. While some earlier models also allow
  negative interest rates, the relevant parameterizations are not empirically
  plausible, since, in these models, negative interest rates must be offset by
  excessively large depreciation rates in order to obtain an equilibrium.}
Moreover, we furnish a general theory on the existence and uniqueness of
equilibrium solutions in this framework.\footnote{\label{fn:finance}While stochastic interest rates are
  relatively novel in the storage model framework, work in the finance
  literature has shown that commodity pricing models benefit from incorporating
  stochastic convenience yields or stochastic interest rates \citep[see,
  e.g.,][]{Gibs90, Schw97, Casa05a}. This literature focuses on questions
  orthogonal to our interests, such as term structure. One bridge between the
  approaches is the study by \citet{Rout00}, which suggests that within a
  storage framework akin to the one we are exploring, convenience yields
  naturally arise from the interaction of supply, demand, and storage dynamics.
  We also note recent studies that examine commodity financialization and
  speculation through financial derivatives, including works by \citet{Basa16},
  \citet{baker2021financialization}, and
  \citet{goldstein2022commodity}. However, we do not pursue these topics in our
  paper.} We are able to account for the nonlinearity of storage and a rich
stochastic structure because, contrary to these
previous studies, our model is not a general equilibrium model and interest
rates are exogenous.

On the empirical side, an extensive literature has analyzed the empirical
validity of the storage model \citep[e.g.,][]{deaton1996competitive,
  cafiero2011the, Cafi15a, Goue22a}. This literature focuses on idiosyncratic
shocks and neglects shocks to storage costs. In contrast, we study the role of
aggregate shocks on storage costs, providing a theoretical analysis of
conditions under which interest rates negatively affect commodity prices
and a quantitative analysis on the impact of interest rate shocks through
speculative and global demand channels.

From a technical perspective, our work overlaps with recent work on household
and consumption problems with state-dependent discounting. For example,
\citet{ma2020income} use Euler equation methods to obtain existence and
uniqueness of solutions to an optimal savings problem in a setting where the
subjective discount rate is state-dependent. \citet{stachurski2021dynamic} and
\citet{toda2021perov} study similar problems.  Like those papers, we tie
stochastic discounting to long-run ``eventual'' contraction methods. Unlike
those papers, we apply eventual contraction methods to commodity pricing
problems.

The rest of the paper is organized as follows. \Fref{sec:or} formulates a
rational expectations competitive storage model with time-varying discounting
and discusses the existence, uniqueness, and computability of the equilibrium
solutions. Sections~\ref{sec:itheory} and~\ref{sec:quant} examine the role of
interest rates on commodity prices from a theoretical and quantitative
perspective, respectively. \Fref{sec:concl} concludes. Proofs, descriptions of
algorithms, and counterexamples can be found in the appendices.

\section{Equilibrium Prices}\label{sec:or}

This section formulates the competitive storage model with time-varying
discounting and discusses conditions under which existence and uniqueness of the
equilibrium pricing rule hold. 

\subsection{The Model}\label{sec:model}

Let $I_t \geq 0$ be the inventory of a given commodity at time $t$, and let
$\delta \geq 0$ be the instantaneous rate of stock deterioration. The cost of
storing $I_t$ units of goods from time $t$ to time $t+1$, paid at time $t$, is
$k I_t$, where $k\geq 0$. Let $Y_t$ be the output of the commodity. Let $X_t$ be
the total available supply at time $t$, which takes values in $\XX \coloneq
[b,\infty)$, where $b \in \RR$, and is defined by
\begin{equation}
  \label{eq:Xdef}
  X_t \coloneq \me^{-\delta} I_{t-1} + Y_t.
\end{equation}
Let $p \colon \XX \to \RR$ be the inverse demand function. We assume that $p$ is
continuous, strictly decreasing, and bounded above.\footnote{We impose an upper
bound to simplify exposition. In Appendix~\ref{sec:proof-sect-refs}
we show that unbounded demand functions can also be treated, and
the theory below still holds.}  Let $P_t$ be the market price at time $t$.
Without inventory, $P_t = p(Y_t)$. In general, market equilibrium requires
that total supply equals total demand (sum of the consumption and the
speculation demand), equivalently,
\begin{equation}\label{eq:equilib}
  X_t = p^{-1}(P_t) + I_t.
\end{equation}
An immediate implication of~\eqref{eq:equilib} is that $P_t \leq p(b)$ and 
\begin{equation}\label{eq:equilibimp}
  P_t \geq p(X_t),
  \;\;
  \text{with equality holding when $I_t = 0$}.
\end{equation}
Let $M_{t+1}$ be the real one-period stochastic discount factor applied by
investors at time $t$. The price process $\{P_t\}$ is restricted by
\begin{equation}\label{eq:arbcon}
  P_t \geq \me^{-\delta} \EE_t M_{t+1} P_{t+1} - k, 
  \;\;
  \text{with equality holding if $I_t > 0$  and $P_t < p(b)$}.
\end{equation}
In other words, per-unit expected discounted returns from storing the commodity
over one period cannot exceed the per-unit cost of taking that position.

Combining~\eqref{eq:equilibimp} and~\eqref{eq:arbcon} yields\footnote{The
minimization over $p(b)$ in~\eqref{eq:foc} is required due to the generic
stochastic discounting setup.  As can be seen below, our theory allows for large
and highly persistent discounting process (e.g., arbitrarily long sequences of
negative low interest rates under risk neutrality), in which case
$\me^{-\delta} \EE_t M_{t+1} > 1$ with positive probability, thus the
marginal reward of speculation, $\me^{-\delta} \EE_t M_{t+1} P_{t+1} - k$, can 
be larger than $p(b)$. The extra minimization operation is then required to
meet the equilibrium condition $P_t \leq p(b)$.

If $\me^{-\delta} \EE_t M_{t+1} P_{t+1} - k > p(b)$, the equilibrium condition 
implies that $P_t \equiv p(b) < \me^{-\delta} \EE_{t}M_{t+1}P_{t+1} - k$. In 
essence, price reaches its upper bound, and the non-arbitrage condition is 
violated. As a consequence, investors are incentivized to maintain inventories 
to exploit arbitrage opportunities. If this pattern persists, investments in 
inventory could grow arbitrarily large.\label{fn:nec_min}}
\begin{equation}
  \label{eq:foc}
  P_t = \min \left\{
    \max \left\{ 
      \me^{-\delta} \EE_t M_{t+1} P_{t+1} - k, \, p (X_t)
    \right\}, \, p(b)
  \right\}.
\end{equation}
Both $\{M_t\}$ and $\{Y_t\}$ are exogenous, obeying
\begin{equation}\label{eq:qy_func}
  M_t = m(Z_t, \epsilon_t) \quad \text{and} \quad  Y_t = y(Z_t, \eta_t),
\end{equation}
where 
$m$ and $y$ are Borel measurable functions satisfying $m \geq 0$ and $y \geq b$,
$\{Z_t\}$ is a time-homogeneous irreducible Markov chain (possibly
    multi-dimensional) taking values in a finite set $\ZZ$, and
the innovations $\{\epsilon_t\}$ and $\{\eta_t\}$ are {\sc iid} and
    mutually independent.

\begin{example}\label{ex:MYsetup}
  The setup in~\eqref{eq:qy_func} is very general and allows us to model both
  correlated and uncorrelated $\{M_t,Y_t\}$ processes. In particular, it does
  not impose that $\{M_t\}$ and $\{Y_t\}$ are driven by a \textit{common} Markov
  process, nor does it restrict that they are mutually dependent. Consider for
  example $Z_t = (Z_{1t},Z_{2t})$, where $\{Z_{1t}\}$ and $\{Z_{2t}\}$ are
  mutually independent, possibly multi-dimensional Markov processes,
  and $M_t = m(Z_{1t}, \epsilon_t)$ and $Y_t = y(Z_{2t}, \eta_t)$.
  In this case, $\{M_t\}$ and $\{Y_t\}$ are mutually independent, although they
  are autocorrelated. If in addition $\{Z_{1t}\}$ (resp., $\{Z_{2t}\}$) is {\sc
  iid} or does not exist, then $\{M_t\}$ (resp., $\{Y_t\}$) is {\sc iid}.
  Obviously, these are all special cases of~\eqref{eq:qy_func}. More examples
  are given in Section~\ref{sec:itheory} below.
\end{example}

Below, the next-period value of a random variable $X$ is denoted by $\hat X$. In
addition, we define $\EE_{z} \coloneq \EE \,(\, \cdot \mid Z=z)$ and assume throughout
that
\begin{equation}\label{eq:ezp}
  \me^{-\delta} \EE_z \hat M p(\hat Y) - k > 0
  \;\text{ for all }\; z \in \ZZ.
\end{equation}
In other words, the present market value of future output
covers the cost of storage.

\subsection{Discounting}\label{ss:discount}

To discuss conditions under which price equilibria exist, we need to jointly
restrict discounting and depreciation. To this end, we introduce the
quantity\footnote{Here and below, expectation without a subscript refers to the
stationary process, where $Z_0$ follows the (necessarily unique) stationary
distribution.}
\begin{equation}\label{eq:yy}
  \kappa(M) \coloneq  \lim_{n \to \infty} \frac{- \ln q_n}{n} 
  \quad \text{ where } 
  \quad q_n \coloneq \EE \prod_{t=1}^{n} M_t .
\end{equation}
To interpret $\kappa(M)$, note that, in this economy, $q_n(z) \coloneq \EE_z
\prod_{t=1}^{n} M_t$ is the state $z$ price of a strip bond with maturity $n$.
Since $\{Z_t\}$ is irreducible, initial conditions do not determine long-run
outcomes, so $q_n(z)$ is approximately constant at $q_n$ defined
in~\eqref{eq:yy} when $n$ is large. As a result, we can interpret $\kappa(M)$ as
the asymptotic yield on risk-free zero-coupon bonds as maturity increases
without limit.

In Lemma~\ref{lm:yy} of Appendix~\ref{sec:proof-sect-refs}, we provide a
numerical method for calculating $\kappa(M)$ by connecting it to the spectral
radius of a discount operator.

\begin{assumption}\label{a:opt}
  $\kappa(M) + \delta > 0$.
\end{assumption}

Assumption~\ref{a:opt} is analogous to the classical condition $r+\delta > 0$
found in constant interest rate environment of \citet{deaton1996competitive} and
many other studies.\footnote{In the model with constant risk-free rate $r$, the
discount rate $M_t$ is $1/(1+r)$ at each $t$, so, by the definition
in~\eqref{eq:yy}, we have $\kappa(M) = \lim_{n \to \infty} {n \ln (1+r)}/{n} =
\ln (1+r) \approx r$.}  In the more general setting we consider,
Assumption~\ref{a:opt} ensures sufficient discounting, adjusted by the
depreciation rate, to generate finite prices in the forward-looking
recursion~\eqref{eq:foc}, while still allowing for arbitrarily long sequences of
negative yields in realized time series.

\subsection{Equilibrium}\label{sec:equilibrium}

We take $(X_t, Z_t)$ as the state vector, taking values in $\SS \coloneq \XX \times
\ZZ$. We assume free disposal as in \citet{Cafi15a} to ensure that the equilibrium
prices are non-negative. Conjecturing that a stationary rational
expectations equilibrium exists and satisfies~\eqref{eq:foc}, an
\textit{equilibrium pricing rule} is defined as a function $f^* \colon \SS \to
\RR_+$ satisfying
\begin{equation*}
  f^*(X_t, Z_t) = \min \left\{
    \max \left\{
      \me^{-\delta} \EE_t M_{t+1} f^*(X_{t+1}, Z_{t+1}) - k, 
      \, p(X_t)
    \right\}, \, p(b)
  \right\}
\end{equation*}
with probability one for all $t$, where $X_{t+1}$ is defined by~\eqref{eq:Xdef}
and, recognizing free disposal, storage therein is determined by $I_t =
i^*(X_t,Z_t)$, where $i^*: \SS \to \RR_+$ is the \textit{equilibrium storage
rule}\footnote{Throughout, we adopt the usual convention that $\inf
\emptyset = \infty$.}
\begin{equation}\label{eq:opt_stor}
  i^*(x,z) \coloneq
  \begin{cases}
    x - p^{-1} [f^*(x,z)], & \text{if }\, x < x^*(z)\\
    x^*(z) - p^{-1}(0), & \text{if }\, x \geq x^*(z)
  \end{cases}
\end{equation}
with
\begin{equation*}
  x^*(z) \coloneq \inf \left\{ x \in \XX: f^*(x,z) = 0 \right\}. 
\end{equation*}

Let $\cC$ be the space of bounded, continuous, and non-negative functions $f$ on
$\SS$ such that $f(x,z)$ is decreasing in $x$, and $f(x,z) \geq p(x)$ for all
$(x,z)$ in $\SS$. Given an equilibrium pricing rule $f^*$, let
\begin{equation*}
  \bar p(z) \coloneq \min \left\{ 
    \me^{-\delta} \EE_z \hat M f^*(\hat Y, \hat Z) - k, \, p(b) 
  \right\}.
\end{equation*}
The next theorem provides conditions under which the equilibrium pricing rule
exists, is uniquely defined, and gives a sharp characterization of its
analytical properties. 

\begin{theorem}[\textbf{Existence and Uniqueness of Equilibrium Price}]\label{t:opt}
  If Assumption~\ref{a:opt} holds, then a unique equilibrium
  pricing rule $f^*$ exists in $\cC$. Furthermore,
  \begin{enumerate}
  \item $f^*(x,z) = p(x)$ if and only if 
    $x \leq p^{-1} [\bar p(z)]$,
    
  \item $f^*(x,z) > \max \{ p(x), 0\}$ if and only if 
    $p^{-1}[\bar p(z)] < x < x^*(z)$,
    
  \item $f^*(x,z) = 0$ if and only if $x \geq x^*(z)$, and
    
  \item $f^*(x,z)$ is strictly decreasing in $x$ when 
    it is strictly positive and $\me^{-\delta} \EE_z \hat M < 1$.
  \end{enumerate}
\end{theorem}

In Appendix~\ref{sec:proof-sect-refs}, we show that the equilibrium pricing rule
is the unique fixed point of an operator defined by the equilibrium conditions
(named as the equilibrium price operator) and can be solved for via successive
approximation. In particular, the equilibrium price operator is an
\textit{eventual contraction mapping} on a suitably constructed candidate space
(which reduces to $\cC$ when the demand function is bounded). This guarantees
existence, uniqueness, and computability of the equilibrium
solutions.\footnote{Since we allow for arbitrarily long sequences of negative
  yields, it is challenging to construct operators that contract in one step,
  since the one-period yield, which is typically required to construct the
  modulus of contraction, can be overly small, violating
  \citeauthor{blackwell1965discounted}'s (\citeyear{blackwell1965discounted})
  sufficient conditions for contraction. To solve this problem,
  Assumption~\ref{a:opt} bounds the asymptotic yield instead of the one-period
  yield, allowing us to construct an $n$-step contraction.\label{fn:nstep}} In
Appendix~\ref{s:alg}, we provide an endogenous grid algorithm that solves
for the equilibrium objects efficiently. In Appendix~\ref{s:necess}, we
show that, in some rather standard settings, Assumption~\ref{a:opt} is necessary
as well as sufficient: no equilibrium price sequence exists if
Assumption~\ref{a:opt} fails.\footnote{\label{fn:5}See Proposition~\ref{pr:modfesol} in
  Appendix~\ref{s:necess}. Note that Theorem~\ref{t:opt} establishes
  existence and uniqueness of stationary equilibria. However, this does not
  preclude the possibility of nonstationary equilibria, including nonstationary
  price sequences and bubbles. For some recent literature on bubbles in asset
  markets, see \citet{barlevy2012rethinking}, \citet{guerron2023bubbles},
  \citet{plantin2023asset}, and \citet{hirano2024bubble}.\label{fn:bubble}}

\begin{figure}
	\centering
	\includegraphics[width=.75\linewidth]{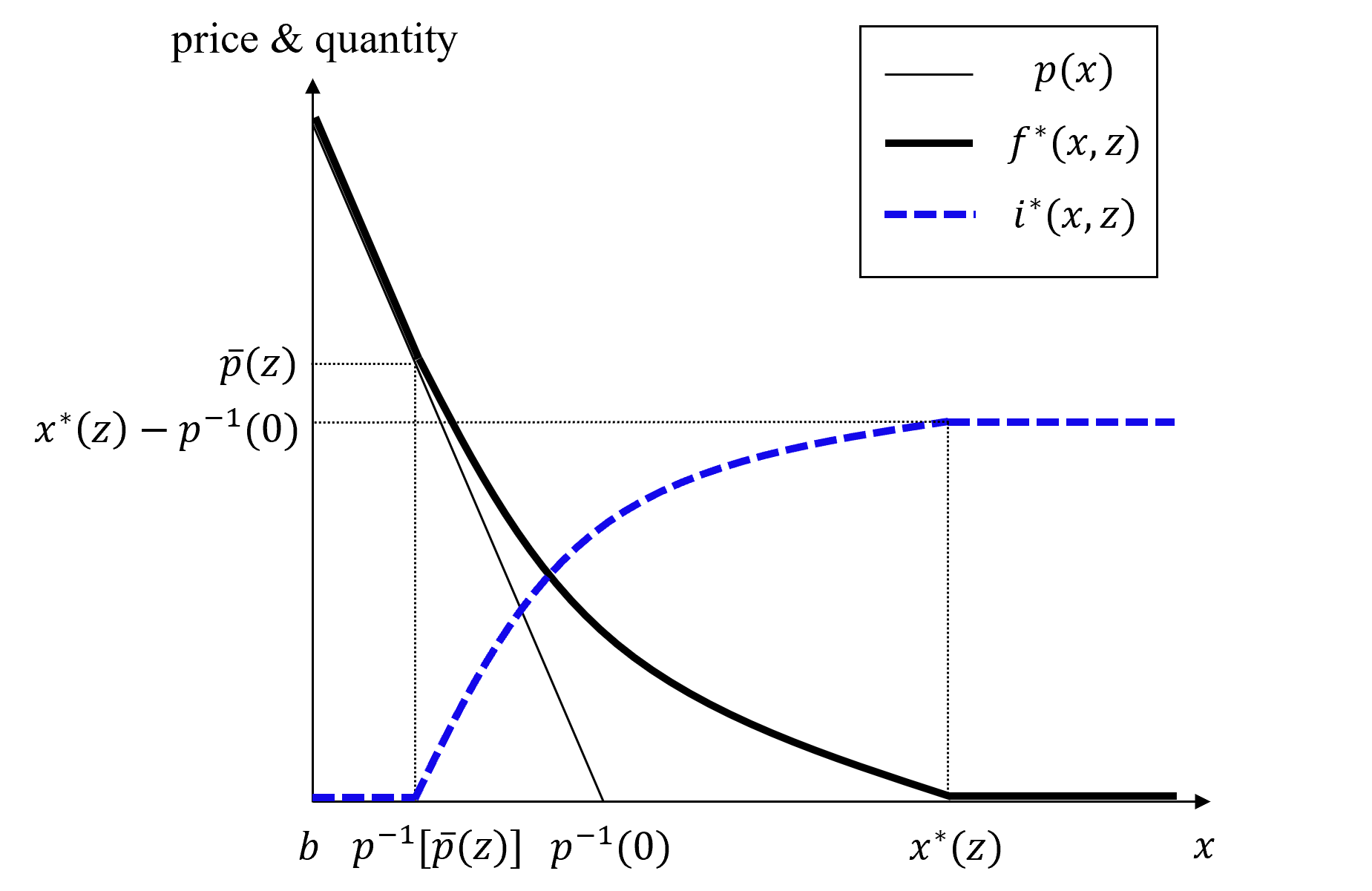}
	\caption{\label{fig:equi_rule}Illustration of the equilibrium price
		function $f^*$ and the equilibrium storage rule $i^*$. $p$ represents
		the inverse demand function, $\bar{p}$ denotes the price threshold at
		which speculators begin holding inventories, and $x^*$ is the
		threshold for free disposal.}
\end{figure}

The next result states the properties of the equilibrium storage rule.

\begin{proposition}[\textbf{Existence and Uniqueness of Equilibrium Storage}]\label{pr:storage}
  If Assumption~\ref{a:opt} holds, then the equilibrium storage rule $i^*(x,z)$
  is increasing in $x$ and continuous. Furthermore,
  \begin{enumerate}
  \item $i^*(x,z) =0$ if and only if $x \leq p^{-1}[\bar p(z)]$,
    
  \item $0< i^*(x,z) <x^*(z) - p^{-1}(0)$ if and only if 
    $p^{-1}[\bar p(z)] < x < x^*(z)$,
    
  \item $i^*(x,z) = x^*(z) - p^{-1}(0)$ if and only if $x \geq x^*(z)$, and
    
  \item $i^*(x,z)$ is strictly increasing in $x$ if 
    $p^{-1}[\bar p(z)] < x < x^*(z)$ and 
    $\me^{-\delta} \EE_z \hat M < 1$. 
  \end{enumerate}
\end{proposition}

\Fref{pr:storage} indicates that speculators hold inventories if and only if the
market value of the total available supply $p(x)$ is below the decision
threshold $\bar p(z)$. Otherwise, selling all commodities at hand is 
optimal, in which case the equilibrium price is $f^*(x,z)=p(x)$. The
equilibrium price and storage properties are illustrated in \fref{fig:equi_rule} under 
a linear demand function. The equilibrium rules are sketched
for a given exogenous state $z$.

\section{Interest Rates and Prices: Theoretical Results}\label{sec:itheory}

Next we inspect the relationship between interest rates and commodity
prices implied by the model. To this end, we assume that speculators
discount future payoffs according to market prices: 
\begin{equation*}
  M_t = \frac{1}{R_t}, \quad \text{where} \quad R_t \coloneq r(Z_t, \epsilon_t).
\end{equation*}
Hence $r$ is a real-valued non-negative Borel measurable function of the state
process and innovation $\epsilon$. The process $\{R_t\}$ is interpreted as 
the gross real interest rate on risk-free
bonds. We therefore preserve the risk-neutrality assumption of the standard
competitive storage model, while allowing the risk-free rate to be
state-dependent. Throughout this section, we impose the assumptions of
Section~\ref{sec:or}.

\subsection{Correlations}\label{sec:correlations}

We first explore general conditions under which interest rates and commodity
prices are negatively correlated.
As a first step, we state a finding concerning
monotonicity of equilibrium objects with respect to the exogenous states.

\begin{proposition}[\textbf{Monotonicity of Equilibrium Objects w.r.t.\ the Exogenous State}]\label{pr:mono_z}%
  If $r(z, \epsilon)$ and $y(z,\eta)$ are nondecreasing in $z$, and $\{Z_t\}$ is
  a monotone Markov process,\footnote{Monotonicity is defined in terms of
  first-order stochastic dominance (see Appendix~\ref{sec:pf_itheory}).}  then the equilibrium pricing rule $f^*(x,z)$, the
 equilibrium inventory $i^*(x,z)$, and the decision threshold $\bar p(z)$ are all
  decreasing in $z$.
\end{proposition}

The intuition is as follows: If (i) a higher $Z_t$ shifts up the distribution of
$Z_{t+1}$ in terms of first-order stochastic dominance and (ii) interest rates
and output are both nondecreasing in this state variable, then a high $Z_t$
today tends to generate both sustained high output and more impatient
speculators in the future. The former boosts supply, while the latter
diminishes the incentive for holding inventories, reducing speculative
demand. As a result, both inventories and prices are lower. 

The assumptions of Proposition~\ref{pr:mono_z} do
not restrict $R_t$ and $Y_t$ to be strictly increasing in
$Z_t$, nor do they impose that $R_t$ and $Y_t$ are driven by a common factor. In
particular, the second assumption concerning monotone Markov process is standard
(see Appendix~\ref{sec:pf_itheory} for sufficient conditions). 
Below, we discuss the first assumption through examples.

\begin{example}
    If $\{R_t\}$ and $\{Y_t\}$ are {\sc iid} and mutually independent, then we
    can set $Z_t \equiv 0$, $\epsilon_t = R_t$ and $\eta_t = Y_t$, in which case
    $r(z,\epsilon) = \epsilon$ and $y(z,\eta) = \eta$. Hence, the first two
    assumptions of Proposition~\ref{pr:mono_z} hold.
\end{example}

\begin{example}\label{ex:RY_indep}
    If $\{R_t\}$ and $\{Y_t\}$ are autocorrelated and mutually independent, then
    we can write $Z_t$ as $Z_t = (Z_{1t}, Z_{2t})$, where $\{Z_{1t}\}$ and
    $\{Z_{2t}\}$ are mutually independent, possibly multi-dimensional Markov
    chains, and
    	$R_t = r(Z_{1t}, \epsilon_t)$ and $Y_t = y(Z_{2t}, \eta_t)$.
    In this case, the first assumption of Proposition~\ref{pr:mono_z} holds as
    long as $r$ is nondecreasing in $Z_{1t}$ and $y$ is nondecreasing in
    $Z_{2t}$. 
\end{example}

\begin{example}\label{ex:RY_MC}
    If $\{R_t\}$ and $\{Y_t\}$ are finite Markov processes, then we can set
    $\epsilon_t = \eta_t \equiv 0$ and define $Z_t = (R_t, Y_t)$,  in which
    case the first assumption of Proposition~\ref{pr:mono_z} holds
    automatically, while the second assumption holds as long as $\{R_t\}$ and
    $\{Y_t\}$ are monotone and non-negatively correlated Markov processes.
\end{example}

We can now state our main result concerning correlation. In doing so,
we suppose that $Z_t = (Z_{1t}, \ldots, Z_{nt})$
takes values in $\RR^n$.

\begin{proposition}[\textbf{Negative Correlation of Interest Rates and Prices}]\label{pr:neg_corr}	
  If the conditions of Proposition~\ref{pr:mono_z} hold and $\{Z_{1t}, \ldots,
  Z_{nt}\}$ are independent for each fixed $t$, then
  \begin{equation*}
    \Cov_{t-1} (P_t, R_t) \leq 0 \quad \text{for all } \, t \in \NN.
  \end{equation*}
\end{proposition}

As Proposition~\ref{pr:mono_z} illustrates, when interest rates and output are
both positively affected by the monotone exogenous state process, commodity
prices will be negatively affected by the exogenous state. Therefore, there is a
trend of comovement (in opposite directions) between commodity price and
interest rate, resulting in a negative correlation. The proof of
Proposition~\ref{pr:neg_corr} relies on the Fortuin--Kasteleyn--Ginibre
inequality.

Note that the independence-across-dimensions condition in
Proposition~\ref{pr:neg_corr} cannot be omitted. In
Appendix~\ref{sec:counter_ex}, we provide examples showing that if $\{Z_{1t},
\ldots, Z_{nt}\}$ are positively or negatively correlated for some $t \in \NN$,
then interest rates and prices can be positively correlated. This is because
contemporaneous correlation across dimensions of $Z_t$ can alter comovement of
interest rates and commodity prices. (Such correlation can either strengthen or
weaken the impact of interest rates on commodity prices, yielding rich model
dynamics.)

\begin{remark}
    In Appendix~\ref{sec:pf_itheory}, we show that Proposition~\ref{pr:neg_corr}
    can be extended to the setting of Section~\ref{sec:or}, where agents
    are not necessarily risk neutral. In particular, $\Cov_{t-1} (P_t, M_t)
    \geq 0$ holds.
\end{remark}

\begin{example}\textbf{(The Speculative Channel).}
  In applications, $\{R_t\}$ typically follows a Markov process, while
  $\{Y_t\}$ represents a sequence of supply shocks (e.g., harvest failures,
  conflicts around oil production sites, so on), which is {\sc iid} and less
  likely to be affected by the monetary conditions \citep[see,
  e.g.,][]{deaton1992on, Cafi15a}. Hence, $\{R_t\}$ and $\{Y_t\}$ are mutually
  independent. In this case, all the effects of interest rates on commodity
  prices transit through commodity speculation. By letting
  	  $Z_t = R_t$, $\epsilon_t \equiv 0$ and 
  	  $\eta_t = Y_t$,
  we have $r(z,\epsilon) = z$ and $y(z,\eta) = \eta$. Hence, all the assumptions
  of Proposition~\ref{pr:neg_corr} hold as long as $\{R_t\}$ is a finite
  monotone Markov process (e.g., a discrete version of a positively correlated
  AR(1) process) and Assumption~\ref{a:opt} holds (see the next section). In
  this case, Proposition~\ref{pr:neg_corr} implies that interest rates are negatively correlated
  with commodity prices, which matches the empirical results of
  \citet{frankel1986expectations, frankel2008, frankel2014effects}. 
\end{example}

\begin{example}\textbf{(The Global Demand Channel).}\label{ex:demand_shock}
  Since the output of the commodity, $Y_t$, enters linearly in total
  availability, it can be redefined as
  $Y_t = Y_t^S - Y_t^D$, where $Y_t^S$ is the supply shock and $Y_t^D$ is the
  demand shock. Hence $Y_t$ can be interpreted as a net supply shock. There is
  widespread evidence that both types of shocks matter in commodity markets,
  albeit with relative importance depending on the commodities \citep[see,
  e.g.,][]{kilian2009not, Goue22a}. Unlike supply shocks, demand shocks are
  likely to be affected by monetary policies. Since interest rates affect global
  demand \citep{Rame16}, an interest rate shock leads to an aggregate demand
  shock that affects all commodities. If interest
  rates follow a Markov process, then $Z_t = (R_t, Z_{2t})$ and
  $Y_t = y(Z_{2t}, \eta_t)$, where $\{Z_{2t}\}$ is a Markov process that is
  correlated with $\{R_t\}$.  Hence, $Z_t$ is contemporaneously correlated, and
  the independence-across-dimensions condition of Proposition~\ref{pr:neg_corr}
  fails. However, the theory of Section~\ref{sec:or} still applies.
\end{example}

\subsection{Causality}\label{sec:causality}

We now study the causal relationship between interest rates and commodity 
prices. As a first step, we state an elementary monotonicity property 
concerning interest rates and prices.
To this end, we take $\{R_t^i\}$ to be the interest rate
process for economy $i\in \{1,2\}$. In addition, let $f_i^*$ and
$\{P_t^i\}$ be the equilibrium pricing rule and the price process
corresponding to $\{R_t^i\}$. 

\begin{proposition}[\textbf{Causal Effect of Ordered Interest Rates on Prices}]\label{pr:mono_seq}
	If $R_t^2 \leq R_t^1$ with probability one for all $t \geq 0$, then $f_1^*
	\leq f_2^*$ and $P_t^1 \leq P_t^2$ with probability one for all
	$t$.
\end{proposition}

The intuition is straightforward. Seen from the speculative channel, lower
interest rates reduce the opportunity cost of storage. Lower storage costs
encourage a build-up of inventories. Higher demand for inventories induces 
higher prices.

Proposition~\ref{pr:mono_seq} has limited implications because it concerns
variations in interest rates that are uniformly ordered over time. Next, we 
aim to relax this assumption.

Let $\{X_t^i\}$ and $\{Z_t^i\}$ be respectively the endogenous and exogenous 
state processes for economy $i\in \{1,2\}$. Unless otherwise specified, we 
assume that both economies experience the same innovation process $\{\eta_t\}$ 
to output.

\begin{proposition}[\textbf{Causal Effect of Interest Rates on Prices}]\label{pr:caus_1}
	Suppose $\{R_t\}$ is a monotone finite-state Markov process and 
	$Y_t = y(R_t, \eta_t)$, where $y$ is nondecreasing in $R$. If 
	$X_{t-1}^2 \leq X_{t-1}^1$, $R_{t-1}^1 \leq R_{t-1}^2$, and $R_t^2 \leq 
	R_t^1$ with probability one, then $P_t^1 \leq P_t^2$ with probability 
	one.
\end{proposition}

Proposition~\ref{pr:caus_1} indicates that if the interest rate is a monotone 
finite Markov process that has a nonnegative effect on output, then, conditional 
on the same previous state, a higher interest rate today reduces commodity 
price in the same period.

\begin{proposition}[\textbf{Causal Effect of Interest Rates on Prices Over Time}]\label{pr:caus_2}
	Suppose $\{R_t\}$ is a monotone finite-state Markov process and $Y_t = 
	y(R_t)$, where $y$ is nondecreasing. If $X_{t-1} \leq X_t$, $R_{t-1} \geq 
	R_t$, and $R_{t} \leq R_{t+1}$ with probability one, 
	then $P_t \geq P_{t+1}$ with probability one.
\end{proposition}

Proposition~\ref{pr:caus_2} above indicates that, if $R_t$ is a monotone finite 
Markov process and has a nonnegative effect on $Y_t$, $X_t$ is no less than its 
previous period level, and $R_t$ is no higher than its previous period level, 
then an increase in interest rates next period causes falling commodity prices. 

\section{Quantitative Analysis}\label{sec:quant}

To illustrate the quantitative implications of our theory, we study the impact
of interest rates on commodity prices through the speculative and
the aggregate demand channels. We use a stylized model that requires a minimum
number of parameters to explore how these channels
operate. We calibrate the model to a quarterly
setting to limit the number of state variables.\footnote{A monthly real interest
  rate process requires a rich autoregressive structure, introducing many lags.}

The main takeaways from this section are fourfold. First, IRFs show that prices
decrease immediately following a positive interest rate shock and slowly
converge to their long-run average, with a stronger decrease and a slower
convergence when the global demand channel is active. Second, inventory dynamics
are nuanced: while inventories typically fall after an interest rate shock due
to higher holding costs, they may rise again when the demand channel is
active. This is because lower demand reduces spot prices, creating profit
opportunities through storage. Overall, inventories tend to return to their
long-run average at a slower pace than prices. Third, price volatility exhibits
sensitivity to inventory dynamics: a larger response in inventory tends to
generate an inversely larger response in price volatility. Fourth, the strength
of the IRFs is highly state-dependent, being more pronounced for high
availabilities.

\subsection{Model Specification}\label{ss:numer_formu}

For our quantitative analysis, we adopt a linear demand function:
\begin{equation}
	\label{eq:7}
	p \left( x \right)=\bar{p} \left[ 1+ \left( x/\mu_{Y}-1 \right)/\lambda \right],
\end{equation}
where $\bar{p}>0$ is the steady-state price,\footnote{If not otherwise
  specified, we designate by steady state the equilibrium in the absence of any
  shocks.}  $\mu_Y>0$ is the mean of the commodity output process (so also the
steady-state consumption level), and $\lambda<0$ is the price elasticity of
demand.\footnote{\label{fn:2}An isoelastic inverse demand function has also been
  tested, and the results are robust to this change.} We assume that storage costs are
represented entirely by depreciation (i.e., $k=0$ and $\delta\ge0$), as previous
research shows that different types of storage costs have indistinguishable
effects on price moments \citep{Goue22a}.

The real annual interest rate, measured at a quarterly frequency, follows a first-order
autoregressive process:
\begin{equation}\label{eq:R_dist}
	R_t^a = \mu_R(1-\rho_R) + \rho_R R_{t-1}^a + 
	\sigma_R \sqrt{1-\rho_R^2} \, \epsilon_t^R, 
	\quad \{\epsilon_t^R\} \iidsim N(0, 1).
\end{equation}
Following Example~\ref{ex:demand_shock}, we treat $\{Y_t\}$ as a net supply
shock, where $Y_t=Y_t^S-Y_t^D$. Commodity output, $\{Y^S_t\}$, follows a
truncated normal distribution with mean $\mu_Y$ and standard deviation
$\mu_Y \sigma_Y$, truncated at 5 standard deviations. The truncation ensures a lower
bound for commodity output and total available supply \citep[see,
e.g.,][]{deaton1992on}. Commodity demand is proportional to economic activity:
$Y_t^D=\alpha A_t$. A simple IS curve represents economic activity:
\begin{equation}
	\label{eq:1}
	A_t = \rho_{A} A_{t-1} - \gamma (R_t^a - \mu_R),
\end{equation}
where $|\rho_{A}| <1$ parameterizes the persistence of economic activity, and
$\gamma\ge 0$ controls the sensitivity of economic activity to deviations in the
interest rate from its mean. For constant interest rates at the mean, economic
activity is just 0.\footnote{We assume a contemporaneous effect of interest rate
  on economic activity to be consistent with recent VAR results \citep{Baue23}.}

To simplify the problem, we have assumed that $A_t$ is
driven only by interest rates and by its own persistence. These assumptions avoid
the need to identify the innovation process of economic activity and, more
importantly, to represent the effect of economic activity on interest rates. The
joint dynamic of $(R_t^a,A_t)$ is represented by a SVAR(1) model:
\begin{equation}
	\label{eq:4}  
	\begin{bmatrix}
		1 & 0 \\
		\gamma & 1
	\end{bmatrix} 
	\begin{bmatrix}
		R_t^a \\
		A_{t}
	\end{bmatrix} =
	\begin{bmatrix}
		\mu_{R}(1-\rho_R) \\
		\gamma\mu_{R}
	\end{bmatrix} +
	\begin{bmatrix}
		\rho_R & 0 \\
		0 & \rho_{A}
	\end{bmatrix}
	\begin{bmatrix}
		R_{t-1}^a \\
		A_{t-1}
	\end{bmatrix} +
	\begin{bmatrix}
		\sigma_{R}\sqrt{1-\rho_R^2} \\
		0
	\end{bmatrix}
	\epsilon_t^R.
\end{equation}
To make this process compatible with our assumptions, we discretize it.

This setup is a special case of the theoretical framework developed in
Section~\ref{sec:itheory}, with $Z_t=(R_t^a,A_t)$, $\epsilon_t \equiv 0$, $\eta_t = Y^S_t$,
and $y(Z_t, \eta_t) = Y_t^S-\alpha A_t$. We use an
annual interest rate process to obtain results that are directly comparable to others in
the literature, but the model calls for an interest rate at the quarterly
frequency, so we define $R_t=r(Z_t, \epsilon_t) = (R_t^a)^{1/4}$. Below
we estimate the interest rate process and show that $\rho_R>0$ (implying
that $\{R_t^a\}$ is a monotone Markov process) and that the discount condition in
Assumption~\ref{a:opt} holds. Hence, all the statements of Theorem~\ref{t:opt}
and Proposition~\ref{pr:storage} are valid.

In a first step, we will analyze the speculative channel without the global
demand channel (so assuming $\alpha=0$). In this setting, $r(z,\epsilon)$ and
$y(z,\eta)$ are increasing in $z$. So, Proposition~\ref{pr:mono_z} is valid, and
since in this case $\{R_t\}$ is independent of $\{Y_t\}$, the assumptions (and
thus conclusions) of Proposition~\ref{pr:neg_corr} also
hold. Propositions~\ref{pr:mono_z} and~\ref{pr:neg_corr} do not generally hold
with the global demand channel.

This choice of parameterization limits the free parameters that matter in the
analysis of price movements to $\delta$, $\lambda$, $\alpha$, and $\sigma_{Y}$. Indeed, the interest
rate process is estimated on observations, the economic activity process is
calibrated based on \citet{Baue23}, and we can normalize $\bar{p}$ and
$\mu_Y$ to unity, since their effect is only to set the average price and
quantity levels. To ease 
interpretation and limit the number of parameters to adjust,
we fix $\sigma_Y$ to 0.05.\footnote{According to \citet{Goue22a}, 
	a coefficient of variation of 5\% for the net supply shock is slightly 
	above the total shock (demand plus supply) affecting the aggregate crop 
	market of maize, rice, soybeans, and wheat, but below the shocks affecting 
	each of these markets individually.} 
If only the speculative channel is active, this choice
is innocuous as we can prove that adjusting the intercept and slope of the
demand function is equivalent to adjusting the mean and variance of the output
process \citep[see the proof in Appendix~\ref{sec:equiv}, which is a
generalization of Proposition~1 of][]{deaton1996competitive}.

To calibrate the real interest rate process $\{R^a_t\}$, we follow the
literature on interest rates and commodity prices \citep[e.g.,][]{frankel2008,
  gruber2018interest, Kili22} and use the nominal one-year treasury
yield. Compared to the 3-month Treasury yield, the one-year yield is less prone
to the distortions caused by the zero lower bound, while still capturing meaningful policy
shifts. We deflate this rate by a measure of expected inflation, calculated through an autoregressive model estimated on a 30-year
window prior to the year of interest to account for changes in the dynamics of
inflation.\footnote{Using lagged inflation or ex-post inflation would lead to a
  similar real interest rate process.} The real interest rate obtained in this
way is represented in \fref{fig:real_int}. Maximum likelihood estimation
of~\eqref{eq:R_dist} over the period 1962--2022 yields $\mu_R = 1.0062$,
$\rho_R = 0.9407$, and $\sigma_R = 0.03$.

These results imply that the stationary mean of the real interest rate process
is about $0.6\%$, with an unconditional standard deviation of about $3\%$.  Note
that any constant spread above the risk-free rates can be captured in our model
by $\delta$. Therefore, when interpreting the values of $\delta$ in what follows, it
should be kept in mind that $\delta$ represents storage costs, a premium above
risk-free rates, and any long-run trend in commodity prices.\footnote{See
  \citet{Bobe17} for an analysis of the role of commodity price trends in the
  storage model.}

We use proxy SVAR estimates from \citet{Baue23} to calibrate the economic
activity process. Our economic activity process presents two free parameters:
$\rho_{A}$ and $\gamma$. We calibrate them by matching two moments: the number of months
needed to attain the minimum activity level after a monetary shock and the
size of the decrease at the minimum. This calibration is done conditional on the
interest rate process calibrated previously. The preferred estimation of
\citet[Figure~8]{Baue23} indicates 9 months to reach a decrease of industrial
production of $-0.4$\% after a 25~basis-point monetary shock. In our setting, this leads
to $\rho_{A}=0.52$ and $\gamma=0.95$.

If only the speculative channel is active, we discretize the interest rate
process~\eqref{eq:R_dist} into an finite-state Markov chain using the method of
\citet{tauchen1986finite}. If both channels are active, we transform the SVAR(1)
model of \fref{eq:4} into a VAR(1) model and discretize it using the approach of
\citet{Schm14a}.

Demand is specified as $p^{-1}(P_t) + Y_t^D$. With the steady-state value of
$p^{-1}(P_t)$ normalized to unity and $Y_t^D$ having a zero mean, the parameter
$Y_t^D$ may be viewed as a deviation from the steady-state demand level. Within
this framework, the parameter $\alpha$ plays a crucial role in dictating the extent
to which shocks to economic activity influence the commodity market. For our
central calibration, we choose $\alpha = 0.2$ to align with the immediate price
reactions of a commodity index to positive interest rate shocks, as identified
in the IRFs reported in \citet[Figure~8]{Baue23}. However, to acknowledge the
variability in the responsiveness of different commodities to economic
conditions and ensure robustness of our findings, we conduct a series of
simulations wherein we systematically vary the value of $\alpha$.

Our first step is to verify Assumption~\ref{a:opt}, which requires
$\kappa(M) > -\delta$. \Fref{fig:stab} plots $\kappa(M)$ for different values of
$(\mu_R, \rho_R, \sigma_R)$.\footnote{The method for computing $\kappa(M)$ is described in
  Lemma~\ref{lm:yy} of Appendix~\ref{sec:proof-sect-refs}.} The figure shows
that $\kappa(M)$ is increasing in $\mu_R$ and decreasing in $\sigma_R$ and
$\rho_R$.  In general, $\kappa(M)>-\delta$ fails only when $\mu_R$ is sufficiently low, or when
$\rho_R$ or $\sigma_R$ is very large. The black solid curves represent the thresholds at
which $\kappa(M) = -0.02, -0.01, 0$, respectively. Clearly, Assumption~\ref{a:opt}
holds at the estimated values of $(\mu_R, \rho_R, \sigma_R)$, even when $\delta=0$.

\begin{figure}
	\begin{subfigure}{.5\linewidth}
		\centering
		\includegraphics[width=\linewidth]{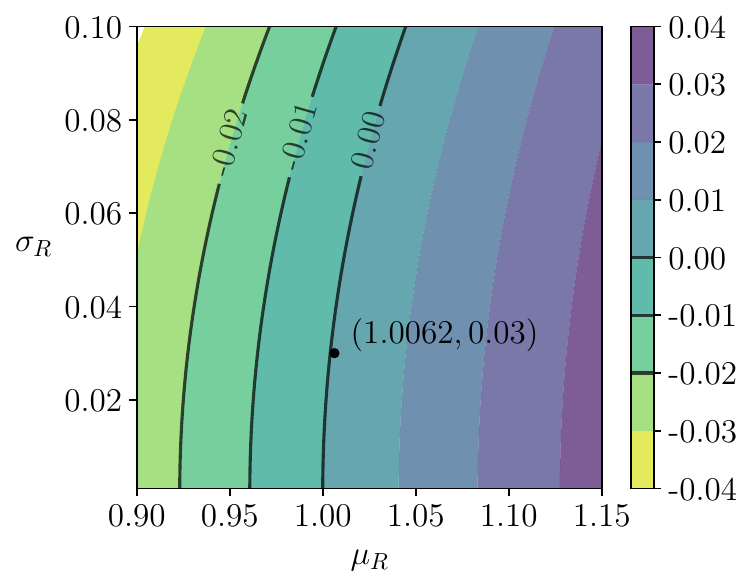}
	\end{subfigure}%
	\begin{subfigure}{.5\linewidth}
		\centering
		\includegraphics[width=\linewidth]{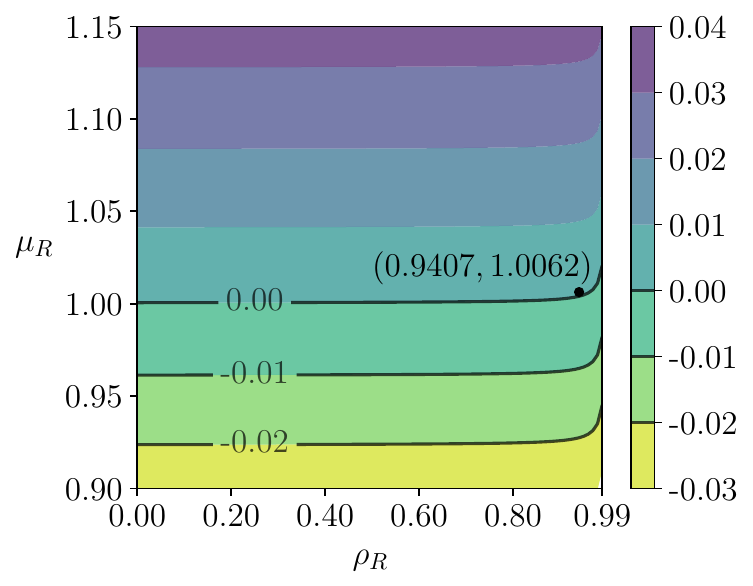}
	\end{subfigure}
	\caption{$\kappa(M)$ values for different combinations of
          $(\mu_R, \rho_R, \sigma_R)$.  The parameter $\kappa(M)$ represents the asymptotic
          yield on risk-free zero-coupon bonds as maturity approaches infinity,
          with $(\mu_R, \rho_R, \sigma_R)$ controlling the dynamics of the gross real
          interest rate process. In the left panel, $\rho_R$ is fixed at its
          estimated value, and a contour plot of $\kappa(M)$ is provided as a
          function of $(\mu_R, \sigma_R)$. In the right panel, $\sigma_R$ is fixed at its
          estimated value, and a contour plot of $\kappa(M)$ is shown as a function
          of $(\mu_R, \rho_R)$.}\label{fig:stab}
\end{figure}

After verifying Assumption~\ref{a:opt}, we solve for the equilibrium pricing
rule using a modified endogenous grid method \citep{carroll2006the}. To capture
the nonlinear dynamics generated by the storage model, we calculate IRFs as
state-and-history-dependent random variables, following the approach of
\citet{Koop96}, and use these IRFs to explore how interest rate shocks affect
prices, inventories, and price volatility. Details of the algorithm and
computation are given in Appendix~\ref{s:alg}.\footnote{Replication materials 
are available at \url{https://github.com/qingyin-ma/cp_code}.}

\subsection{Experiments}\label{ss:qe}

We begin by examining the speculative channel in isolation (i.e., assuming
$\alpha=0$). We then investigate the role of the global demand channel. While
speculative incentives remain active under the global demand channel (for
example, via autocorrelation of economic activity, which affects the
speculators' expectations), we can isolate and analyze the effects through each
channel by comparing the IRFs.

\subsubsection{Speculative Channel}\label{sss:speculative-channel}

Unless otherwise specified, we assume $\delta=0.02$ and $\lambda=-0.06$. This combination
of parameters leads to the following price moments on the asymptotic
distribution: a coefficient of variation of $24\%$, a first-order
autocorrelation of $0.61$, and a skewness of $2.9$, aligning with the empirical
observations.\footnote{\label{fn:6}See, e.g., Table~V in \citet{Goue15} for measures of
  these moments for a sample of commodities.}

Since the Markov process we adopt here is symmetric around the mean, the pricing
rule fluctuates symmetrically around the corresponding constant-discounting
pricing rule. This has implications for unconditional price moments. Notably,
the standard price moments of interest---the coefficient of variation,
autocorrelation, and skewness---are nearly indistinguishable between this model
and a constant-discoun\-ting model with the same average discounting (differing
only at three digits). While this outcome may seem surprising, considering that
the real rate volatility could be perceived as an additional source of
volatility, in practice this does not create additional standard demand or
supply shocks. When the interest rate falls below its mean, it prompts
additional demand for storage, driving prices up. However, when the interest
rate rises, these additional stocks are sold, exerting downward pressure on
prices. These effects tend to offset each other. This indicates that the
speculative channel may not contribute significantly to empirical analyses based
on unconditional price moments.

However, the presence of time-varying interest rates holds significant
implications for conditional moments, which could be exploited empirically.
Below, we delve into this exploration using IRFs. We calculate the IRFs to
a 100~bp interest rate impulse (i.e., a $1\%$ increase in the real interest
rate). All IRFs represent percentage deviation from the benchmark simulation.

\Fref{fig:irf_param} shows the IRFs calculated at the stationary mean of
$(X_{t-1},R_{t-1}^a)$. We first discuss the central IRFs corresponding to
$\delta=0.02$ and $\lambda=-0.06$, before analyzing the sensitivity of these IRFs to the
parameters and states. The left panels present the IRFs for prices, which show
an immediate price decrease followed by a gradual convergence to the long-run
average over 2 to 4 years.\footnote{\label{fn:3}These dynamics are very 
	different from what would be expected from a transitory MIT shock to the 
	interest rate in a standard storage model with a constant interest rate. In 
	the latter scenario, an unexpected increase in the interest rate would also 
	depress prices due to decreased stockholding. However, this price decrease 
	would be short-lived, lasting only one period. With a transitory shock, the 
	interest rate would revert to its baseline after one period, incentivizing 
	stockpiling and consequently driving prices above their non-shock levels in 
	all subsequent periods. This point is demonstrated in
	Appendix~\ref{sec:mit-shocks}, Proposition~\ref{pr:mit-shock-constant}.} 
The middle panels display the IRFs
for inventories, which, unlike prices, reach their lowest value more than a year
after the shock, and even after 4 years, they have not returned to their
long-run values. Finally, the right panels illustrate the IRFs for price
volatility, namely, the conditionally expected standard deviation of
price.\footnote{Price volatility is the
	square root of the conditional variance:
	$\sqrt{\EE_{t-1} [f^{*}(X_{t},Z_{t})]^2 -[\EE_{t-1} f^{*}(X_{t},Z_{t})]^2}$.}
They indicate that price volatility largely follows stock dynamics with a peak
reached after a year and an incomplete convergence after 4 years. This finding
is consistent with the empirical results of \citet{gruber2018interest}, who show that
higher interest rates imply higher price volatility.

\begin{figure}
	\centering
	\scalebox{0.7}{\includegraphics{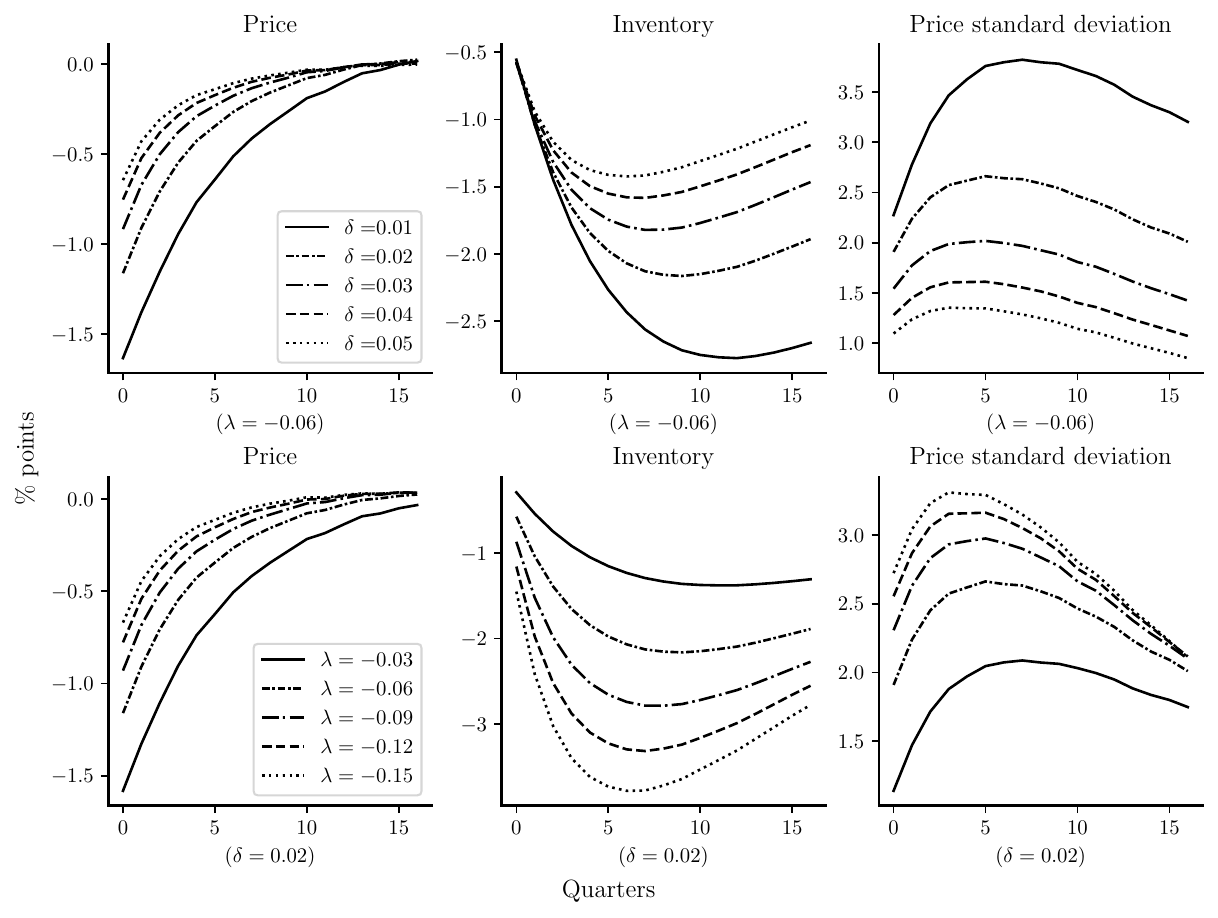}}
	\caption{\label{fig:irf_param}IRFs for a 100~bp real interest rate shock 
		under different parameter setups without demand channel, fixing 
		$X_{t-1}$ and $R_{t-1}^a$ at the stationary mean. $X_{t-1}$ and 
		$R_{t-1}^a$ are the total available supply and the annual 
		gross real interest rate in the previous quarter, $\delta$ is the rate 
		of depreciation, and $\lambda$ is the demand elasticity.}
\end{figure}

The mechanism explaining this behavior is as follows. When the interest rate
increases, speculators tend to dispose of stocks. This
decreases current prices due to increased supply. However, the price decrease
mitigates the extent of stock selling compared to what would happen if prices
remained constant. In subsequent periods, stocks remain excessive due to
persistent high opportunity costs, prompting speculators to continue selling
them. Again, the price decrease acts as a cushion, preventing a complete
sell-off.  After more than a year of this dynamic, with increasingly
smaller quantities being sold from inventories, interest rates have declined,
alleviating the pressure to sell inventories, and prices are below their
long-run values with anticipation of eventual convergence.  Consequently, stock
accumulation gradually increases as interest rates decline. Since stock accumulation is slow, prices converge to their
long-run values from below without overshooting, although this remains a
possibility. After four years, prices have converged to their long-run values,
as the only factor influencing the increase in stock levels is the reversion of
the interest rate to its long-run average. As for price volatility, its
behavior, opposite to that of inventories, can be explained by the fact that
inventories serve as a key determinant of price volatility by providing a buffer
against production shocks.

\Fref{fig:irf_param} also includes IRFs for various sets of parameters. The
impact on commodity price increases as storage costs decrease: lower storage 
costs result in a higher initial decline and a slower return to 
equilibrium.\footnote{\label{fn:IRF_delta}Price 
	dynamics may not consistently show a less significant decrease with a 
	higher storage cost when the global demand channel is activated, in which 
	case storage costs have only a marginal and short-term effect on price 
	dynamics.} 
Intuitively, when storage costs are lower, opportunity costs become more 
significant and, as a result, variations in the interest rate have a greater 
influence on storage behavior and prices. This effect is evident in the 
observation that stocks decrease more sharply with an interest rate increase 
when storage costs are lower. In such scenarios, stock levels tend to be higher 
on average and more susceptible to changes in opportunity costs.

The price effects are more pronounced with a more inelastic demand function, as
prices react significantly to variations in the sale of stock when demand is
less elastic. Likewise, the stock decrease is less substantial when demand is
more inelastic. Since with inelastic demand, even a
slight increase in sales from inventory can greatly depress prices, thereby
reducing the incentive to sell off stocks excessively.

To explore the sensitivity of IRFs to states, \fref{fig:irf_state} draws the
IRFs calculated for different realized values of $(X_{t-1}, R_{t-1}^a)$. We use
$(X_{t-1}^p, R_{t-1}^{a,p})$ to denote the percentile points of the realized
$(X_{t-1}, R_{t-1}^a)$ states on the stationary distribution. The top left panel
shows that price responses are stronger when availability becomes larger. The
immediate responses of price are respectively 1.72 and 2.17 times larger when
availability increases from the $25\%$ percentile to the $75\%$ and $95\%$
percentiles. This is because when availability is lower, inventory tends to be
lower (Proposition~\ref{pr:storage}), hence there is less room for stock
adjustment and prices react much less in response to the interest rate shock.
This intuition is verified by the top middle panel, which shows that a higher
availability causes stock decumulation to last longer, yielding a larger decline
in inventory in the medium to long run (despite a slightly lower immediate
decline). 

\begin{figure}
	\centering
	\scalebox{0.7}{\includegraphics{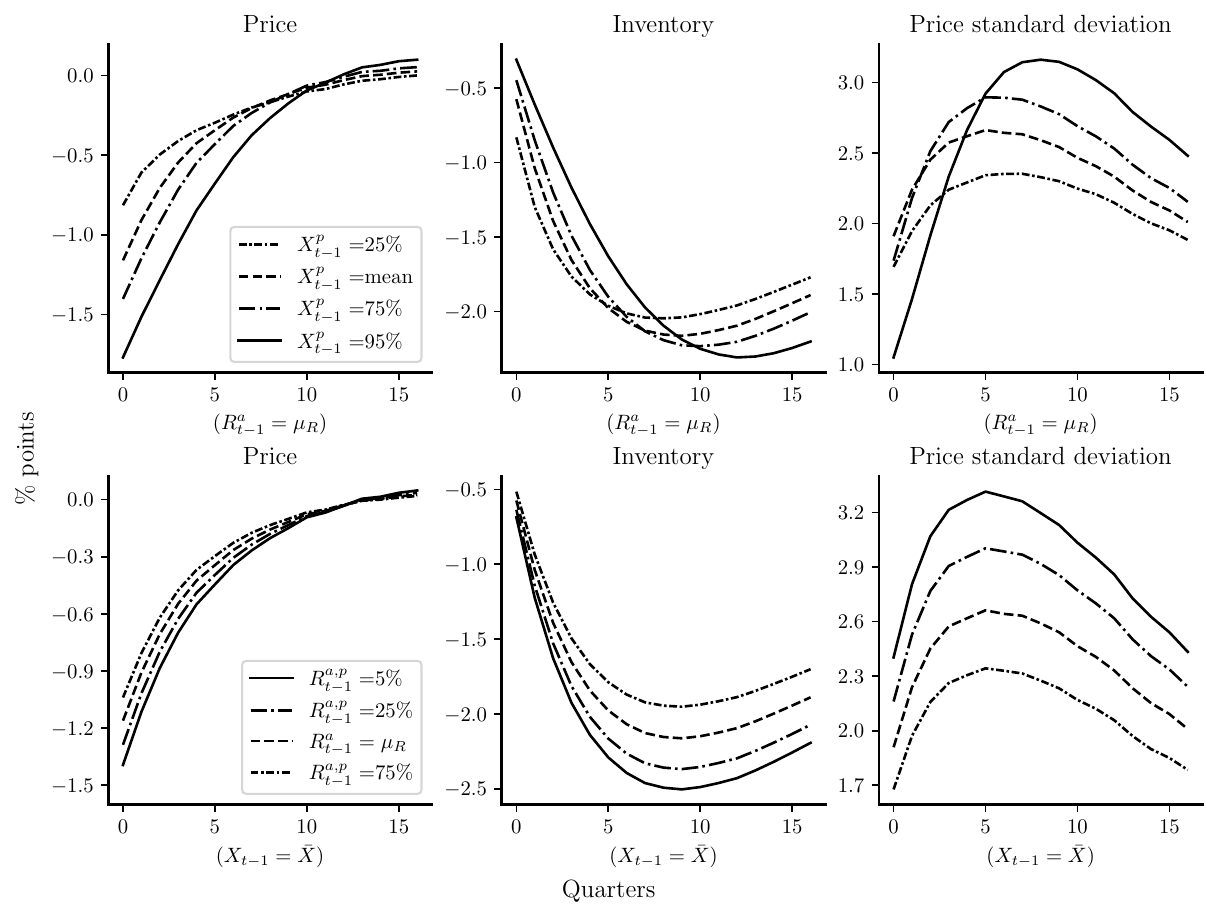}}
	\caption{\label{fig:irf_state}IRFs for a 100~bp real interest rate shock
          conditional on different states without demand channel.
          $(X_{t-1}^p, R_{t-1}^{a,p})$ denotes the percentile points of the
          realized total available supply and interest rate states on the
          stationary distribution, and $\mu_R$ and $\bar X$ are the stationary
          means of the interest rate and the total available supply processes,
          respectively.}
\end{figure}

The bottom left panel in \fref{fig:irf_state} shows that price responses to a
100~bp interest rate shock tend to be slightly larger when interest rates are
relatively lower. The overall trend of price and storage IRFs in the bottom
panels of \fref{fig:irf_state} is consistent with our theory
(Proposition~\ref{pr:mono_z}), which shows that under lower interest rates,
prices and inventories are generally higher and, therefore, more sensitive to
variations in opportunity costs.\footnote{\label{fn:IRF_Rs}When the global 
	demand channel is activated, market dynamics become more complex, in 
	which case this monotone pattern can be disrupted due to global economic 
	activity. We leave further investigation of this issue for future 
	research.}

As in the previous cases, the right panels of \fref{fig:irf_state} show that
the dynamics of price volatility are highly consistent with the inventory
dynamics, with a larger response in speculative storage causing an oppositely
larger response in price volatility. 

\subsubsection{Global Demand Channel}\label{sss:glob-demand-chann}

\begin{figure}
	\centering
	\includegraphics[width=\textwidth]{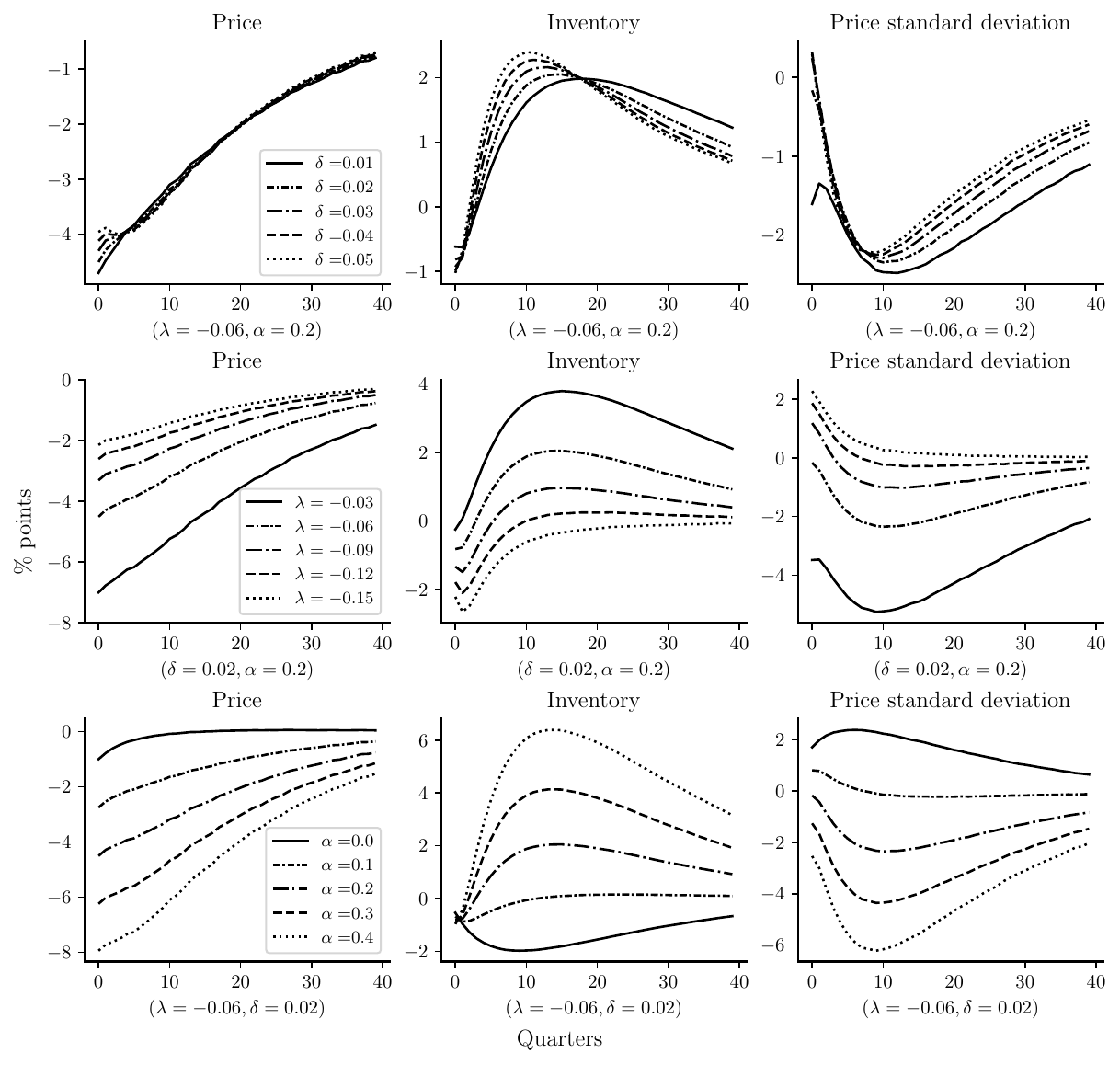}
	\caption{\label{fig:irf_demand_param}IRFs for a 100~bp real interest
		rate shock under different parameter setups, fixing $X_{t-1}$ and
		$R_{t-1}^a$ at the stationary mean. $X_{t-1}$ and $R_{t-1}^a$ are the
		total available supply and the annual gross real interest rate in the
		previous quarter, $\delta$ is the rate of depreciation, $\lambda$ is 
		the demand
		elasticity, and $\alpha$ governs the sensitivity of the demand shock to
		economic activity.}
\end{figure}

Next, we study the global demand channel as described in
Section~\ref{ss:numer_formu}. 
A notable difference with the speculative channel is that the unconditional
price moments are affected by the presence of the demand channel. This is due to
the fact that the demand channel involves a demand shock driven by interest
rates shocks, while the speculative channel implies offsetting demand and supply
shocks. In addition, since the demand shock follows an autocorrelated process,
it adds persistence to the price process. In the benchmark setting
($\alpha=0.2$), for example, price autocorrelation increases from 0.61 to 0.88. Thus,
unlike the speculative channel, the demand channel can play a role in
econometric analysis of the storage model using unconditional price
moments.\footnote{This potential role of the demand channel aligns with the
  findings of \citet{Goue22a}, which elucidate the empirical relevance of
  demand shocks and their role in explaining price autocorrelation.}
Nonetheless, the speculative channel, as well as the demand channel, could be
empirically relevant for estimation strategies relying on conditional moments or
unconditional moments calculated on quantities.

\Fref{fig:irf_demand_param} shows the IRFs with the global demand channel for 
various parameters. The impact of an interest rate shock on price is 
consistently more pronounced when the demand channel is activated. This is 
because the demand channel amplifies the price decline already induced by the 
speculative channel. In our central calibration, the dominance of the demand 
channel is evident, as the price effect is three times larger than under the 
speculative channel alone. Furthermore, since interest rates take a long time to
converge back to their steady state, so does economic activity. As a 
consequence, the price effect also becomes more persistent. 

The inventory dynamics differ notably with the inclusion of the demand channel.
Initially, inventories experience a decline comparable in magnitude to that
observed under the speculative channel alone. Subsequently, they rebound and
tend to surpass their long-run level. This overshooting behavior indicates that,
after several quarters, the influence of low demand outweighs that of high
interest rates. The low demand leads to suppressed prices with an expectation of
future price recovery once demand returns to normal, incentivizing higher
inventory holdings despite the disincentives posed by elevated interest rates.

These inventory dynamics also manifest in the behavior of price standard
deviation, which moves inversely to inventory dynamics. Notably, it is important
to recognize that the figure depicts the price standard deviation, rather than
the coefficient of variation. Given that price standard deviation decreases less
than the price level, the coefficient of variation actually experiences an
increase over time.

The sensitivity of the IRFs to the model's parameters facilitates a deeper
understanding of the mechanisms behind the demand channel. Storage cost 
exhibits a minimal impact on price dynamics, except within the initial 
quarters. A comparison with \fref{fig:irf_param} suggests that this marginal 
effect is primarily driven by the speculative channel, indicating that the bulk 
of the demand channel's impact stems from a direct demand shock. Any indirect 
effect of the demand shock through storage is likely limited. The effect of 
demand elasticity is straightforward: a less elastic demand corresponds to a 
more pronounced price decrease and a greater increase in stock levels. 

The parameter $\alpha$ governs the sensitivity of the demand shock
to economic activity and, as described above, we calibrate its baseline value to 
$\alpha=0.2$ to align with the commodity index dynamics of \citet[][Figure~8]{Baue23}.
Given the heterogeneous exposure of commodities to economic 
activity, the importance of the demand channel varies across different 
commodities. Modulating this parameter elucidates the relative significance of 
the demand channel compared to the speculative channel (case $\alpha=0$).

\section{Conclusion}\label{sec:concl}

This paper extends the classical competitive storage model to account for
time-varying interest rates, offering a unified theoretical framework for
understanding how interest rates and other aggregate factors influence commodity
prices. We proposed readily verifiable conditions under which a unique
equilibrium solution exists and can be efficiently computed. These conditions
have a natural interpretation in terms of the asymptotic yield on long-maturity
risk-free assets. We also provided a sharp characterization of the analytical
properties of the equilibrium objects and developed an efficient solution
algorithm.

Within this framework, we investigated the dynamic causal effect of interest
rates on commodity prices from a theoretical and quantitative perspective. Our
theoretical results established conditions under which interest rates exert a
negative effect on commodity prices or are negatively correlated with them. The
quantitative analysis explored the impact through the
speculative and global demand channels. Impulse response analysis demonstrated a
substantial and persistent negative effect of interest rates on commodity prices
in most empirically relevant scenarios. Moreover, the magnitude of this effect
is state-dependent, varying according to market supply conditions and interest
rate regimes.

Our quantitative application focuses on the speculative and the global demand
channels. However, exploring the impact of interest rates on commodity
prices through various other channels, and the impact of more sophisticated
stochastic discount factors \citep[as found, for example,
in][]{schorfheide2018identifying} and risky returns are equally important.
Although these topics lie beyond the scope of the current paper, the theory we
develop provides a solid foundation for such extensions.


\ifpdf\pdfbookmark{Appendix}{sec:Appendix}\fi
\appendix 

\numberwithin{equation}{section}
\renewcommand\theequation{\thesection.\arabic{equation}}
\setcounter{equation}{0}
\numberwithin{figure}{section}
\renewcommand\thefigure{\thesection.\arabic{figure}}
\setcounter{figure}{0}
\numberwithin{table}{section}
\renewcommand\thetable{\thesection.\arabic{table}}
\setcounter{table}{0}

\section{Proof of Section~\ref{sec:or} Results}\label{sec:proof-sect-refs}

Here and in the remainder of the appendix, we let $\Phi$ be the probability
transition matrix of $\{Z_t\}$. In particular, $\Phi(z, \hat z)$ denotes the
probability of transitioning from $z$ to $\hat z$ in one step. Recall $M_t$
defined in~\eqref{eq:qy_func}.  We denote $\EE_z \coloneq \EE(\cdot \mid Z=z)$ and
$\EE_{\hat z} \coloneq \EE(\cdot \mid \hat Z = \hat z)$, and introduce the matrix $L$
defined by
\begin{equation}\label{eq:ell}
	L(z, \hat z) \coloneq \Phi \left(z, \hat z\right) \EE_{\hat z} \, m(\hat z, \hat \epsilon).
\end{equation}
Here $L$ is expressed as a function on $\ZZ \times \ZZ$ but can be represented
in traditional matrix notation by enumerating $\ZZ$.  Specifically, if $\ZZ =
\{z_1, \ldots, z_N\}$, then $L = \Phi D$, where $D \coloneq \diag \left\{ \EE_{z_1} M,
\ldots, \EE_{z_N} M \right\}$.

For a square matrix $A$, let $s(A)$ denote its spectral radius. In other words,
$s(A) \coloneq \max_{\alpha \in \Lambda} |\alpha|$, where $\Lambda$ is the set of
eigenvalues of $A$.

\begin{lemma}\label{lm:yy}
    Given $L$ defined in~\eqref{eq:ell}, the asymptotic yield satisfies
    $\kappa(M) = - \ln s(L)$.       
\end{lemma}

\begin{proof}
    By induction, we can show that, for any function $h: \ZZ \to \RR$ and $n \in
    \NN$,
	\begin{equation}\label{eq:L_map}
		L^n h (z) = \EE_z \left(\prod_{t=1}^{n} M_t\right) h(Z_n),
	\end{equation}
    where $L^n$ is the $n$-th composition of the operator $L$ with itself or,
    equivalently, the $n$-th power of the matrix $L$. By Theorem~9.1 of
    \citet{krasnosel2012approximate} and the positivity of $L$, we have 
	\begin{equation}\label{eq:sL}
		s(L) = \lim_{n\to\infty} \| L^n h \|^{1/n},
	\end{equation}
    where $h$ is any function on $\ZZ$ that takes positive values, and $\|
    \cdot\|$ is any norm on the set of real-valued functions defined on $\ZZ$.
    Letting $h \equiv 1$ and $\|f\| \coloneq \EE |f(Z_0)|$, we obtain
	\begin{equation*}
		s(L) 
		= \lim_{n \to \infty} \left(
		\EE \left|
		\EE_{Z_0} \prod_{t=1}^{n} M_t 
		\right|
		\right)^{1/n} 
		= \lim_{n \to \infty} \left(
		\EE \prod_{t=1}^n M_t
		\right)^{1/n} 
		= \lim_{n\to\infty} q_n^{1/n},
	\end{equation*}
    where the first and the last equalities are by definition, and the second
    equality is due to the Markov property. Since the log function is
    continuous, we then have $\ln s(L) = \lim_{n \to \infty} \ln q_n /n =
    -\kappa(M)$ and the claim follows.
\end{proof}

\begin{corollary}\label{cor:sL}
	Assumption~\ref{a:opt} holds if and only if $s(L) <\me^{\delta}$.
\end{corollary}

This follows directly from $\kappa(M) = - \ln s(L)$. Below we routinely use the
alternative version $s(L) < \me^{\delta}$ for Assumption~\ref{a:opt}.

In the main text we imposed $p(b) < \infty$ to simplify analysis. Here and
below, we relax this assumption. We assume instead $p(b) \leq \infty$, and prove
that all the theoretical results in Sections~\ref{sec:or}--\ref{sec:itheory}
still hold in this generalized setup. To that end, we assume that
\begin{equation}\label{eq:EY_bd}
	\EE_z \max \{p(Y), 0\} < \infty \quad \text{for all } z \in \ZZ.
\end{equation}
(This mild assumption confines the expected market value of commodity
output to be finite and holds trivially in the setting of
\fref{sec:or}, where $p$ is bounded above.)
We then update the endogenous state space $\XX$ and define it as
\begin{equation*}
  \XX \coloneq 
  \begin{cases}
    (b, \infty),  &\text{if } p(b) = \infty, \\
    [b, \infty),  &\text{if } p(b) < \infty.
  \end{cases}
\end{equation*}
There is no loss of generality to truncate the endogenous state space when $p(b)
= \infty$, because in this case,~\eqref{eq:EY_bd} implies that $Y_t > b$ almost
surely, and thus $X_t > b$ with probability one for all $t$.

Let $p_0(x) \coloneq \max \{p(x), 0\}$ and let
$\cC$ be all continuous $f \colon \SS \to \RR$ such that $f$ is
decreasing in its first argument, $f(x,z) \geq p_0(x)$ for all $(x,z) \in \SS$,
and
\begin{equation*}
	\sup_{(x,z) \in \SS} \left| 
	f(x,z) - p_0(x)
	\right| < \infty.
\end{equation*}
Obviously, $\cC$ reduces to the candidate space in Theorem~\ref{t:opt} when the
demand function $p$ is bounded above, i.e., when $p(b) < \infty$. To compare
pricing policies, we metrize $\cC$ via
\begin{equation*}
	\rho (f,g) \coloneq \|f - g\| 
	\coloneq \sup_{(x,z) \in \SS} |f(x,z) - g(x,z)|.
\end{equation*}
Although $f$ and $g$ are not required to be bounded, one can show that $\rho$ is
a valid metric on $\cC$ and that $(\cC, \rho)$ is a complete metric space
\citep[see, e.g.,][]{ma2020income}.

We aim to characterize the equilibrium pricing rule as the unique fixed point of
the \textit{equilibrium price operator} described as follows: For fixed $f \in
\cC$ and $(x,z) \in \SS$, the value of $Tf$ at $(x,z)$ is defined as the $\xi
\geq p_0(x)$ that solves
\begin{equation}
	\label{eq:sol_euler}
	\xi = \psi (\xi, x,z)  
	\coloneq \min \left\{
	\max \left\{
	\me^{-\delta} \EE_z \hat M 
	f \left(
	\me^{-\delta} I (\xi, x, z) + \hat Y, \hat Z 
	\right) - k, \, p(x)
	\right\}, \, p(b)
	\right\},
\end{equation} 
where, considering free disposal, 
\begin{equation}\label{eq:storage}
	I(\xi, x, z) \coloneq 
	\begin{cases}
		x - p^{-1} (\xi), & \text{if }\, x < x_f^*(z) \\
		x_f^*(z) - p^{-1} (\xi), & \text{if }\, x \geq x_f^*(z) 
	\end{cases}
\end{equation}
with
\begin{equation*}
	x_f^*(z) \coloneq \inf \left\{
	x \geq p^{-1}(0) \colon 
	\me^{-\delta} \EE_z \hat M f \left(
	\me^{-\delta} [x - p^{-1}(0)] + \hat Y, \hat Z
	\right) - k = 0
	\right\}.
\end{equation*}
The domain of $\psi$ is
\begin{equation}
	\label{eq:G_set_gen}
	\mathsf G \coloneq \left\{ 
	(\xi, x, z) \in \RR_+ \times \SS 
	\colon \xi \in B(x) 
	\right\},
\end{equation}
where $B(x)$ is defined for each $x$ as 
\begin{equation}
	\label{eq:b_set}
	B(x) \coloneq 
	\begin{cases}
		[p_0(x), \infty),  &\text{if } p(b) = \infty, \\
		[p_0(x), p(b)],  &\text{if } p(b) < \infty.
	\end{cases}
\end{equation}

\begin{proposition}\label{pr:wel_def}
	If $f \in \cC$ and $(x,z) \in \SS$, then $Tf (x,z)$ is uniquely defined. 
\end{proposition}

\begin{proof}
    Fix $f \in \cC$ and $(x,z) \in \SS$. Since $f$ is decreasing in its first
    argument and $p^{-1}$ is decreasing (by the inverse function theorem), the
    map $\xi \mapsto \psi(\xi, x, z)$ is decreasing. Since the left-hand-side of
    \fref{eq:sol_euler} is strictly increasing in $\xi$,~\eqref{eq:sol_euler}
    can have at most one solution. Hence, uniqueness holds.  Existence follows
    from the intermediate value theorem provided we can show that
	\begin{enumerate}
		\item[(a)] $\xi \mapsto \psi (\xi, x,z)$ is a continuous function,
		\item[(b)] there exists $\xi \in B(x)$ such that 
		$\xi \leq \psi (\xi, x,z)$, and
		\item[(c)] there exists $\xi \in B(x)$ such that 
		$\xi \geq \psi (\xi, x, z)$. 
	\end{enumerate}
	For part (a), it suffices to show that 
	\begin{equation*}
		g(\xi) \coloneq \EE_z \hat M f \left(
		\hat Y + \me^{-\delta} I(\xi,x,z), \hat Z
		\right)
	\end{equation*}
    is continuous on $B(x)$. To see this, fix $\xi \in B(x)$ and $\xi_n \to
    \xi$. Since $f \in \cC$, there exists $D \in \RR_+$ such that
	\begin{equation*}
		0 \leq \hat M f \left(
		\hat Y + \me^{-\delta} I(\xi_n,x,z), \hat Z
		\right)
		\leq \hat M f(\hat Y, \hat Z)
		\leq \hat M \, [ p_0(\hat Y) + D ].
	\end{equation*}
    Since $\EE_z \hat M p_0(\hat Y) = \EE_z \left[\EE_{\hat z} \hat M \,
    \EE_{\hat z}p_0(\hat Y)\right]$, the last term is integrable
    by~\eqref{eq:EY_bd}.  Hence, the dominated convergence theorem applies. From
    this fact and the continuity of $f$, $p^{-1}$, and $I$, we obtain $g(\xi_n)
    \to g(\xi)$. Hence, $\xi \mapsto \psi (\xi, x,z)$ is continuous.
	
	Regarding part (b), consider $\xi=p_0(x)$. If $p(x) \geq 0$, then $\xi = p(x)$ and thus
	\begin{equation*}
		\psi(\xi,x,z) \geq \min \{p(x), p(b)\}
		= p(x) = \xi.
	\end{equation*}
	If $p(x)<0$, then $\xi = 0$. In this case,
		$I(\xi,x,z) = I(0,x,z) \leq x_f^*(z) - p^{-1}(0)$.
	The monotonicity of $f$ and the definition of $x_f^*$ then imply that
	\begin{equation*}
		\me^{-\delta} \EE_z \hat M f \left(
		\me^{-\delta} I(\xi,x,z) + \hat Y, \hat Z
		\right) - k 
		\geq 
		\me^{-\delta} \EE_z \hat M f \left(
		\me^{-\delta} [x_f^*(z) - p^{-1}(0)] + \hat Y, \hat Z
		\right) - k
		=0.
	\end{equation*}
	By the definition of $\psi$, 
	\begin{equation*}
		\psi(\xi,x,z) \geq
		\min \{\max \{0,p(x)\}, p(b)\}
		= \min \{0,p(b)\} = 0 = \xi.
	\end{equation*}
	We have now verified part~(b).
	
    If $p(b) < \infty$, then part (c) holds by letting $\xi = p(b)$.  If $p(b) =
    \infty$, then part (c) holds as $\xi$ gets large since $\xi \mapsto \psi
    (\xi, x, z)$ is decreasing and bounded.
	
	In summary, we have verified both existence and uniqueness. 
\end{proof}

\begin{proposition}\label{pr:selfmap}
	$T f \in \cC$ for all $f \in \cC$.
\end{proposition}

\begin{proof}
	Fix $f \in \cC$ and define $g (\xi, x, z) \coloneq 
	\EE_z \hat M 
	f \left( 
	\hat Y + \me^{-\delta} I(\xi,x,z), 
	\hat Z 
	\right)$.

    First, we show that $Tf$ is continuous. To this end, we first show that $\psi$
    in~\eqref{eq:sol_euler} is jointly continuous on the set $\mathsf G$ defined
    in~\eqref{eq:G_set_gen}. This will be true if $g$ is jointly continuous on
    $\mathsf G$. For any $\{(\xi_n, x_n, z_n)\}$ and $(\xi, x, z)$ in $\mathsf
    G$ with $(\xi_n, x_n, z_n) \to (\xi, x, z)$, we need to show that $g(\xi_n,
    x_n, z_n) \to g(\xi, x, z)$. Define
	\begin{equation*}
		h_1(\xi, x, z, \hat Z, \hat \epsilon, \hat \eta), \,
		h_2(\xi, x, z, \hat Z, \hat \epsilon, \hat \eta)
		\coloneq \hat M f (\hat Y, \hat Z )
		\pm \hat M 
		f \left( 
		\hat Y + \me^{-\delta} I(\xi, x, z), 
		\hat Z 
		\right),
	\end{equation*}
    where $\hat M \coloneq m(\hat Z, \hat \epsilon)$ and $\hat Y = y (\hat Z, \hat
    \eta)$. Then $h_1$ and $h_2$ are continuous in $(\xi, x, z,\hat Z)$ by the
    continuity of $f$, $p^{-1}$, and $I$, and non-negative by the monotonicity
    of $f$ in its first argument.
	
    Let $\pi_\epsilon$ and $\pi_\eta$ denote respectively the probability
    measure of $\{\epsilon_t\}$ and $\{\eta_t\}$. Fatou's lemma and Theorem~1.1
    of \citet{feinberg2014fatou} imply that
	\begin{align*}
		&\int \int \sum_{\hat z \in \ZZ} 
		h_i (\xi, x, z, \hat z, \hat \epsilon, \hat \eta) 
		\Phi (z, \hat z) 
		\pi_{\epsilon} \diff (\hat \epsilon) 
		\pi_\eta \diff (\hat \eta)  \\
		& \leq \int \int 
		\liminf_{n \to \infty} \sum_{\hat z \in \ZZ} 
		h_i (\xi_n, x_n, z_n, \hat z, \hat \epsilon, \hat \eta)
		\Phi(z_n, \hat z) 
		\pi_\epsilon (\diff \hat \epsilon) 
		\pi_\eta (\diff \hat \eta)  \\
		& \leq \liminf_{n \to \infty} \int \int
		\sum_{\hat z \in \ZZ} h_i(\xi_n, x_n, z_n, \hat z, \hat \epsilon, \hat \eta)
		\Phi(z_n, \hat z) 
		\pi_\epsilon (\diff \hat \epsilon) 
		\pi_\eta (\diff \hat \eta).
	\end{align*}
	Since in addition $z \mapsto \EE_z \hat M f(\hat Y, \hat Z)$ is continuous, 
	we have
	\begin{equation*}
		\pm \EE_{z} \hat M 
		f \left( \hat Y + \me^{-\delta} I(\xi,x,z), \, \hat Z \right)  
		\leq 
		\liminf_{n \to \infty} \left( 
		\pm \EE_{z_n} \hat M 
		f \left( \hat Y + \me^{-\delta} I(\xi_n, x_n,z_n), \, \hat Z \right)
		\right).
	\end{equation*}
	Then $g$ is continuous, since the above inequality is equivalent to 
	\begin{equation*}
		\limsup_{n \to \infty} g(\xi_n, x_n, z_n) 
		\leq g(\xi, x, z)
		\leq \liminf_{n \to \infty} g(\xi_n, x_n, z_n).
	\end{equation*}
	Hence, $\psi$ is continuous on $\mathsf G$, as was to be shown. Since $\xi \mapsto \psi (\xi, x,z)$ takes values in 
	\begin{equation*}
		\Gamma(x,z) \coloneq \left[
		p_0(x), \,
		\min \left\{
		p(b), \, 
		p_0(x) + \me^{-\delta} \EE_z \hat M (p_0(\hat Y) + D)
		\right\}
		\right]
	\end{equation*}
    for some $D \in \RR_+$, and the correspondence $(x,z) \mapsto \Gamma(x,z)$
    is nonempty, compact-valued and continuous, Theorem~B.1.4 of
    \citet{stachurski2009economic} implies that $Tf$ is continuous on $\SS$.
	
    Second, we show that $Tf$ is decreasing in $x$. Suppose for some $z
    \in \ZZ$ and $x_1$, $x_2 \in \XX$ with $x_1 < x_2$, we have $\xi_1 \coloneq Tf
    (x_1, z) < Tf(x_2, z) =: \xi_2$. Since $f$ is decreasing in its first
    argument by assumption and $I$ defined in~\eqref{eq:storage} is increasing
    in $\xi$ and $x$, $\psi$ is decreasing in $\xi$ and $x$. Then $\xi_2 > \xi_1
    = \psi (\xi_1, x_1, z) \geq \psi(\xi_2, x_2, z) = \xi_2$, which is a
    contradiction. 
	
	Third, we show that $\sup_{(x,z) \in \SS} | Tf (x,z) - p_0(x) | < \infty$. This obviously holds since
	\begin{align*}
		&|Tf(x,z) - p_0(x)| = Tf(x,z) - p_0(x)  \\
		&\leq \me^{-\delta} \EE_z \hat M 
		f \left( 
		\hat Y + \me^{-\delta} I \left(Tf(x,z), x, z \right), 
		\hat Z 
		\right)
		\leq \me^{-\delta} \EE_z \hat M [p_0(\hat Y) + D ]  
	\end{align*}
    for all $(x,z) \in \SS$ and some $D \in \RR_+$, and the last term is finite
    by~\eqref{eq:EY_bd}.
    
    Finally, Proposition~\ref{pr:wel_def} implies that $Tf(x,z) \in B(x)$ for
    all $(x,z) \in \SS$. In conclusion, we have shown that $Tf(x,z) \in \cC$.
\end{proof}

\begin{lemma}\label{lm:order_pres}
    $T$ is order preserving on $\cC$. That is, $T f_1 \leq T f_2$ for all $f_1,
    f_2 \in \cC$ with $f_1 \leq f_2$.
\end{lemma}

\begin{proof}
    Let $f_1, f_2$ be functions in $\cC$ with $f_1 \leq f_2$. Recall $\psi$
    defined in~\eqref{eq:sol_euler}. With a slight abuse of notation, we define
    $\psi_f$ such that $\psi_f(Tf(x,z),x,z) = Tf (x,z)$ for $f \in \{f_1,
    f_2\}$.  Then $f_1 \leq f_2$ implies that $\psi_{f_1} \leq \psi_{f_2}$.
    Suppose to the contrary that there exits $(x,z) \in \SS$ such that $\xi_1 \coloneq
    Tf_1 (x,z) > Tf_2 (x,z) = \xi_2$.
	
    Since we have shown that $\xi \mapsto \psi(\xi, x,z)$ is decreasing for each
    $f \in \cC$ and $(x,z) \in \SS$, we have $\xi_1 = \psi_{f_1} (\xi_1, x, z)
    \leq \psi_{f_2} (\xi_1, x, z) \leq \psi_{f_2} (\xi_2, x, z) = \xi_2$, which
    is a contradiction. Therefore, $T$ is order preserving.
\end{proof}

\begin{lemma}\label{lm:ctr}
	There exist $N \in \NN$ and $\alpha \in (0,1)$ such that, for all $n \geq 
	N$, 
	\begin{equation*}
		\max_{z \in \ZZ} \, 
		\EE_z \prod_{t=1}^{n} \me^{-\delta} M_t < \alpha^n.
	\end{equation*}
    Moreover, $D_1 \coloneq \sum_{t=0}^{\infty} \max_{z \in \ZZ} \EE_z \prod_{i=1}^t
    \me^{-\delta} M_i < \infty$.
\end{lemma}

\begin{proof}
    The second inequality follows immediately from the first inequality. 
    To verify the first inequality, note that letting $h \equiv 1$ and $\|f\| =
    \max_{z \in \ZZ} |f(z)|$ in~\eqref{eq:sL} yields
	\begin{equation*}
		s(L) = \lim_{n \to \infty} \left(
		\max_{z\in \ZZ} \EE_z \prod_{t=1}^{n} M_t
		\right)^{1/n}.
	\end{equation*}
    Since $\me^{-\delta} s(L) <1$ by Corollary~\ref{cor:sL}, there exists $N \in
    \NN$ and $\alpha<1$ such that for all $n \geq N$,
	\begin{equation*}
		\me^{-\delta} \left(
		\max_{z\in \ZZ} \EE_z \prod_{t=1}^{n} M_t
		\right)^{1/n}
		= \left(
		\max_{z\in \ZZ} \EE_z \prod_{t=1}^{n} 
		\me^{-\delta} M_t
		\right)^{1/n} < \alpha. 
	\end{equation*}
	Hence, the first inequality holds, and the proof is now complete.
\end{proof}

To simplify notation, for given $\hat Y$, we denote
\begin{equation}
	\label{eq:simp_func}
	h(\xi, x,z) \coloneq \hat Y + \me^{-\delta} I(\xi,x,z) 
	\quad \text{and} \quad
	g(\zeta, x) \coloneq \min \left\{ \max \{\zeta - k, p(x) \}, p(b) \right\}.
\end{equation}
By definition, $\xi \mapsto h(\xi,x,z)$ and $\zeta \mapsto g(\zeta, x)$ are
increasing given $(x,z)$.

\begin{lemma}\label{lm:disc}
	For all $m \in \NN$, $(x,z) \in \SS$, and $\gamma \geq 0$, we have
	\begin{equation}
		\label{eq:disc}
		T^m (f + \gamma) (x,z) \leq 
		T^m f(x,z) + \gamma \, \EE_z \prod_{t=1}^m \me^{-\delta} M_t.
	\end{equation}
\end{lemma}

\begin{proof}
    Fix $f \in \cC$, $\gamma \geq 0$, and let $f_\gamma (x,z) \coloneq f(x,z) +
    \gamma$. By the definition of $T$, 
	\begin{align*}
		Tf_\gamma (x,z) &= g \left[
		\me^{-\delta} \EE_z \hat M 
		f_\gamma \left( h [Tf_\gamma (x,z), x,z], \hat Z \right), x
		\right] \\
		&\leq g \left[
		\me^{-\delta} \EE_z \hat M 
		f \left(h [Tf_\gamma (x,z), x,z], \hat Z \right), x
		\right] + \gamma \, \EE_z \me^{-\delta} \hat M \\
		&\leq 
		g \left[
		\me^{-\delta} \EE_z \hat M
		f \left(h[Tf (x,z), x,z], \hat Z \right), x
		\right] + \gamma \, \EE_z \me^{-\delta} \hat M,
	\end{align*}
        where the second inequality is due to the fact that $f \leq f_\gamma$ and
        $T$ is order preserving. Hence,
        $T(f+\gamma) (x,z) \leq Tf(x,z) + \gamma \, \EE_z \me^{-\delta} \hat M$
        and~\eqref{eq:disc} holds for $m=1$. Suppose~\eqref{eq:disc} holds for
        arbitrary $m$. It remains to show that it holds for $m+1$.  For
        $z\in \ZZ$, let
        $\ell (z) \coloneq \gamma \, \EE_z \prod_{t=1}^{m} \me^{-\delta} M_t$.  By the induction
        hypothesis, Lemma~\ref{lm:order_pres}, and the Markov property,
	\begin{align*}
		T^{m+1} f_\gamma (x,z)  
		&= g \left[ 
		\me^{-\delta} \EE_z \hat M (T^{m} f_\gamma) 
		\left(h[T^{m+1} f_\gamma (x,z), x,z], \hat Z \right), x
		\right]  \\
		&\leq g \left[
		\me^{-\delta} \EE_z \hat M (T^{m} f + \ell) 
		\left(h[T^{m+1} f_\gamma (x,z), x,z], \hat Z \right), x
		\right]  \\
		&\leq g \left[
		\me^{-\delta} \EE_z \hat M (T^{m} f) 
		\left(h[T^{m+1} f_\gamma (x,z), x, z], \hat Z \right), x
		\right]
		+ \EE_z \me^{-\delta} M_1 \ell (Z_1) \\
		&\leq T^{m+1} f(x,z) + \gamma \, \EE_z \me^{-\delta} M_1 
		\EE_{Z_1} \me^{-\delta} M_1 \cdots \me^{-\delta} M_m \\
		&= T^{m+1} f(x,z) + \gamma \, \EE_z \prod_{t=1}^{m+1} \me^{-\delta} M_t.
	\end{align*}
	Hence~\eqref{eq:disc} holds by induction. 
\end{proof}

\begin{lemma}\label{lm:discounting}
	There exist $n\in \NN$ and $\theta \in (0,1)$ such that
	\begin{equation*}
		T^n (f+\gamma) \leq T^n f +\theta \gamma
		\quad \text{for all }  f \in \cC \text{ and }\gamma \geq 0. 
	\end{equation*}
\end{lemma}

\begin{proof}
	By the first part of Lemma~\ref{lm:ctr}, there exist $n \in \NN$ and 
	$\alpha \in (0,1)$ such that 
	\begin{equation*}
		\EE_z \prod_{t=1}^n \me^{-\delta} M_t 
		< \alpha^n \quad \text{for all } z \in \ZZ.
	\end{equation*}
    Letting $\theta \coloneq \alpha^n$, we have $\theta < 1$. The stated claim then 
    follows from Lemma~\ref{lm:disc}.
\end{proof}

\begin{theorem}\label{t:opt_gen}
    If Assumption~\ref{a:opt} holds, then $T$ is well defined on the function
    space $\cC$, and there exists an $n \in \NN$ such that $T^n$ is a
    contraction mapping on $(\cC, \rho)$.  Moreover,
	\begin{enumerate}
		\item $T$ has a unique fixed point $f^*$ in $\cC$.
		\item The fixed point $f^*$ is the unique equilibrium pricing rule in $\cC$.
		\item For each $f$ in $\cC$, we have $\rho(T^kf, f^*)$ as $k \to \infty$.
	\end{enumerate}
\end{theorem}

\begin{proof}
    Proposition~\ref{pr:wel_def} shows that $T$ is a well-defined operator on
    $\cC$. Since in addition (i) $T$ is order preserving by 
    Lemma~\ref{lm:order_pres}, (ii) $\cC$ is closed under the addition of 
    non-negative constants, and (iii) there exist $n \in \NN$ and $\theta < 1$ 
    such that $T^n (f + \gamma) \leq T^n f + \theta \gamma$ for all $f \in \cC$ 
    and $\gamma \geq 0$ by Lemma~\ref{lm:discounting}, we have: $T^n$ is a 
    contraction mapping on $(\cC, \rho)$ of modulus $\theta$ based on
    \citet{blackwell1965discounted}. Claims~(i)--(iii) then follow from the 
    Banach contraction mapping theorem and the definition of the equilibrium 
    pricing rule.
\end{proof}

For each $f$ in $\cC$, we define
\begin{equation*}
	\bar p_f^0 (z) \coloneq \me^{-\delta} \EE_z \hat M 
	f(\hat Y, \hat Z)-k
	\quad \text{and} \quad
	\bar p_f(z) \coloneq \min \{\bar p_f^0 (z), p(b)\}.
\end{equation*}

\begin{lemma}\label{lm:threshold}
	For each $f$ in $\cC$, $Tf$ satisfies
	\begin{enumerate}
		\item $Tf(x,z) = p(x)$ if and only if $x \leq p^{-1} [\bar p_f(z)]$,
		\item $Tf(x,z) > p_0(x)$ if and only if $p^{-1} [\bar p_f(z)] < x < x_f^*(z)$, and
		\item $Tf(x,z) =0$ if and only if $x \geq x^*_f(z)$.
	\end{enumerate}
\end{lemma}

\begin{proof}
	Regarding claim~(i), suppose $Tf(x,z) = p(x)$. We show that 
	$x \leq p^{-1} [\bar p_f(z)]$. Note that in this case, 
	$x \leq p^{-1}(0) \leq x_f^*(z)$ since $p(x) = Tf(x,z) \geq 0$. 
	Hence,
	\begin{equation*}
		I \left[Tf(x,z), x, z \right] = x - p^{-1}[Tf(x,z)] = 0.
	\end{equation*}
	By the definition of $T$, we have
	\begin{align*}
		p(x) = Tf(x,z) &= \min\left\{
		\max\left\{
		\me^{-\delta} \EE_z \hat M f(\hat Y, \hat Z)-k, p(x) 
		\right\}, p(b)
		\right\} \\
		&\geq \min\left\{
		    \bar p_f^0(z), p(b)
		\right\}
		= \bar p_f(z).
	\end{align*}
	Since $p$ is decreasing, this implies $x \leq p^{-1} [\bar p_f(z)]$.
	
	Next, we prove that $x \leq p^{-1} [\bar p_f(z)]$ implies 
	$Tf(x,z) = p(x)$. If $\bar p_f^0(z) \geq p(b)$, then 
	\begin{equation*}
		\bar p_f(z) = p(b) \quad \Longrightarrow \quad
		x \leq p^{-1}[\bar p_f(z)] = p^{-1} [p(b)] = b.
	\end{equation*}
	Hence $x = b$. Then by definition $Tf(x,z) = \min\{\bar p_f^0(z), p(b)\} = p(b) = p(x)$. 
	
	If $\bar p_f^0(z) < p(b)$,
	then $\bar p_f (z) = \bar p_f^0(z)$. Since in addition
	\begin{equation*}
		\bar p_f^0 (z) \geq \me^{-\delta} 
		\EE_z \hat M p(\hat Y) - k >0
		\quad \text{and} \quad
		x \leq p^{-1} [\bar p_f(z)],
	\end{equation*}
	we have $x < p^{-1}(0)\leq x_f^*(z)$ in this case. 
	Suppose to the contrary that $Tf(x,z) > p(x)$ for some 
	$(x,z) \in \SS$. Then by the definition of $T$,
	\begin{align*}
		p(x) < \me^{-\delta} \EE_z \hat M f \left[
		\me^{-\delta} \left(
		x - p^{-1} [Tf(x,z)]
		\right) + \hat Y, \hat Z
		\right] - k.
	\end{align*}
	The monotonicity of $f$ in its first argument then 
	implies that 
	\begin{equation*}
        p(x) < \me^{-\delta} \EE_z \hat M f(\hat Y, \hat Z) - k = \bar p_f^0(z)
        = \bar p_f(z),
	\end{equation*}
	which is a contradiction. Claim~(i) is now verified.	
	
	Note that claim~(ii) follows immediately once claim~(iii) is 
	verified. To see that claim~(iii) is true, suppose to the 
	contrary that $x \geq x_f^*(x)$ and $Tf(x,z) >0$ for some 
	$(x,z) \in \SS$. Then
	\begin{equation*}
		I\left[Tf(x,z),x,z\right] = x^*_f(z) - p^{-1}[Tf(x,z)] 
		> x^*_f(z) - p^{-1}(0).
	\end{equation*} 
	By the definition of $x_f^*(z)$ and the monotonicity of $f$,
	this gives 
	\begin{equation*}
		\me^{-\delta} \EE_z \hat M f \left(
		\me^{-\delta} I \left[
		Tf(x,z), x,z
		\right] + \hat Y, \hat Z
		\right) - k 
		\leq 0.
	\end{equation*}
	Using the definition of $T$, we obtain
	$0 < Tf(x,z) \leq \min \{\max\{0,p(x)\}, p(b)\} = 0$,
	which is a contradiction. Hence, $x \geq x^*_f(z)$ 
	implies $Tf(x,z)=0$.
	
	Now suppose $Tf(x,z) = 0$. The definition of $T$ 
	implies that
	\begin{equation*}
		\me^{-\delta} \EE_z \hat M f \left(
		\me^{-\delta} I(0, x, z) + \hat Y, \hat Z
		\right) - k \leq 0.
	\end{equation*}
	By the definition of $x_f^*(z)$, this gives 
	$x \geq x_f^*(z)$. Claim~(iii) is now verified. 
\end{proof}

\begin{proof}[Proof of Theorem~\ref{t:opt}]
	Theorem~\ref{t:opt_gen} implies that there exists 
	a unique equilibrium pricing rule $f^*$ in $\cC$. 
	Claims~(i)--(iii) follow immediately from 
	Lemma~\ref{lm:threshold} since
		$\bar p(z) = \bar p_{f^*} (z)$
	and $f^*$ is the unique fixed point of $T$ in $\cC$.
	
	To see that claim~(iv) holds, suppose $f^*(x,z)$ 
	is not strictly decreasing in $x$ under the given conditions. 
	Then by claims~(i)--(iii), there exists $z \in \ZZ$ 
	with $\me^{-\delta} \EE_z \hat M < 1$ and a first 
	interval $[x_0, x_1] \subset 
	\left(p^{-1}[\bar p(z) ], x^*(z) \right)$ 
	such that $f^*(x,z) \equiv B$ on this interval 
	for some constant $B>0$. By the definition of $T$, 
	for all $x \in [x_0,x_1]$, 
	\begin{equation*}
		B = f^*(x,z) = \me^{-\delta} \EE_z \hat M f \left(
		\me^{-\delta} I\left[
		f^*(x,z), x,z
		\right] + \hat Y, \hat Z
		\right) - k.
	\end{equation*}
	Since the left-hand-side is a constant, 
	$f \left(
	\me^{-\delta} I\left[
	f^*(x,z), x,z
	\right] + \hat Y, \hat Z
	\right) = B'(\hat Z)$ 
	for some constant $B'(\hat Z)$. Moreover, $B'(\hat Z)\leq B$ since 
	$f$ is decreasing in $x$ and $[x_0,x_1]$ is the 
	first interval on which $f$ is constant in $x$. Since in 
	addition $\me^{-\delta} \EE_z \hat M < 1$, we have
	$B \leq \me^{-\delta} \EE_z \hat M B - k < B - k \leq B$,
	which is a contradiction. Hence claim~(iv) must be true.
\end{proof}

\begin{proof}[Proof of Proposition~\ref{pr:storage}]
    The continuity of $i^*$ and claims~(i)--(iii) follow from
    Theorem~\ref{t:opt} and the definition of $i^*$. We next show that
    $i^*(x,z)$ is increasing in $x$. Since $i^*(x,z)$ is constant given $z$ when
    $x \leq p^{-1}[\bar p(z)]$ and when $x \geq x^*(z)$ by claim~(i) and
    claim~(iii), it remains to show that $i^*(x,z)$ is increasing in $x$ when
    $p^{-1} [\bar p(z)] < x<x^*(z)$. In this case
		$i^*(x,z) = x - p^{-1} [f^*(x,z)]$.
	Suppose to the contrary that there exist $z \in \ZZ$ and 
	$x_1, x_2 \in \left(p^{-1} [\bar p(z)], x^*(z)\right)$ 
	such that $x_1 < x_2$ and $i^*(x_1,z) > i^*(x_2,z)$. 
	Then by definition,
	\begin{equation*}
		x_1 - p^{-1}[f^*(x_1,z)] > x_2 - p^{-1}[f^*(x_2,z)].
	\end{equation*}
	Since $x_1 < x_2$, this gives 
		$p^{-1}[f^*(x_2, z)] > p^{-1} [f^*(x_1,z)]$.
	But by (ii) of Theorem~\ref{t:opt} and the 
	definition of $T$, we obtain
	\begin{align*}
		f^*(x_1,z) &= \me^{-\delta} \EE_z \hat M f^* \left( 
		\me^{-\delta} i^*(x_1, z) + \hat Y, \hat Z 
		\right)-k  \\
		&\leq \me^{-\delta} \EE_z \hat M f^* \left( 
		\me^{-\delta} i^*(x_2,z) + \hat Y, \hat Z 
		\right)-k  
		\leq f^*(x_2,z),
	\end{align*}
	which implies $p^{-1}[f^*(x_1, z)] \geq p^{-1}[f^*(x_2,z)]$. 
	This is a contradiction. Hence, it is true that $i^*(x,z)$ is 
	increasing in $x$.
	
	It remains to verify claim~(iv). Pick any $z \in \ZZ$ 
	and $x_1, x_2 \in \left(p^{-1} [\bar p(z)], x^*(z)\right)$ 
	with $x_1 < x_2$. By claim~(iv) of Theorem~\ref{t:opt}, 
	we have $f^*(x_1,z) > f^*(x_2,z)$. Using the definition 
	of $T$ and claim~(ii) of Theorem~\ref{t:opt} again, 
	we have
	\begin{equation*}
		\EE_z \hat M f^* \left( 
		\me^{-\delta} i^*(x_1, z) + \hat Y, \hat Z 
		\right) > 
		\EE_z \hat M f^* \left( 
		\me^{-\delta} i^*(x_2, z) + \hat Y, \hat Z 
		\right).
	\end{equation*}
	The monotonicity of $f^*$ then gives $i^*(x_1,z)< i^*(x_2,z)$. 
	Hence claim~(iv) holds.
\end{proof}

\section{Proof of Section~\ref{sec:itheory} Results}\label{sec:pf_itheory}

A Markov chain $\{Z_t\}$ with transition matrix $F$ is called
\textit{monotone} if
\begin{equation*}
	\int h(\hat z) \diff F(z_1, \hat z) \leq 
	\int h(\hat z) \diff F(z_2, \hat z)
\end{equation*}
whenever $z_1 \leq z_2$ and $h: \ZZ \to \RR$ is bounded and increasing.
In proofs, when stating equality or inequality conditions between different 
random variables, we understand them as holding with probability one.

\begin{proof}[Proof of Proposition~\ref{pr:mono_z}]
	Let $\cC_1$ be the elements in $\cC$ such that $z \mapsto f(x,z)$ is decreasing
	for all $x$. Obviously, $\cC_1$ is a closed subset of $\cC$. Therefore, to show 
	that $z \mapsto f^*(x,z)$ is decreasing for all $x$, it suffices to verify 
	$T \cC_1 \subset \cC_1$.
	
	Fix $f \in \cC_1$ and $z_1, z_2 \in \ZZ$ with $z_1 \leq z_2$. Suppose there exists 
	an $x$ such that 
	\begin{equation}\label{eq:equal}
		\xi_1 \coloneq Tf(x, z_1) < Tf(x, z_2) =:\xi_2.
	\end{equation}
	Note that $f$ is a 
	decreasing function since $f\in \cC_1$. Moreover, by assumption $m(z, \epsilon) 
	= 1/r(z, \epsilon)$ is decreasing in $z$, $y(z, \eta)$ is increasing in $z$, and 
	$\Phi$ is monotone. Therefore, for all $\xi \in B(x)$, we have
	\begin{equation}\label{eq:mono}
		\EE_{z_1} \hat M f \left(
		\me^{-\delta} [x - p^{-1}(\xi)] +\hat Y, \hat Z
		\right)
		\geq \EE_{z_2} \hat M f \left(
		\me^{-\delta} [x - p^{-1}(\xi)] + \hat Y, \hat Z
		\right).
	\end{equation} 
	In particular, by the definition of $x_f^*$, we have $x_f^*(z_2) \leq x_f^*(z_1)$.
	If $x < x_f^*(z_2)$, then 
	\begin{equation*}
		I(\xi_1, x, z_1) = x - p^{-1}(\xi_1) \leq
		x - p^{-1}(\xi_2) = I(\xi_2, x, z_2).
	\end{equation*}
	Recall $\psi$ defined in~\eqref{eq:sol_euler}. The above inequality
        and~\eqref{eq:mono} imply that
	\begin{equation*}
		\xi_1 = \psi (\xi_1, x, z_1) \geq \psi (\xi_1, x, z_2) \geq 
		\psi(\xi_2, x, z_2) = \xi_2.
	\end{equation*}
	If $x \geq x_f^*(z_2)$, then we also have $\xi_1 \geq \xi_2$ since 
	$\xi_2 = 0$ and $\xi_1 \geq 0$. In either case, this is contradicted 
	with~\eqref{eq:equal}. Therefore, we have shown that $z \mapsto Tf(x,z)$ 
	is decreasing for all $x$ and $T \cC_1 \subset \cC_1$. It then follows 
	that $z \mapsto f^*(x,z)$ is decreasing for all $x$.
	
	To see that $i^*(x,z)$ is decreasing in $z$, pick any $z_1, z_2 \in \ZZ$ 
	with $z_1 \leq z_2$. By the definition of $x^*(z)$ and the monotonicity 
	of $f^*(x,z)$ in $z$, we have
	\begin{equation*}
		0 = f^*(x^*(z_1), z_1) \geq f^*(x^*(z_1), z_2)
	\end{equation*}
	and thus $x^*(z_1) \geq x^*(z_2)$. The definition of $i^*$ and the monotonicity of $p^{-1}$ and $f^*$ then implies that 
	\begin{align*}
		i^*(x,z_1) &= \min \{x, x^*(z_1)\} - p^{-1}[f^*(x,z_1)]  \\
		&\geq \min \{x, x^*(z_2)\} - p^{-1}[f^*(x,z_2)] = i^*(x,z_2).
	\end{align*}
	Hence $z \mapsto i^*(x,z)$ is decreasing for all $x$.
	
	Finally, note that
	$\hat Z \mapsto \hat M f^*(\hat Y, \hat Z) 
	= m(\hat Z, \hat \epsilon) f^*(y(\hat Z, \hat \eta), \hat Z)$ is decreasing 
	because $f^*$ is decreasing, $y$ is increasing in $z$, and $m$ is decreasing in $z$.
	Since in addition $\{Z_t\}$ is monotone, it follows immediately by definition that 
	$z \mapsto \EE_z \hat M f^*(\hat Y, \hat Z)$ is decreasing. 
	Hence $\bar p$ is decreasing by definition.
\end{proof}

Next, we discuss the correlation between commodity price and stochastic discount
factor. To state the result, we suppose $Z_t = (Z_{1t}, \ldots, Z_{nt})$ takes
values in $\RR^n$. The following is a simple corollary of the key result of
\citet{fortuin1971correlation}.

\begin{lemma}[Fortuin--Kasteleyn--Ginibre]\label{lm:fkg}
    If $f, g$ are decreasing integrable functions
    on $\RR^n$ and  $W = (W_1, \ldots, W_n)$ is a random vector on $\RR^n$
    such that $\{W_1, \ldots, W_n\}$ are independent, then 
	$\EE f(W) \EE g(W) \leq \EE f(W) g(W)$. 
\end{lemma}

Lemma~\ref{lm:fkg} implies that if $f$ is decreasing and $g$ is
nondecreasing (so that $-g$ is decreasing), then we have $\EE f(W) \EE g(W) \geq
\EE f(W) g(W)$. 

\begin{proposition}\label{pr:pos_corr} 
    If $m(z, \epsilon)$ is decreasing in $z$, $y(z,\eta)$ is
    nondecreasing in $z$, $\Phi$ is monotone, and $\{Z_{1t}, \ldots, Z_{nt}\}$
    are independent for each fixed $t$, then
	$\Cov_{t-1} (P_t, M_t ) \geq 0$ for all $t \in \NN$.
\end{proposition}

\begin{proof}
	The equilibrium path is
	 $X_t = \me^{-\delta} \, i^*(X_{t-1}, Z_{t-1}) + y(Z_t, \eta_t)$
	where
	\begin{equation*}
		i^*(X_{t-1}, Z_{t-1}) = \min\left\{X_{t-1}, x^*(Z_{t-1}) \right\}
		- p^{-1} \left[ f^*(X_{t-1}, Z_{t-1})\right].
	\end{equation*}
    Note that $X_t$ is a nondecreasing function of $Z_t$ since 
    $z \mapsto y(z, \eta)$ is nondecreasing for all $\eta$. 
    Moreover, the proof of Proposition~\ref{pr:mono_z} implies 
    that $f^*$ is a decreasing function under the assumptions 
    of the current proposition. 
    Hence, $Z_t \mapsto f^*(X_t, Z_t)$ 
    is decreasing. Since in addition 
    $z \mapsto m(z, \epsilon)$ is decreasing for all $\epsilon$ and 
    $\{Z_{1t}, \ldots, Z_{nt}\}$ are independent, applying Lemma~\ref{lm:fkg} 
    (taking $W = Z_t$) yields
	%
	\begin{align*}
		&\EE \left[
		f^*(X_t,Z_t) m(Z_t, \epsilon_t) 
		\mid X_{t-1}, Z_{t-1}, \epsilon_t, \eta_t 
		\right]  \\
		&\geq \EE \left[f^*(X_t,Z_t) \mid X_{t-1}, Z_{t-1}, \epsilon_t, \eta_t
		\right]
		\EE \left[m(Z_t, \epsilon_t)
		\mid X_{t-1}, Z_{t-1}, \epsilon_t, \eta_t
		\right]  \\
		&= \EE \left[f^*(X_t,Z_t) \mid X_{t-1}, Z_{t-1}, \eta_t
		\right]
		\EE \left[m(Z_t, \epsilon_t) \mid Z_{t-1}, \epsilon_t
		\right].
	\end{align*}
	Using this result, it follows that
	\begin{align*}
		&\EE \left(P_t M_t \mid X_{t-1}, Z_{t-1} \right)  
		= \EE \left[
		f^*(X_t,Z_t) m(Z_t, \epsilon_t) \mid X_{t-1}, Z_{t-1} 
		\right]  \\
		&= \EE \left(
		\EE \left[
		f^*(X_t,Z_t) m(Z_t, \epsilon_t) 
		\mid X_{t-1}, Z_{t-1}, \epsilon_t, \eta_t
		\right] 
		\mid X_{t-1}, Z_{t-1}
		\right) \\
		&\geq \EE \left\{
		\EE \left[f^*(X_t,Z_t) \mid X_{t-1}, Z_{t-1}, \eta_t
		\right]
		\EE \left[m(Z_t, \epsilon_t) \mid Z_{t-1}, \epsilon_t
		\right]
		\mid X_{t-1}, Z_{t-1}
		\right\}  \\
		&= \EE \left[f^*(X_t,Z_t) \mid X_{t-1}, Z_{t-1} \right]
		\EE \left[m(Z_t, \epsilon_t) \mid X_{t-1}, Z_{t-1} \right]  \\
		&= \EE \left(P_t \mid X_{t-1}, Z_{t-1} \right) 
		\EE \left(M_t \mid X_{t-1}, Z_{t-1} \right),
	\end{align*}
	where the second-to-last equality holds because $\eta_t$ is independent 
	of $\epsilon_t$. Hence, 
	\begin{align*}
		&\Cov_{t-1} (P_t, M_t)
		=\Cov(P_t, M_t \mid X_{t-1}, Z_{t-1})  \\
		&=\EE \left(P_t, M_t \mid X_{t-1}, Z_{t-1} \right) - 
		\EE \left(P_t \mid X_{t-1}, Z_{t-1} \right) 
		\EE \left(M_t \mid X_{t-1}, Z_{t-1} \right) \geq 0,
	\end{align*}
	as was to be shown. 
\end{proof}

\begin{proof}[Proof of Proposition~\ref{pr:neg_corr}]
  Since $R_t = 1/M_t$, applying Lemma~\ref{lm:fkg} again
  and working through similar steps to the proof of
  Proposition~\ref{pr:pos_corr}, we can show that $\Cov_{t-1}(P_t, R_t) \leq 0$
  for all $t$. The details are omitted.
\end{proof}

\begin{proof}[Proof of Proposition~\ref{pr:mono_seq}]
	Let $T_1$ and $T_2$ be respectively the equilibrium price operators 
	corresponding to $\{R_t^1\}$ and $\{R_t^2\}$. It suffices to show that
	$T_1 f \leq T_2 f$ for all $f \in \cC$. To see this, we adopt an induction 
	argument. Suppose $T_1^k f \leq T_2^k f$. Then by the order preserving 
	property of the equilibrium price operator and the initial argument 
	$T_1 f \leq T_2 f$ for all $f$ in $\cC$, we have
	$T_1^{k+1} f = T_1(T_1^k f) \leq T_1 (T_2^k f) \leq T_2 (T_2^k f) = T_2^{k+1}f$.
	Hence, $T_1^k f \leq T_2^k f$ for all $k \in \NN$ and $f \in \cC$. 
	Letting $k \to \infty$ then yields $f_1^* \leq f_2^*$.
	
	We now show that $T_1 f \leq T_2 f$ for all $f \in \cC$. Suppose there exists
        $(x,z) \in \SS$ such that $\xi_1 \coloneq T_1f(x,z) > T_2 f(x,z) =: \xi_2$.  Let
        $M_t^i = 1 / R_t^i$ for $i=1,2$. Since $R_t^1 \geq R_t^2$, we have
        $M_t^1 \leq M_t^2$. Letting $\hat M_i$ be the next-period discount factor
        of economy $i$, the monotonicity of $g$ and $h$ defined
        in~\eqref{eq:simp_func} then implies that
	\begin{align*}
		\xi_1 &= g \left[
		\me^{-\delta} \EE_z \hat M_1 
		f \left( h(\xi_1,x), \hat Z \right), x
		\right]  \\
		&\leq g \left[
		\me^{-\delta} \EE_z \hat M_2
		f \left( h(\xi_1,x), \hat Z \right), x
		\right]  
		\leq g \left[
		\me^{-\delta} \EE_z \hat M_2
		f \left( h(\xi_2,x), \hat Z \right), x
		\right] 
		= \xi_2,
	\end{align*}
	which is a contradiction. Therefore, $T_1 f \leq T_2 f$ and all the stated 
	claims hold.
\end{proof}

\begin{proof}[Proof of Proposition~\ref{pr:caus_1}]
	In this case, the exogenous state is $Z_t = R_t$, which is a monotone 
	Markov process, and $r(z, \epsilon) \equiv z$ is strictly increasing in 
	$z$. Since $R_{t}^2 \leq R_{t}^1$, $y$ is nondecreasing in $R$, and both 
	economies share the same innovation process $\{\eta_t\}$, we have 
	$Y_{t}^{2} \leq Y_{t}^{1}$. Since in addition $X_{t-1}^2 \leq X_{t-1}^1$ 
	and $R_{t-1}^1 \leq R_{t-1}^2$, Propositions~2.1 and 3.2 and the law of 
	motion of the state process imply that
	\begin{equation*}
		X_t^2 = \me^{-\delta} i^*(X_{t-1}^2, R_{t-1}^2) + Y_t^2 
		\leq \me^{-\delta} i^*(X_{t-1}^1, R_{t-1}^1) + Y_t^1 = X_t^1
	\end{equation*}
	with probability one. Applying Proposition~3.2 again and the monotonicity 
	of $f^*$ with respect to the endogenous state yields
	\begin{equation*}
		P_t^1 = f^*(X_t^1, R_t^1) \leq 
		f^*(X_t^1, R_t^2) \leq 
		f^*(X_t^2, R_t^2) = P_t^2
	\end{equation*}
	with probability one. Hence the statement is verified.
\end{proof}

\begin{proof}[Proof of Proposition~\ref{pr:caus_2}]
	We have seen that, in this case, the exogenous state is $Z_t = R_t$, which 
	is a monotone Markov process, and $r(z, \epsilon) \equiv z$ is strictly 
	increasing in $z$. Since $R_t \leq R_{t+1}$ and $y$ is nondecreasing, we 
	have $Y_{t} \leq Y_{t+1}$. Since in addition $X_{t-1} \leq X_t$ and 
	$R_{t-1} \geq R_t$, applying Propositions~2.1 and 3.2, the law of motion of 
	the state process indicates that
	\begin{equation*}
		X_{t} = \me^{-\delta} i^*(X_{t-1}, R_{t-1}) + Y_t  
		\leq \me^{-\delta} i^*(X_{t}, R_{t}) + Y_{t+1} = X_{t+1}
	\end{equation*}
	with probability one. Then Proposition~3.2 and the monotonicity of $f^*$ 
	with respect to the endogenous state imply that
	\begin{equation*}
		P_{t} = f^*(X_{t}, R_{t}) \geq 
		f^*(X_{t}, R_{t+1}) \geq 
		f^*(X_{t+1}, R_{t+1}) = P_{t+1}
	\end{equation*}
	with probability one. Hence the stated claim holds.
\end{proof}

\section{Positive Correlation}\label{sec:counter_ex}

Here we provide examples showing that Proposition~\ref{pr:neg_corr} does not
hold in general if $Z_t$ is positively or negatively
correlated across dimensions. We begin with the following.

\begin{proposition}\label{pr:conv}
  If Assumption~\ref{a:opt} holds and the inverse demand function is $p(x) = a
  + dx$ with $a>0$ and $d< 0$, then the equilibrium pricing rule $f^*(x,z)$
  is convex in $x$.
\end{proposition}

\begin{proof}
  Let $\cC_2$ be the elements in $\cC$ such that $x \mapsto f(x,z)$ is convex
  for all $z \in \ZZ$. Then $\cC_2$ is a closed subset of $\cC$. Hence it
  suffices to show that $T \cC_2 \subset \cC_2$. Fix $f \in \cC_2$, since $Tf
  \in \cC$ by Proposition~\ref{pr:selfmap}, it remains to show that $Tf(x,z)$ is
  convex in $x$. Since $Tf(x,z)$ is decreasing in $x$ and, by
  Lemma~\ref{lm:threshold}, $Tf(x,z)$ is linear in $x$ when $x \leq p^{-1}[\bar
  p(z)]$ or $x \geq x^*_f(z)$, it suffices to show that $Tf(x,z)$ is convex in
  $x$ on $B_0(z)\coloneq (p^{-1}[\bar p(z)], x_f^*(z) )$. In this case,
  \begin{equation*}
    Tf(x,z) = \me^{-\delta} \EE_z \hat M f \left(
      \me^{-\delta} \left(x - p^{-1}[Tf(x,z)] \right) + \hat Y, \hat Z
    \right) - k.
  \end{equation*}
  Suppose to the contrary that $Tf(x,z)$ is not convex, then there exist $z\in
  \ZZ$, $x_1,x_2 \in B_0(z)$, and $\alpha \in [0,1]$ such that, letting $x_0 \coloneq
  \alpha x_1 + (1-\alpha) x_2$,
  \begin{align*}
    &\alpha Tf(x_1,z) + (1-\alpha) Tf(x_2,z) < Tf(x_0, z) \\
    &= \me^{-\delta} \EE_z \hat M f \left(
      \me^{-\delta} \left(x - p^{-1}[Tf(x_0,z)] \right) + \hat Y, \hat Z
      \right) - k  \\
    &\leq \me^{-\delta} \EE_z \hat M f \left(
      \me^{-\delta} \left(x - p^{-1} \left[
      \alpha Tf(x_1,z) + (1-\alpha) Tf(x_2,z)
      \right] \right) + \hat Y, \hat Z
      \right) - k  \\
    &\leq \alpha Tf(x_1,z) + (1 - \alpha) Tf(x_2,z),
  \end{align*}
  where the last inequality is by convexity of $f(x,z)$ in $x$ and the
  linearity of $p(x)$. This is a contradiction. Hence $Tf(x,z)$ is convex in
  $x$ on $B_0(z)$ and the stated claim holds.
\end{proof}

Suppose $R_t = 0.98$ with probability $0.5$ and $R_t = 1.02$
with probability $0.5$. If $R_t = 0.98$, then $Y_t = y_0$ with
probability one, and if $R_t = 1.02$, then $Y_t = y_1$ with probability $\phi$
and $Y_t = y_2$ with probability $1-\phi$. This is a special case of our framework. In particular,
\begin{align*}
	\epsilon_{t} = \eta_t = 0, &\quad
	Z_{t} = (Z_{1t}, Z_{2t}) = (R_t, Y_t),  \\
	r(Z_t, \epsilon_{t}) = r(R_t, Y_t, \epsilon_t) = R_t 
	&\quad \text{ and } \quad
	y(Z_t, \eta_t) = y(R_t, Y_t, \eta_t) = Y_t.
\end{align*}
Note that $\{Z_t\}$ is {\sc iid}. Hence, it is naturally monotone and the
equilibrium pricing rule is not a function of $Z_t$.  Since in addition $r(z,
\epsilon)$ and $y(z,\eta)$ are increasing in $z$, the assumptions of
Proposition~\ref{pr:mono_z} hold. However, because $Z_{1t}$ and $Z_{2t}$ (i.e.,
$R_t$ and $Y_t$) are correlated, the assumptions of
Proposition~\ref{pr:neg_corr} are violated.

Some simple algebra shows that $\EE R_t = 1$,
\begin{equation*}
	\EE Y_t = \frac{y_0}{2} + \frac{\phi y_1}{2} + \frac{(1-\phi) y_2}{2} 
	\quad \text{and} \quad 
	\EE R_t Y_t = \frac{0.98 y_0}{2} + \frac{1.02\phi y_1}{2} + \frac{1.02(1-\phi) y_2}{2}.
\end{equation*}
Hence 
	$\Cov_{t-1} (R_t,Y_t) 
	= \EE R_t Y_t - \EE R_t \EE Y_t 
	= -0.01 y_0 + 0.01 \phi y_1 + 0.01(1-\phi) y_2$.
Choose $\delta$ such that $\delta > \log \EE (1/R_t) \approx 0.0004$. Then
$\beta \coloneq \me^{-\delta} \EE(1/R_t) < 1$ and the following result holds under the
current setup.

\begin{lemma}\label{lm:bind_ex}
    Either (i) $P_t = 0$ for all $t$ or (ii) $I_t = 0$ in finite time with
    probability one. If the
    per-unit storage cost $k>0$, then (ii) holds.
\end{lemma}

\begin{proof}
    Suppose (ii) does not hold, then $I_t > 0$ for all $t$ with positive
    probability, and the equilibrium price path satisfies
	\begin{equation}\label{eq:bind_equiv}
		P_0 
		\leq \me^{-\delta t} \EE \left( \prod_{i=0}^t \frac{1}{R_i} \right) P_t 
		- \left(\sum_{i=0}^{t-1} \me^{-\delta i} \right) k
		\quad \text{for all }\, t.
	\end{equation}
    Note that $\{P_t\}$ is bounded with probability one since
    $f^* \in \cC$ implies that, for some $L_0< \infty$, we have
    $P_t = f^*(X_t) \leq f^*(Y_t) \leq f^*(\underline y) 
    \leq p_0(\underline y) + L_0 =: L_1 < \infty$
    with probability one, where $\underline y \coloneq \min\{y_0, y_1, y_2\}$. 
    If in addition (i) does not hold, we may assume $P_0>0$ 
    without loss of generality. In this case,~\eqref{eq:bind_equiv} 
    implies that, when $t$ is sufficiently large, we have 
    $0 < P_0 \leq \beta^t L_1 < P_0 $ with positive probability, 
    which is a contradiction. Hence either (i) or (ii) holds. 
    If, on the other hand, $k>0$ and (ii) does not hold, then for 
    sufficiently large $t$,~\eqref{eq:bind_equiv} implies that $P_0 < P_0$ with positive
    probability for all $P_0\geq 0$, which is also a contradiction.
    Hence (ii) holds and the second claim is also verified.
\end{proof}

Consider a linear inverse demand function $p$ as in Proposition~\ref{pr:conv}.
If $P_t = 0$ for all $t$, then $\Cov (R_t,P_t) = \Cov_{t-1}(R_t,P_t) = 0$ for
all $t$. Otherwise, $I_{t-1} = 0$ for some finite $t$, in which case $X_t =
Y_t$, $P_t = f^*(Y_t)$, and thus
\begin{multline*}
	\Cov_{t-1} (R_t, P_t) 
    = \EE_{t-1} R_t P_t - \EE R_t \, \EE_{t-1} P_t  
	= \EE R_t f^*(Y_t) - \EE R_t \, \EE f^*(Y_t)   \\
	= -0.01 f^*(y_0) + 0.01 \phi f^*(y_1) + 0.01(1-\phi) f^*(y_2). 
\end{multline*}
If $y_1, y_2 < y_0$, then
    $\Cov_{t-1} (R_t, Y_t) < 0$ and $\Cov_{t-1} (R_t, P_t) > 0$
    based on the monotonicity of $f^*$. If on the other
    hand $y_0, y_1$ and $y_2$ satisfy\footnote{Since $\me^{-\delta} \EE (1/R_t)
        < 1$, we have $\bar p = \min\{\me^{-\delta} \EE f^*(\hat Y) / \hat R -
        k, p(b)\} < p(b)$. Thus $p^{-1}(\bar p) > b$ and this choice of $y_0,
    y_1, y_2$ is feasible.}
	$y_1 < p^{-1} (\bar p) < y_0 < y_2$,
then since $f^*$ is convex by Proposition~\ref{pr:conv}, and $f^*(x) > p(x)$
whenever $x > p^{-1}(\bar p)$ by Theorem~\ref{t:opt}, 
\begin{equation*}
	\frac{f^*(y_0)-f^*(y_2)}{f^*(y_1) - f^*(y_2)} 
	< \frac{y_2 - y_0}{y_2 - y_1}.
\end{equation*}  
Hence $\phi$ can be chosen such that 
$y_2 - y_0 > \phi (y_2 - y_1)$ and $f^*(y_0) - f^*(y_2) < \phi [f^*(y_1) -
f^*(y_2)]$.
In particular, the above inequalities respectively imply that 
\begin{equation*}
	\Cov_{t-1} (R_t, Y_t) > 0
	\quad \text{and} \quad
    \Cov_{t-1} (R_t, P_t) > 0. 
\end{equation*}

\section{An Identification Equivalence Result}\label{sec:equiv}

Consider an economy $E$ with linear inverse demand function $p(x) = a + d x$
where $a > 0$ and $d < 0$. Let $\{Y_t\}$ be a stationary Markov process with
transition probability $\Psi$. Let $b$ be the lower bound of the total available
supply in this economy.

Let $\tilde E$ be another economy where the output process satisfies 
$\tilde Y_t = \mu + \sigma Y_t$ with $\sigma > 0$
and the transition probability of $\{\tilde Y_t\}$
satisfies\footnote{Condition~\eqref{eq:Psi_ti} obviously holds if, for example,
	$\{Y_t\}$ is {\sc iid} and follows a truncated normal distribution with mean
	$\mu_0$, variance $\sigma_0^2$, and truncation thresholds $y_l<y_u$.  Because in
	this case, $\{\tilde Y_t\}$ is {\sc iid} and follows a truncated normal
	distribution as well, with mean $\mu + \sigma \mu_0$, variance $\sigma^2
	\sigma_0^2$, and truncation thresholds $\mu + \sigma y_l<\mu + \sigma y_u$.
	Note that~\eqref{eq:Psi_ti} does not hold if $\{Y_t\}$ and $\{\tilde Y_t\}$ do
	not follow the same \textit{type} of distribution. For example, it does not hold
	if $\{Y_t\}$ is {\sc iid} lognormally distributed, since $\{\tilde Y_t\}$ is not
	lognormally distributed as a linear transform of $\{Y_t\}$.}
\begin{equation}\label{eq:Psi_ti}
	\tilde \Psi (y, \hat Y) = \Psi \left(
	\frac{y- \mu}{\sigma}, \frac{\hat Y-\mu}{\sigma}
	\right).
\end{equation}
Moreover, let the lower bound of the total available supply of economy $\tilde E$ be
$\tilde b = \mu + \sigma b$
and the inverse demand function be
\begin{equation}\label{eq:p_ti}
	\tilde p(x) = \left(
	a - \frac{d \mu}{\sigma}
	\right) + \frac{d}{\sigma} x.
\end{equation}
The remaining assumptions are the same across economies $E$ and
$\tilde E$.

\begin{proposition}\label{pr:iequiv}
	$\tilde E$ and $E$ generate the same commodity price process.
\end{proposition} 

\begin{proof}
	To simplify notation, let $f$ and $i$ be the equilibrium pricing function
	and the equilibrium inventory function of the baseline economy $E$.  Without
	loss of generality, we may assume $Z_t = Y_t$.\footnote{In general, $Z_t$ is
		a multivariate Markov process and $Y_t$ corresponds to one dimension of
		$Z_t$.}  Then for all $(x,y) \in \SS$, $\{f(x,y), i(x,y)\}$ is the unique
	solution to
	\begin{align}\label{eq:sol_equiv1}
		f(x,y) &= \min \left\{
		\max \left\{
		\me^{-\delta} \EE_y \hat M f[\me^{-\delta} i(x,y) + \hat Y, \hat Y] - k,\, p(x)
		\right\}, p(b)
		\right\} \\
		\label{eq:sol_equiv2}
		i(x,y) &= 
		\begin{cases}
			x - p^{-1}[f(x,y)], & x < x_f(y)\\
			x_f(y) - p^{-1}(0), & x\geq x_f(y)
		\end{cases}
	\end{align} 
	where 
	\begin{equation*}
		x_f(y) \coloneq \inf \left\{x \geq p^{-1}(0): f(x,y) = 0 \right\}.
	\end{equation*}
	Consider economy $\tilde E$, where all magnitudes are denoted with tildes. Let
	\begin{align} \nonumber
		\tilde x = \mu + \sigma x,
		&\quad
		\tilde y = \mu + \sigma y,  \\ 
		\label{eq:def}
		\tilde m(\tilde y, \epsilon) = m(y, \epsilon), 
		\quad 
		\tilde f(\tilde x, \tilde y) &= f(x,y),
		\quad
		\tilde \imath (\tilde x, \tilde y) = \sigma \,i(x,y). 
	\end{align}
	To prove the statement of the proposition, it suffices to show that $\{
	\tilde f (\tilde x, \tilde y), \tilde \imath (\tilde x, \tilde y) \}$ is the
	unique solution to
	\begin{align}\label{eq:s_eq1}
		\tilde f(\tilde x, \tilde y) 
		&= \min \left\{
		\max \left\{
		\me^{-\delta} \EE_{\tilde y} \hat{\tilde{M}} \tilde f[\me^{-\delta} \tilde \imath(\tilde x, \tilde y) + \hat{\tilde Y}, 
		\hat{\tilde Y}], \tilde p(\tilde x)
		\right\} - k, \tilde p(\tilde b)  
		\right\} \\
		\label{eq:s_eq2}
		\tilde \imath(\tilde x, \tilde y) 
		&= 
		\begin{cases}
			\tilde x - \tilde p^{-1}[\tilde f(\tilde x,\tilde y)], 
			& \tilde x > x_{\tilde f}(\tilde y) \\
			x_{\tilde f}(\tilde y) - \tilde p^{-1}(0), 
			& \tilde x\leq x_{\tilde f}(\tilde y)
		\end{cases}
	\end{align} 
	where 
	\begin{equation*}
		x_{\tilde f}(\tilde y) \coloneq \inf \left\{
		\tilde x \geq \tilde p^{-1}(0): \tilde f(\tilde x,\tilde y) = 0 
		\right\}.
	\end{equation*}
	This is true by referring to~\eqref{eq:sol_equiv1}--\eqref{eq:sol_equiv2}.
	In particular, by~\eqref{eq:p_ti} and~\eqref{eq:def},
	\begin{equation*}
		\hat{\tilde{M}} = \tilde m(\hat{\tilde{Y}}, \hat \epsilon) 
		= m(\hat{Y}, \hat \epsilon) = \hat M, 
		\quad
		\tilde p(\tilde x) = p(x)
		\quad \text{and} \quad 
		\tilde p(\tilde b) = p(b).
	\end{equation*}
	Furthermore, $\tilde \Psi(\tilde y, \hat{\tilde Y}) = \Psi (y, \hat Y)$ by
	the definition in~\eqref{eq:Psi_ti} and 
	\begin{align*}
		\tilde f \left[
		\me^{-\delta} \tilde \imath (\tilde x, \tilde y) + \hat{\tilde Y}
		, \hat{\tilde Y}
		\right] 
		&= \tilde f\left[
		\me^{-\delta} \sigma i(x,y) + \mu + \sigma \hat Y, \hat{\tilde{Y}}
		\right] \\ 
		&= \tilde f \left[
		\mu + \sigma \left(\me^{-\delta} i(x,y) + \hat Y\right), \hat{\tilde{Y}}
		\right] 
		= f \left[
		\me^{-\delta} i(x,y) + \hat Y, \hat Y
		\right].
	\end{align*}
	The above analysis implies that~\eqref{eq:s_eq1} holds. To see that~\eqref{eq:s_eq2} holds, note that
	\begin{align*}
		x_{\tilde f} (\tilde y) &= \inf \left\{
		\mu + \sigma x \geq \mu - \frac{a\sigma}{d}:
		f(x,y) = 0
		\right\} \\
		&= \mu + \sigma \inf \left\{
		x \geq p^{-1}(0) : f(x,y) = 0
		\right\}
		= \mu + \sigma x_f (y),
	\end{align*}
	where we have used the definition of $p$ and $\tilde p$. This yields
	$\tilde x < x_{\tilde f} (\tilde y)$ iff $x < x_f(y)$.
	In combination with~\eqref{eq:sol_equiv2}, we obtain
	\begin{equation*}
		\tilde \imath(\tilde x, \tilde y) = \sigma i(x,y) = 
		\begin{cases}
			\sigma \left(x - p^{-1}[f(x,y)]\right),
			& \tilde x < x_{\tilde f}(\tilde y),  \\
			\sigma \left(x_f(y) - p^{-1}(0)\right),
			& \tilde x \geq x_{\tilde f}(\tilde y).
		\end{cases}
	\end{equation*}
	When $\tilde x < x_{\tilde f} (\tilde y)$, using~\eqref{eq:def} and the definition of $p$ and $\tilde p$, we obtain
	\begin{align*}
		\sigma \left(x - p^{-1}[f(x,y)]\right) 
		&= \sigma x - \sigma p^{-1}[\tilde f(\tilde x, \tilde y)]  \\
		&= \tilde x - \mu - \sigma \left( 
		\frac{\tilde f(\tilde x, \tilde y) - a}{d} 
		\right)
		= \tilde x - \tilde p^{-1} [\tilde f (\tilde x, \tilde y)].
	\end{align*}
	When $\tilde x \geq x_{\tilde f} (\tilde y)$, using the definition of $p$ and $\tilde p$ again yields
	\begin{align*}
		\sigma \left(x_f(y) - p^{-1}(0)\right)
		= x_{\tilde f} (\tilde y) - \mu + \frac{a\sigma}{d} 
		= x_{\tilde f} (\tilde y) - \tilde p^{-1}(0).
	\end{align*}
	The above analysis implies that~\eqref{eq:s_eq2} holds. Therefore, economies
	$E$ and $\tilde E$ generate the same commodity price process.
\end{proof}

\section{Algorithms}\label{s:alg}

The storage model is solved by a modified version of the endogenous
grid method of \citet{carroll2006the}. The candidate space $\cC$ and $\bar{p}_f(z)$ are as defined in Appendix~\ref{sec:proof-sect-refs}. We derive the
following property in order to handle free-disposal and state-dependent
discounting in the numerical computation.

\begin{lemma}\label{lm:egm}
	For each $f$ in the candidate space $\cC$, if
	$x > p^{-1} [\bar p_f(z)]$, then
	\begin{equation*}
		Tf(x,z) = \max\left\{
		\min \left\{
		\me^{-\delta} \EE_z \hat M f\left( 
		\me^{-\delta} \left(x - p^{-1}[Tf(x,z)] \right) + \hat Y, \hat Z
		\right) - k, p(b)
		\right\}, 0
		\right\}.
	\end{equation*}
	If in addition $\bar p_f^0(z) \leq p(b)$, then 
	\begin{equation*}
		Tf(x,z) = \max \left\{
		\me^{-\delta} \EE_z \hat M f\left( 
		\me^{-\delta} \left(x - p^{-1}[Tf(x,z)] \right) + \hat Y, \hat Z
		\right) - k, 0
		\right\}.
	\end{equation*}
\end{lemma}

\begin{proof}
	The first statement is immediate by Lemma~\ref{lm:threshold} (ii)--(iii) and
	the fact that $Tf(x,z)$ is decreasing in $x$. The second statement follows 
	from the definition of $\bar p_f^0 (z)$ and the monotonicity of $f$ 
	in its first argument as a candidate in $\cC$.
\end{proof}

\subsection{The Endogenous Grid Algorithm}

We define a finite Markov Chain $\{Z_t\}$.\footnote{In the model with only the
	speculative channel, we discretize the interest rate process following
	\citet{tauchen1986finite} and use a Markov chain with $N=101$ states. In the model with the demand channel, we discretize
	the VAR model representing the joint dynamics of interest rate and economic
	activity following \citet{Schm14a}.} The states are indexed by $j$ and $m$,
and the transition matrix has elements $\Phi_{j,m}$. Moreover, we use
$\mathcal{D}(x,z) = p^{-1} [f(x, z)]$ to denote a candidate equilibrium demand
function. The endogenous grid algorithm for computing the equilibrium pricing
rule is described in Algorithm~\ref{alg:egm}.

\begin{algorithm}
	\caption{\,The endogenous grid algorithm}\label{alg:egm}
	\small 
	\begin{enumerate}[label=Step \arabic*., font={\bfseries}, leftmargin=*, ref=\arabic*, itemsep=1mm]
		\item Initialization step. Choose a convergence criterion $\varpi > 0$, a grid
		on storage $\left\{I_{s}\right\}$ starting at 0, a grid on discount factor
		shocks for numerical integration $\{\epsilon_l\}$ with associated weights
		$\omega_l$, a grid on production shocks for numerical integration
		$\{\eta_{n}\}$ with associated weights $w_{n}$, and an initial policy rule
		(guessed): $\{X_{s,j}^{1}\}$ and $\{P_{s,j}^{1}\}$. Start iteration at
		$i=1$.
		\item\label{item:2} Update the demand function via linear interpolation and extrapolation:
		\begin{equation}
			\label{eq:3}
			p^{-1}\left(P^{i}_{s,j}\right)=\mathcal{D}^{i}\left(X^{i}_{s,j},Z_{j}\right).
		\end{equation}
		\item\label{item:3} Obtain prices and availability consistent with the grid of stocks and
		Markov Chain: 
		\begin{equation}
			\label{eq:2}
			P^{i+1}_{s,j}=\max\left\{ 
			\min\left\{\tilde P_{s,j}^{i+1}, \, p(b) \right\},
			\, 0
			\right\} 
			\quad \text{and} \quad 
			X^{i+1}_{s,j}=I_{s}+p^{-1}\left(P^{i+1}_{s,j}\right),
		\end{equation}
		where
		\begin{equation*}
			\tilde P_{s,j}^{i+1} =
			\me^{-\delta}\sum_{l,m,n}\omega_{l}\Phi_{j,m}w_{n}m\left(Z_m,\epsilon_l\right)p\left(\mathcal{D}^{i}\left(y\left(Z_m,\eta_n\right)+\me^{-\delta}I_{s},Z_{m}\right)\right) - k.
		\end{equation*}
		\item\label{item:4} Terminal step. If
		$\max_{s,j}|P_{s,j}^{i+1}-P_{s,j}^{i}|\ge \varpi $ then increment $i$ to
		$i+1$ and go to step~\ref{item:2}. Otherwise, approximate the equilibrium pricing rule by
		$f^*(x, z) = p [\mathcal{D}^{i}(x, z)]$.
	\end{enumerate}
\end{algorithm}

In particular, we choose to approximate the demand function $\mathcal{D}(x,z)$
in Step~\ref{item:2} instead of the price function $f(x,z)$. This is helpful for
improving both precision and stability of the algorithm when the demand function
diverges at the lower bound of the endogenous state space. A typical example is
the exponential demand $p(x) = x^{-1/\lambda} (\lambda>0)$, which is commonly
adopted by applied research \citep[see, e.g.,][]{deaton1992on, Goue22a}. If the
inverse demand function is linear as in Section~\ref{sec:quant}, however, then
it is innocuous to approximate the price function directly.

Moreover, the validity and convergence of the updating process in
Step~\ref{item:3} are justified by Theorem~\ref{t:opt_gen},
Lemma~\ref{lm:threshold}, and Lemma~\ref{lm:egm} above.

\subsection{Solution Precision}\label{sec:sp}

To evaluate the precision of the numerical solution, we refer to a suitably
adjusted version of the bounded rationality measure originally designed by
\citet{judd1992projection}, which we name as the Euler equation error and
measures how much solutions violate the optimization conditions. In the current
context, it is defined at state $(x,z)$ as
\begin{equation*}
	EE_f(x,z) = 1 - \frac{\mathcal D_1(x,z)}{\mathcal D_2 (x,z)},
\end{equation*}
where $f$ is the numerical solution of the equilibrium price,
\begin{equation*}
	\mathcal D_1(x,z) = p^{-1} \left[
	\min \left\{
	\max \left\{
	\me^{-\delta} \EE_z \hat M f(\hat X, \hat Z) - k,
	p(x)
	\right\}, p(b)
	\right\}
	\right] - b
\end{equation*}
and $\mathcal D_2(x,z) = p^{-1} [f(x,z)] - b$.
In particular, both $\mathcal D_1(x,z)$ and $\mathcal D_2(x,z)$ are expressed in
terms of the \textit{relative} demand for commodity, since $b$ is the greatest
lower bound (hence corresponds to the \textit{zero} level) of the total
available supply. Therefore, $EE_f(x,z)$ measures the error at state $(x,z)$, in
terms of the quantity consumed, incurred by using the numerical solution instead
of the true equilibrium pricing rule.

To evaluate the precision of the endogenous grid algorithm in the 
context of Section~\ref{sss:speculative-channel}, we simulate a time series $\{(X_t, R_t)\}_{t=1}^T$ of length $T=$20,000
based on the state evolution path
$X_{t+1} = \me^{-\delta} i(X_t, R_t) + Y_t$ and $R_{t+1} \sim \Phi (R_t,
\cdot)$,
where $(X_0,R_0)$ is given, and
$$i(X_t,R_t) = \min\{X_t, x_f^*(R_t)\} - p^{-1}[f(X_t,R_t)]$$
is the equilibrium storage function computed by the
endogenous grid algorithm. We discard the first 1,000 draws, and then compute
the Euler equation error at the truncated time series. When applying the
endogenous grid algorithm, we use an exponential grid for storage in the range
$[0,2]$ with median value $0.5$, function iteration is implemented via linear
interpolation and linear extrapolation, and we terminate the iteration process
at precision $\varpi = 10^{-4}$. The rest of the setting is same to
Section~\ref{sec:quant}.

\begin{table}[!htb]
	\begin{threeparttable}
		\caption{Precision under different grid sizes and different parameters}\label{tab:prec}
		\footnotesize
		\begin{tabular*}{\textwidth}{@{}l@{\extracolsep{\fill}}*{9}{d{-2}}@{}}
			\toprule
			\multicolumn{10}{@{}l@{}}{\textbf{A.} Different grid sizes}\\
			& \multicolumn{3}{c}{$K=$ 100} & \multicolumn{3}{c}{$K=$ 200} & \multicolumn{3}{c@{}}{$K=$ 1,000}\\
			\cmidrule(r){2-4} \cmidrule(lr){5-7} \cmidrule(l){8-10}
			Precision & \mcone{$N=7$} & \mcone{$N=51$} & \mcone{$N=101$} & \mcone{$N=7$} & \mcone{$N=51$} & \mcone{$N=101$}& \mcone{$N=7$} & \mcone{$N=51$} & \multicolumn{1}{c@{}}{$N=101$}\\
			\midrule
			max&-3.61&-3.64&-3.64&-3.96&-4.01&-4.02&-4.69&-5.21&-5.16\\
			95\% &-4.69&-4.66&-4.67&-5.44&-5.39&-5.39&-6.72&-6.76&-6.77\\
			\midrule
			\multicolumn{10}{@{}l@{}}{\textbf{B.} Different parameters}\\
			& \multicolumn{3}{c}{$\lambda=-0.03$} & \multicolumn{3}{c}{$\lambda=-0.06$} & \multicolumn{3}{c@{}}{$\lambda=-0.15$}\\
			\cmidrule(r){2-4} \cmidrule(lr){5-7} \cmidrule(l){8-10}
			Precision\;\; & \mcone{$\delta=0.01$} & \mcone{$\delta=0.02$} & \mcone{$\delta=0.05$} & \mcone{$\delta=0.01$} & \mcone{$\delta=0.02$} & \mcone{$\delta=0.05$} & \mcone{$\delta=0.01$} & \mcone{$\delta=0.02$} & \multicolumn{1}{c@{}}{$\delta=0.05$}\\
			\midrule
			max&-3.63&-3.65&-3.63&-3.64&-3.64&-3.59&-3.68&-3.68&-3.67\\
			95\% &-5.09&-4.9&-4.65&-4.9&-4.66&-4.45&-4.63&-4.52&-4.66\\
			\bottomrule
		\end{tabular*}
		\begin{tablenotes}
			\footnotesize Notes: In Panel A, we fix $\lambda=-0.06$ and $\delta=0.02$, simulate a
			time series of length $T=$ 20,000, discard the first 1,000 draws, and then
			compute the level of precision as $\log_{10} |EE_f|$. When applying the
			endogenous grid algorithm, we use an exponential grid for storage in the
			range $[0,2]$ with median value $0.5$, function iteration is implemented
			via linear interpolation and linear extrapolation, and we terminate the
			iteration process at precision $\varpi = 10^{-4}$. The rest of the setting
			is same to Section~\ref{sss:speculative-channel}. In Panel B, we fix the grid size to
			$K=100$ and $N=51$, and vary the parameters.
		\end{tablenotes}
	\end{threeparttable}
\end{table}

Summary statistics (maximum as well as $95$-th percentile) are reported in
\fref{tab:prec}, where $K$ is the number of grid points for storage, $N$ is the
number of state points for interest rates, and precision at $(x,z)$ is evaluated
as $\log_{10} |EE_f(x,z)|$. The results demonstrate that the endogenous grid
algorithm attains a high level of precision, with an Euler equation error
uniformly less than $0.025\%$.

\subsection{The Generalized Impulse Response Function}\label{ss:nl_irf}

To properly capture the nonlinear asymmetric dynamics of the
competitive storage model and effectively study the dynamic causal effect
of interest rates on commodity prices, we refer to the generalized impulse
response function proposed by \citet{Koop96}, which defines IRFs as
state-and-history-dependent random variables and is applicable to both linear
and nonlinear multivariate models. We are interested in calculating the 
IRFs when $(X_{t-1}, Z_{t-1})$ are held at different 
percentiles of the stationary distribution.  

Algorithm~\ref{alg:nl_irf} clarifies the computation process of the generalized
IRFs based on the setting of Section~\ref{sec:quant}. However, the algorithm can
be easily extended to handle more general settings as formulated in
Section~\ref{sec:or}, where more advanced interest rate and production setups
are allowed. To proceed, we define
\begin{equation*}
	F(x, z, Y) \coloneq \me^{-\delta} \left(
	\min\{x, x^*(z)\} - p^{-1}[f^*(x, z)]
	\right) + Y.
\end{equation*}

\begin{algorithm}[!htb]
	\caption{\,The generalized impulse response function}\label{alg:nl_irf}
	\small  
	\begin{enumerate}[label=Step \arabic*., font={\bfseries}, leftmargin=*, ref=\arabic*, itemsep=1mm]
		\item Initialization step. Choose 
		initial values for $X_{t-1}$ and $Z_{t-1}$, and
		a finite horizon $H$ and a size of Monte Carlo samples $N$.
		Furthermore, set the initial samples as
		\begin{equation*}
			\tilde X_{t-1}^n = X_{t-1}^n \equiv X_{t-1}\quad\text{and}\quad
			Z_{t-1}^n = \tilde Z_{t-1}^n \equiv Z_{t-1}.
		\end{equation*}
		\item Randomly sample $(H+1)\times N$ values of production shocks $\left\{Y_{t+h}^{S,n} \right\}_{(h,n)=(0,1)}^{(H,N)}$.
		\item (Baseline Economy) Sample $(H+1)\times N$ values of the exogenous 
		states and calculate the net production
		\begin{gather*}
			\left\{ Z_{t+h}^n \right\}_{(h,n)=(0,1)}^{(H,N)} \quad \text{where } \;
			Z_{t+h}^n \sim \Pi (Z_{t+h-1}^n, \cdot\, ) ,\\
			\left\{ Y_{t+h}^n \right\}_{(h,n)=(0,1)}^{(H,N)} \quad \text{where }
			\; Y_{t+h}^n=y \left( Z_{t+h}^n,Y_{t+h}^{S,n} \right).
		\end{gather*}
		\item (Impulse Shock Economy) Compute the period-$t$ exogenous states 
		$\{\tilde Z_t^n\}_{n=1}^N$ after the shock. Sample $H\times N$ 
		values of the exogenous states and calculate the net production
		\begin{gather*}
			\left\{ \tilde Z_{t+h}^n \right\}_{(h,n)=(1,1)}^{(H,N)} \quad \text{where } \; \tilde Z_{t+h}^n \sim \Pi (\tilde Z_{t+h-1}^n, \cdot\, ),\\
			\left\{ \tilde{Y}_{t+h}^n \right\}_{(h,n)=(0,1)}^{(H,N)} \quad \text{where }
			\; \tilde{Y}_{t+h}^n=y \left( \tilde{Z}_{t+h}^n,Y_{t+h}^{S,n} \right).              
		\end{gather*}
		\item For $h= 0, \dots, H$ and $n=1, \dots, N$, compute the sequence of availability  
		\begin{equation*}
			X_{t+h}^n = F(X_{t+h-1}^n, Z_{t+h-1}^n, Y_{t+h}^n)
			\quad \text{and} \quad 
			\tilde X_{t+h}^n
			= F(\tilde X_{t+h-1}^n, \tilde Z_{t+h-1}^n, \tilde{Y}_{t+h}^n).
		\end{equation*}
		\item For $h=0,\dots, H$, compute the period-$(t+h)$ impulse response
		$$
		IRF(t+h) = \frac{1}{N} \sum_{n=1}^N f^*(\tilde X_{t+h}^n, \tilde Z_{t+h}^n )
		- \frac{1}{N} \sum_{n=1}^N f^*( X_{t+h}^n, Z_{t+h}^n).
		$$
	\end{enumerate}
\end{algorithm}


The stationary distribution of the state process is computed based on ergodicity.
Once $f^*$, $i^*$, and $x^*$ are calculated, we simulate a time series of 
$\{(X_t, Z_t)\}_{t=1}^T$ according to 
\begin{equation*}
	X_{t+1} = \me^{-\delta} \left(\min\{X_t, x^*(Z_t)\} - p^{-1}[f^*(X_t,Z_t)]
	\right) + y \left( Z_{t+1},Y_{t+1}^S \right)
	\; \text{and} \;
	Z_{t+1} \sim \Pi(Z_t, \,\cdot) 
\end{equation*}
for $T=$ 200,000 periods, discard the first 50,000 samples, and use the
remainder to approximate the stationary distribution.

\section{MIT shock in a Constant Interest Rate Model}\label{sec:mit-shocks}

Consider the classical competitive storage model of \citet{deaton1992on}. 
As in their paper, we assume that gross interest rates are constant at 
$r_\ell>1$, and output is {\sc iid} and satisfies $Y_t = y(\eta_t)$. Denote 
this as economy 1.

Consider an alternative economy, in which an MIT shock increases interest 
rates from $r_\ell$ to $r_h$ at the beginning of period $t$, and interest rates 
return to $r_\ell$ in all subsequent periods. Denote this as economy 2. 

Let $f^*_j$ be the equilibrium pricing rule, $i^*_j$ be the equilibrium 
storage, $x^*_j$ be the equilibrium free-disposal threshold, and $T_j$ be the 
equilibrium price operator when interest rates are constant at $r_j$ for $j\in 
\{\ell,h\}$. 

\begin{lemma}\label{lm:mono_constR}
	$T_\ell f \geq T_h f$ for all $f \in \cC$. Moreover, $f_\ell^* \geq f_h^*$ 
	and $x^*_\ell \geq x^*_h$.
\end{lemma}

\begin{proof}
	Since interest rates are constant and the output process is {\sc iid}, the 
	equilibrium objects are functions of only the endogenous state. Fix $f \in 
	\cC$. Denote $x^*_{fj}$ as the free-disposal threshold (related to 
	candidate $f$) when interest rates are constant at $r_j$ for $j\in 
	\{\ell,h\}$. Since $r_\ell < r_h$, we have
	\begin{equation*}
		\frac{\me^{-\delta}}{r_\ell} \EE f \left(
		\me^{-\delta} [x^*_{f\ell} - p^{-1}(0)] + \hat Y
		\right)
		\geq \frac{\me^{-\delta}}{r_h} \EE f \left(
		\me^{-\delta} [x^*_{f\ell} - p^{-1}(0)] + \hat Y
		\right).
	\end{equation*}
	By the continuity and monotonicity of $f$ and the definition of the 
	free-disposal threshold, we have $x^*_{f \ell} \geq x^*_{f h}$. In 
	particular, since $f$ is chosen arbitrarily and $f^* \in \cC$, we have 
	$x^*_\ell \geq x^*_h$.
	
	Next, we show that $T_\ell f \geq T_h f$. Since $T_h f(x) = 0$ whenever $x 
	\geq x^*_{fh}$ by Lemma~A.5 and $T_\ell f \geq 0$, it suffices to verify 
	that $T_\ell f(x) \geq T_h f(x)$ for all $x < x^*_{fh}$. Suppose on the 
	contrary that $\xi_1\coloneq T_\ell f(x) < T_h f(x) =:\xi_2$ for some $x < 
	x^*_{fh}$. By the definition of the equilibrium price operator, we have
	\begin{align*}
		\xi_1 &= \min \left\{
		\max \left\{
		\frac{\me^{-\delta}}{r_\ell} \EE f \left(
		\me^{-\delta} [x - p^{-1}(\xi_1)] + \hat Y
		\right) - k, p(x)
		\right\}, p(b)
		\right\}  \\
		&\geq \min \left\{
		\max \left\{
		\frac{\me^{-\delta}}{r_\ell} \EE f \left(
		\me^{-\delta} [x - p^{-1}(\xi_2)] + \hat Y
		\right) - k, p(x)
		\right\}, p(b)
		\right\}  \\
		&\geq \min \left\{
		\max \left\{
		\frac{\me^{-\delta}}{r_h} \EE f \left(
		\me^{-\delta} [x - p^{-1}(\xi_2)] + \hat Y
		\right) - k, p(x)
		\right\}, p(b)
		\right\} = \xi_2,
	\end{align*}
	which is a contradiction. Hence we have shown that $T_\ell f \geq T_h f$ 
	for all $f \in \cC$.
	
	Suppose $T^{n-1}_\ell f \geq T^{n-1}_h f$. Since the equilibrium price 
	operator is order-preserving by Lemma~A.2, we have $T^n_\ell f = T_\ell 
	(T_\ell^{n-1} f) \geq T_\ell (T_h^{n-1} f) \geq T_h^{n}f$. By induction, we 
	have shown that $T_\ell^n f \geq T_h^n f$ for all $n$. Because $T_\ell^n f 
	\to f_\ell^*$ and $T_h^n f \to f^*_h$ by Theorem~A.1, letting $n \to 
	\infty$ gives $f^*_\ell \geq f^*_h$. 
\end{proof}

The next result indicates that, in the classical competitive storage model, an MIT interest rate shock only has a contemporaneous negative effect on commodity prices, which will die out immediately starting from the next period. In each future period, commodity prices are at least as large as their level when shocks are absent.

\begin{proposition}[MIT Shock]\label{pr:mit-shock-constant}
	If $X_{t-1}^1 = X_{t-1}^2$ with probability one, then $P_t^1 \geq P_t^2$ 
	and $P_{t+h}^1 \leq P_{t+h}^2$ with probability one for all $h \geq 1$.
\end{proposition}

\begin{proof}
	Since shocks are unexpected and both economies share the same interest rate 
	$r_\ell$ in period $t-1$ and the same output $Y_t = y(\eta_t)$ in each 
	period, we have
	\begin{equation*}
		X_t^1 = \me^{-\delta} i_\ell^*(X_{t-1}^1) + Y_t
		= \me^{-\delta} i_\ell^*(X_{t-1}^2) + Y_t = X_t^2
	\end{equation*}
	with probability one. Since in period $t$ interest rate in economy 2 is 
	higher, agents will account for the temporarily different incentive, hence 
	the price function in economy 2 is $f^*_{h\ell} = T_h f_\ell^*$ and the 
	inventory function is 
	\begin{equation*}
		i_{h\ell}^*(x) = \min\{x_\ell^*,x\} - p^{-1}[f_{h\ell}^*(x)] 
	\end{equation*}
	Based on Lemma~\ref{lm:mono_constR}, we have
	\begin{equation*}
		f^*_{h\ell} = T_h f_\ell^* \leq T_\ell f_\ell^* = f_\ell^*
	\end{equation*}
	and thus
	\begin{equation*}
		i^*_{h\ell}(x) \leq \min\{x_\ell^*,x\} - p^{-1}[f_{\ell}^*(x)] 
		= i^*_\ell(x). 
	\end{equation*}
	As we have shown that $X_t^1 = X_t^2$ with probability one, the above 
	inequalities imply that
	\begin{equation*}
		P_t^1 = f_\ell^*(X_t^1) = f_\ell^*(X_t^2) \geq f_{h\ell}^*(X_t^2) = 
		P_t^2 
	\end{equation*}
	and 
	\begin{equation*}
		X_{t+1}^1 = \me^{-\delta} i_\ell^*(X_t^1) + Y_{t+1} 
		\geq \me^{-\delta} i_{h\ell}^*(X_t^2) + Y_{t+1} = X_{t+1}^2
	\end{equation*}
	with probability one. Since starting from period $t+1$ interest rates 
	return to $r_\ell$, the equilibrium objects become $f_\ell^*$ and 
	$i_\ell^*$ in both economies. By the monotonicity of $f_\ell^*$ and 
	$i_{\ell}^*$, we have
	\begin{equation*}
		P_{t+1}^1 = f_\ell^*(X_{t+1}^1) \leq f_\ell^*(X_{t+1}^2) = P_{t+1}^2
	\end{equation*}
	and 
	\begin{equation*}
		X_{t+2}^1 = \me^{-\delta} i_\ell^*(X_{t+1}^1) + Y_{t+2} 
		\geq \me^{-\delta} i_\ell^*(X_{t+1}^2) + Y_{t+2} = X_{t+2}^2 
	\end{equation*}
	with probability one. By induction, we can then show that $P_{t+h}^1 \leq 
	P_{t+h}^2$ with probability one for all $h \geq 1$.
\end{proof}

\section{A Necessity Result for Discounting}\label{s:necess}

In this section, we show that Assumption~\ref{a:opt} is necessary
in a range of standard settings: no equilibrium solution exists if 
Assumption~\ref{a:opt} fails. Throughout, we focus on the case
where\textit{}
\begin{equation*}
	b=0, \quad \XX = (b, \infty)
	\quad \text{and} \quad
	M_t = m(Z_t).
\end{equation*}
Let $\bB_{\SS}$ be the Borel subsets of $\SS$. Let $\mathcal K(\SS, \XX)$ be 
the set of all stochastic kernels $\Psi(x, z, \diff x')$ from $\SS$ to $\XX$ 
such that 
\begin{equation*}
	Qf(x, z)
	\coloneq \sum_{z'} \int f(x', z') \Phi(z, z') \Psi(x, z, \diff x'),
	\qquad (x, z) \in \SS
\end{equation*}
has a stationary distribution $\pi$ on the set of Borel probability measures on
$\SS$ and is irreducible and weakly compact as an operator on $L_1 \coloneq
L_1(\SS, \bB_{\SS}, \pi)$.

Given $\Psi \in \mathcal K(\SS, \XX)$, we consider the functional equation
\begin{equation}\label{eq:modfe}
	f(x, z) = \max
	\left\{ 
	\me^{-\delta} \sum_{z'} m(z') \int f(x', z') \Phi(z, z') 
	\Psi(x, z, \diff x'), 
	p(x)
	\right\},
\end{equation}
where $m$ is a positive function on $\ZZ$ and $p$ is a decreasing map from $\XX$
to itself with $\int p \diff \pi < \infty$ and $p(x) \uparrow \infty$ as $x
\downarrow 0$. Letting $K$ be the positive linear operator on $L_1$ defined by 
\begin{equation*}
	Kf(x,z) \coloneq
	\me^{-\delta} \sum_{z'} m(z') \int f(x', z') \Phi(z, z') 
	\Psi(x, z, \diff x'),
	\qquad (x, z) \in \SS,
\end{equation*}
we can also write~\eqref{eq:modfe} as $f = Kf \vee p$. Same as above,
let $s(K)$ be the spectral radius of $K$ as a linear operator on $L_1$.

\begin{lemma}\label{l:specradk}
	If $\Psi \in \mathcal K(\SS, \XX)$, then $- \ln s(K) = \delta + \kappa(M)$.
\end{lemma}

\begin{proof}
	An induction argument shows that
	\begin{equation*}
		K^n \1 (x, z) = \me^{-\delta n} \EE_z \prod_{t=1}^n M_t.
	\end{equation*}
	Letting $(X_0, Z_0)$ be a draw from $\pi$, we obtain
	\begin{equation*}
		\| K^n \1 \| 
		= \EE K^n \1(X_0, Z_0)
		= \me^{-\delta n}
		\EE
		\left[ \EE_{Z_0}  \prod_{t=1}^n M_t  \right]
		= \me^{-\delta n} \EE  \prod_{t=1}^n M_t .
	\end{equation*}
	Hence, by weak compactness of $K$ \citep[which implies compactness of $K^2$ by
	Theorem~9.9 of ][]{schaefer1974banach} and Theorem~B2 of
	\citet{borovivcka2020necessary}, we have
	\begin{equation*}
		s(K) 
		= \lim_{n \to \infty} \| K^n \1 \|^{1/n}
		= \lim_{n \to \infty} 
		\left\{ 
		\me^{-\delta n} \EE  \prod_{t=1}^n M_t
		\right\}^{1/n}
		= \me^{-\delta} \lim_{n \to \infty} q_n^{1/n}.
	\end{equation*}
	It follows that $- \ln s(K) = \delta + \kappa(M)$, as was to be shown. 
\end{proof}

The following result demonstrates the necessity of Assumption~\ref{a:opt} in 
the above standard setting.

\begin{proposition}\label{pr:modfesol}
	If there exists an $f \in L_1$ and $\Psi \in \mathcal K(\SS, \XX)$ such
	that~\eqref{eq:modfe} holds, then $\delta + \kappa(M) \geq 0$.  If, in addition, 
	\begin{equation*}
		E \coloneq \setntn{(x, z) \in \SS}{Kf(x, z) < p(x)} 
	\end{equation*}
	obeys $\pi(E) > 0$, then $\delta + \kappa(M) > 0$.
\end{proposition}

\begin{proof}
	Let $f$ and $h$ have the stated properties.  
	Regarding the first claim, we note that, since $Q$ is weakly compact and
	irreducible on $L_1$, and since $m$ is positive and bounded, the operator
	$K$ is likewise weakly compact and irreducible. 
	By the Krein--Rutman theorem, combined with irreducibility
	and weak compactness of $K$ \citep[see, in particular, Lemma~4.2.11
	of ][]{meyer2012banach}, there exists an $e \in L_\infty \coloneq
	L_\infty(\SS, \bB_\SS, \pi)$ such that $e > 0$ $\pi$-almost everywhere and
	$K^* e = s(K) e$, where $K^*$ is the adjoint of $K$.
	
	By~\eqref{eq:modfe} we have $f = Kf \vee p$, so $f \geq Kf$.  Hence
	\begin{equation}\label{eq:ef}
		\inner{e, f} \geq \inner{e, Kf} = \inner{K^* e, f} = s(K) \inner{e, f}.
	\end{equation}
	Since $f \geq p > 0$ and $e$ is positive $\pi$-a.e., we have $\inner{e, f} >
	0$.  Hence $s(K) \leq 1$. Applying Lemma~\ref{l:specradk} now yields $\delta + \kappa(M) \geq 0$.
	
	Regarding the second claim, suppose that $\pi(E) > 0$, where $E$ is as
	defined in Proposition~\ref{pr:modfesol}. It then follows from $f = Kf \vee
	p$ that $f > Kf$ on a set of positive $\pi$-measure.  But then, since $e$ is
	positive $\pi$-a.e., we have $\inner{e, f} > \inner{e, Kf} = s(K) \inner{e,
		f}$, where the equality is from~\eqref{eq:ef}. As before we have $\inner{e,
		f} > 0$, so $s(K) < 1$. Using Lemma~\ref{l:specradk} again we obtain $\delta
	+ \kappa(M) > 0$.
\end{proof}


\ifpdf\pdfbookmark{References}{sec:references}\fi
\bibliographystyle{aer}
\bibliography{./cp}

@article{stachurski2021dynamic,
  title={Dynamic programming with state-dependent discounting},
  author={Stachurski, John and Zhang, Junnan},
  journal={Journal of Economic Theory},
  volume={192},
  pages={105190},
  year={2021},
  publisher={Elsevier}
}

@article{toda2021perov,
  title={Perov's contraction principle and dynamic programming with stochastic discounting},
  author={Toda, Alexis Akira},
  journal={Operations Research Letters},
  volume={49},
  number={5},
  pages={815--819},
  year={2021},
  publisher={Elsevier}
}

@article{goldstein2022commodity,
	title={Commodity financialization and information transmission},
	author={Goldstein, Itay and Yang, Liyan},
	journal={The Journal of Finance},
	volume={77},
	number={5},
	pages={2613--2667},
	year={2022},
	publisher={Wiley Online Library}
}

@article{baker2021financialization,
	title={The financialization of storable commodities},
	author={Baker, Steven D},
	journal={Management Science},
	volume={67},
	number={1},
	pages={471--499},
	year={2021},
	publisher={INFORMS}
}

@article{hirano2024bubble,
  title={Bubble economics},
  author={Hirano, Tomohiro and Toda, Alexis Akira},
  journal={Journal of Mathematical Economics},
  volume={111},
  pages={102944},
  year={2024},
  publisher={Elsevier}
}

@InCollection{barlevy2012rethinking,
  author    = {Barlevy, Gadi},
  booktitle = {New Perspectives on Asset Price Bubbles},
  publisher = {Oxford University Press},
  title     = {Rethinking Theoretical Models of Bubbles: Reflections Inspired by the Financial Crisis and Allen and Gorton's Paper ``Churning Bubbles''},
  year      = {2012},
  chapter   = {3},
  editor    = {Douglas D. Evanoff and George G. Kaufman and A. G. Malliaris},
  isbn      = {9780199844333},
  month     = mar,
  pages     = {41--62},
  doi       = {10.1093/acprof:osobl/9780199844333.003.0003},
}

@article{guerron2023bubbles,
  title={Bubbles, crashes, and economic growth: Theory and evidence},
  author={Guerron-Quintana, Pablo A and Hirano, Tomohiro and Jinnai, Ryo},
  journal={American Economic Journal: Macroeconomics},
  volume={15},
  number={2},
  pages={333--371},
  year={2023},
  publisher={American Economic Association 2014 Broadway, Suite 305, Nashville, TN 37203-2425}
}

@article{plantin2023asset,
  title={Asset bubbles and inflation as competing monetary phenomena},
  author={Plantin, Guillaume},
  journal={Journal of Economic Theory},
  volume={212},
  pages={105711},
  year={2023},
  publisher={Elsevier}
}

@book{meyer2012banach,
  title={Banach Lattices},
  author={Meyer-Nieberg, Peter},
  year={2012},
  publisher={Springer}
}

@article{borovivcka2020necessary,
  title={Necessary and sufficient conditions for existence and uniqueness of recursive utilities},
  author={Borovi{\v{c}}ka, Jaroslav and Stachurski, John},
  journal={The Journal of Finance},
  volume={75},
  number={3},
  pages={1457--1493},
  year={2020},
  publisher={Wiley Online Library}
}

@book{schaefer1974banach,
  title={Banach Lattices and Positive Operators},
  author={Schaefer, Helmut H},
  year={1974},
  publisher={Springer}
}

@techreport{frankel2018,
    title={Rising US Real Interest Rates Imply Falling Commodity Prices},
    author={Frankel, Jeffrey A.},
    year={2018},
    type={blog post},
    institution={Harvard Kennedy School, Belfer Center for Science and International Affairs}
}

@article{cody1991role,
	title={The role of commodity prices in formulating monetary policy},
	author={Cody, Brian J and Mills, Leonard O},
	journal={The Review of Economics and Statistics},
	pages={358--365},
	year={1991},
	volume={73},
	number={2},
	publisher={JSTOR}
}

@InCollection{christiano1999monetary,
  author    = {Lawrence J. Christiano and Martin Eichenbaum and Charles L. Evans},
  booktitle = {Handbook of Macroeconomics},
  publisher = {Elsevier},
  title     = {Monetary policy shocks: What have we learned and to what end?},
  year      = {1999},
  chapter   = {2},
  editor    = {John B. Taylor and Michael Woodford},
  pages     = {65--148},
  volume    = {1, Part A},
  doi       = {10.1016/s1574-0048(99)01005-8},
}

@article{bernanke2005measuring,
	title={Measuring the effects of monetary policy: a factor-augmented vector autoregressive (FAVAR) approach},
	author={Bernanke, Ben S and Boivin, Jean and Eliasz, Piotr},
	journal={The Quarterly Journal of Economics},
	volume={120},
	number={1},
	pages={387--422},
	year={2005},
	publisher={MIT Press}
}

@article{peersman2022international,
	title={International food commodity prices and missing (dis)inflation in the euro area},
	author={Peersman, Gert},
	journal={Review of Economics and Statistics},
	volume={104},
	number={1},
	pages={85--100},
	year={2022},
	publisher={MIT Press One Rogers Street, Cambridge, MA 02142-1209, USA journals-info~…}
}

@article{gospodinov2013commodity,
	title={Commodity prices, convenience yields, and inflation},
	author={Gospodinov, Nikolay and Ng, Serena},
	journal={Review of Economics and Statistics},
	volume={95},
	number={1},
	pages={206--219},
	year={2013},
	publisher={The MIT Press}
}

@article{eberhardt2021commodity,
	title={Commodity prices and banking crises},
	author={Eberhardt, Markus and Presbitero, Andrea F},
	journal={Journal of International Economics},
	volume={131},
	pages={103474},
	year={2021},
	publisher={Elsevier}
}

@article{byrne2013primary,
	title={Primary commodity prices: Co-movements, common factors and fundamentals},
	author={Byrne, Joseph P and Fazio, Giorgio and Fiess, Norbert},
	journal={Journal of Development Economics},
	volume={101},
	pages={16--26},
	year={2013},
	publisher={Elsevier}
}

@article{judd1992projection,
	title={Projection methods for solving aggregate growth models},
	author={Judd, Kenneth L},
	journal={Journal of Economic Theory},
	volume={58},
	number={2},
	pages={410--452},
	year={1992},
	publisher={Elsevier}
}

@article{frankel2014effects,
  title =	 {Effects of speculation and interest rates in a ``carry trade''
                  model of commodity prices},
  author =	 {Frankel, Jeffrey A.},
  journal =	 {Journal of International Money and Finance},
  volume =	 42,
  pages =	 {88--112},
  year =	 2014,
  publisher =	 {Elsevier}
}

@article{harvey2017long,
  title =	 {Long-run commodity prices, economic growth, and interest
                  rates: 17th century to the present day},
  author =	 {Harvey, David I. and Kellard, Neil M. and Madsen, Jakob B. and
                  Wohar, Mark E.},
  journal =	 {World Development},
  volume =	 89,
  pages =	 {57--70},
  year =	 2017,
  publisher =	 {Elsevier}
}

@article{kilian2009not,
  title =	 {Not all oil price shocks are alike: Disentangling demand and
                  supply shocks in the crude oil market},
  author =	 {Kilian, Lutz},
  journal =	 {American Economic Review},
  volume =	 99,
  number =	 3,
  pages =	 {1053--69},
  year =	 2009
}

@article{alquist2020commodity,
  title =	 {Commodity-price comovement and global economic activity},
  author =	 {Alquist, Ron and Bhattarai, Saroj and Coibion, Olivier},
  journal =	 {Journal of Monetary Economics},
  volume =	 112,
  pages =	 {41--56},
  year =	 2020,
  publisher =	 {Elsevier}
}

@article{scrimgeour2015commodity,
  title =	 {Commodity price responses to monetary policy surprises},
  author =	 {Scrimgeour, Dean},
  journal =	 {American Journal of Agricultural Economics},
  volume =	 97,
  number =	 1,
  pages =	 {88--102},
  year =	 2015,
  publisher =	 {Wiley Online Library}
}

@article{chambers1996theory,
  title =	 {A theory of commodity price fluctuations},
  author =	 {Chambers, Marcus J. and Bailey, Roy E.},
  journal =	 {Journal of Political Economy},
  volume =	 104,
  number =	 5,
  pages =	 {924--957},
  year =	 1996,
  publisher =	 {The University of Chicago Press}
}

@article{frankel1986expectations,
  title =	 {Expectations and commodity price dynamics: The overshooting
                  model},
  author =	 {Frankel, Jeffrey A.},
  journal =	 {American Journal of Agricultural Economics},
  volume =	 68,
  number =	 2,
  pages =	 {344--348},
  year =	 1986,
  publisher =	 {Oxford University Press}
}

@article{newbery1982optimal,
  title =	 {Optimal commodity stock-piling rules},
  author =	 {Newbery, David M. G. and Stiglitz, Joseph E.},
  journal =	 {Oxford Economic Papers},
  pages =	 {403--427},
  volume =       {34},
  year =	 1982,
  number =       {3},
}

@article{schorfheide2018identifying,
  title =	 {Identifying long-run risks: A Bayesian mixed-frequency
                  approach},
  author =	 {Schorfheide, Frank and Song, Dongho and Yaron, Amir},
  journal =	 {Econometrica},
  volume =	 86,
  number =	 2,
  pages =	 {617--654},
  year =	 2018,
  publisher =	 {Wiley Online Library}
}

@article{wright1982economic,
  title =	 {The economic role of commodity storage},
  author =	 {Wright, Brian D. and Williams, Jeffrey C.},
  journal =	 {The Economic Journal},
  volume =	 92,
  number =	 367,
  pages =	 {596--614},
  year =	 1982,
  publisher =	 {JSTOR}
}

@article{samuelson1971stochastic,
  title =	 {Stochastic speculative price},
  author =	 {Samuelson, Paul A.},
  journal =	 {Proceedings of the National Academy of Sciences},
  volume =	 68,
  number =	 2,
  pages =	 {335--337},
  year =	 1971,
  publisher =	 {National Acad Sciences}
}

@article{scheinkman1983simple,
  title =	 {A simple competitive model with production and storage},
  author =	 {Scheinkman, Jos{\'e} A. and Schechtman, Jack},
  journal =	 {The Review of Economic Studies},
  volume =	 50,
  number =	 3,
  pages =	 {427--441},
  year =	 1983,
  publisher =	 {Wiley-Blackwell}
}

@article{fortuin1971correlation,
  title =	 {Correlation inequalities on some partially ordered sets},
  author =	 {Fortuin, C. M. and Kasteleyn, P. W. and Ginibre, J.},
  journal =	 {Communications in Mathematical Physics},
  volume =	 22,
  number =	 2,
  pages =	 {89--103},
  year =	 1971
}

@article{arseneau2013commodity,
  title =	 {Commodity price movements in a general equilibrium model of
                  storage},
  author =	 {Arseneau, David M. and Leduc, Sylvain},
  journal =	 {IMF Economic Review},
  volume =	 61,
  number =	 1,
  pages =	 {199--224},
  year =	 2013
}

@article{cafiero2011the,
  title =	 {The empirical relevance of the competitive storage model},
  author =	 {Cafiero, Carlo and Bobenrieth, Eugenio S.A. and Bobenrieth, Juan R.A. and Wright, Brian D.},
  journal =	 {Journal of Econometrics},
  volume =	 162,
  number =	 1,
  pages =	 {44--54},
  year =	 2011
}

@article{gruber2018interest,
  title =	 {Interest rates and the volatility and correlation of commodity
                  prices},
  author =	 {Gruber, Joseph W. and Vigfusson, Robert J.},
  journal =	 {Macroeconomic Dynamics},
  volume =	 22,
  number =	 3,
  pages =	 {600--619},
  year =	 2018
}

@article{deaton1992on,
  title =	 {On the behavior of commodity prices},
  author =	 {Deaton, Angus and Laroque, Guy},
  journal =	 {The Review of Economic Studies},
  year =	 1992,
  volume =	 59,
  number = 1,
  pages =	 {1--23}
}

@article{deaton1996competitive,
  title =	 {Competitive storage and commodity price dynamics},
  author =	 {Deaton, Angus and Laroque, Guy},
  journal =	 {Journal of Political Economy},
  volume =	 104,
  number =	 5,
  pages =	 {896--923},
  year =	 1996
}

@InCollection{frankel2008,
  author =	 {Jeffrey A. Frankel},
  booktitle =	 {Asset Prices and Monetary Policy},
  publisher =	 {National Bureau of Economic Research},
  title =	 {The effect of monetary policy on real commodity prices},
  year =	 2008,
  chapter =	 7,
  editor =	 {John Y. Campbell},
  pages =	 {291--333},
  url =		 {http://ideas.repec.org/h/nbr/nberch/5374.html},
}

@article{frankel2008monetary,
  title =	 {Monetary policy and commodity prices},
  author =	 {Frankel, Jeffrey A.},
  journal =	 {VoxEU},
  year =	 2008
}

@article{frankel2008explanation,
  title =	 {An explanation for soaring commodity prices},
  author =	 {Frankel, Jeffrey A.},
  journal =	 {VoxEU},
  year =	 2008
}

@article{carroll2006the,
  title =	 {The method of endogenous gridpoints for solving dynamic
                  stochastic optimization problems},
  author =	 {Carroll, Christopher D.},
  journal =	 {Economics Letters},
  volume =	 91,
  number =	 3,
  pages =	 {312--320},
  year =	 2006
}

@Article{ma2020income,
  author =	 {Ma, Qingyin and Stachurski, John and Toda, Alexis Akira},
  title =	 {The income fluctuation problem and the evolution of wealth},
  journal =	 {Journal of Economic Theory},
  year =	 2020,
  volume =	 187,
  pages =	 105003,
  publisher =	 {Elsevier},
}

@Book{krasnosel2012approximate,
  title =	 {Approximate Solution of Operator Equations},
  publisher =	 {Springer Netherlands},
  year =	 2012,
  author =	 {Krasnosel'skii, M. A. and Vainikko, G. M. and Zabreyko,
                  R.P. and Ruticki, Y. B. and Stet'senko, V. V.},
}

@article{blackwell1965discounted,
  title =	 {Discounted dynamic programming},
  author =	 {Blackwell, David},
  journal =	 {Annals of Mathematical Statistics},
  volume =	 36,
  number =	 1,
  pages =	 {226--235},
  year =	 1965,
  publisher =	 {JSTOR}
}

@Article{feinberg2014fatou,
  author =	 {Feinberg, Eugene A. and Kasyanov, Pavlo O. and Zadoianchuk,
                  Nina V.},
  title =	 {{F}atou's lemma for weakly converging probabilities},
  journal =	 {Theory of Probability \& Its Applications},
  year =	 2014,
  volume =	 58,
  number =	 4,
  pages =	 {683--689},
  publisher =	 {SIAM},
}

@book{stachurski2009economic,
  title =	 {Economic Dynamics: Theory and Computation},
  author =	 {Stachurski, John},
  year =	 2009,
  publisher =	 {MIT Press}
}

@Article{tauchen1986finite,
  author =	 {Tauchen, George},
  title =	 {Finite state {M}arkov-chain approximations to univariate and
                  vector autoregressions},
  journal =	 {Economics Letters},
  year =	 1986,
  volume =	 20,
  number =	 2,
  pages =	 {177--181},
  publisher =	 {Elsevier},
}

@Article{Anzu13,
  author =	 {Alessio Anzuini and Marco J. Lombardi and Patrizio Pagano},
  title =	 {The impact of monetary policy shocks on commodity prices},
  year =	 2013,
  volume =	 9,
  number =	 3,
  pages =	 {125--150},
  url =		 {http://ideas.repec.org/a/ijc/ijcjou/y2013q3a4.html},
  journal =	 {International Journal of Central Banking},
}

@Article{Baue23,
  author    = {Michael D. Bauer and Eric T. Swanson},
  journal   = {{NBER} Macroeconomics Annual},
  title     = {A Reassessment of Monetary Policy Surprises and High-Frequency Identification},
  year      = {2023},
  pages     = {87--155},
  volume    = {37},
  doi       = {10.1086/723574},
  publisher = {University of Chicago Press},
}

@Article{Kili22,
  author    = {Lutz Kilian and Xiaoqing Zhou},
  journal   = {Journal of International Money and Finance},
  title     = {Oil prices, exchange rates and interest rates},
  year      = {2022},
  pages     = {102679},
  volume    = {126},
  doi       = {10.1016/j.jimonfin.2022.102679},
  publisher = {Elsevier {BV}},
}

@InCollection{Rame16,
  author    = {Valerie A. Ramey},
  booktitle = {Handbook of Macroeconomics},
  publisher = {Elsevier},
  title     = {Macroeconomic shocks and their propagation},
  year      = {2016},
  chapter   = {2},
  editor    = {John B. Taylor and Harald Uhlig},
  pages     = {71--162},
  volume    = {2},
  doi       = {10.1016/bs.hesmac.2016.03.003},
}

@TechReport{Goue22a,
  author      = {Christophe Gouel and Nicolas Legrand},
  institution = {CEPII},
  title       = {The role of storage in commodity markets: Indirect inference based on grains data},
  year        = {2022},
  number      = {2022-04},
  type        = {Working Paper},
  url         = {http://www.cepii.fr/CEPII/fr/publications/wp/abstract.asp?NoDoc=13445},
}

@Article{Rosa14,
  author     = {Rosa, Carlo},
  journal    = {Energy Economics},
  title      = {The high-frequency response of energy prices to {U.S}. monetary policy: {Understanding} the empirical evidence},
  year       = {2014},
  issn       = {0140-9883},
  pages      = {295--303},
  volume     = {45},
  doi        = {10.1016/j.eneco.2014.06.011},
  language   = {en},
  shorttitle = {The high-frequency response of energy prices to {U}.{S}. monetary policy},
  url        = {https://www.sciencedirect.com/science/article/pii/S0140988314001443},
  urldate    = {2023-04-26},
}

@Article{Fran85,
  author    = {Frankel, Jeffrey A. and Hardouvelis, Gikas A.},
  journal   = {Journal of Money, Credit, and Banking},
  title     = {Commodity prices, money surprises and {Fed} credibility},
  year      = {1985},
  issn      = {00222879},
  number    = {4},
  pages     = {425--438},
  volume    = {17},
  doi       = {10.2307/1992439},
}

@Article{Bobe17,
  author    = {Bobenrieth, Eugenio S.A. and Bobenrieth, Juan R.A. and Ernesto A. Guerra and Brian D. Wright and Di Zeng},
  journal   = {American Journal of Agricultural Economics},
  title     = {Putting the empirical commodity storage model back on track: crucial implications of a ``negligible'' trend},
  year      = {2021},
  number    = {3},
  pages     = {1034--1057},
  volume    = {103},
  doi       = {10.1111/ajae.12133},
  publisher = {Wiley},
}

@Article{Cafi15a,
  author   = {Cafiero, Carlo and Bobenrieth, Eugenio S.A. and Bobenrieth, Juan R.A. and Wright, Brian D.},
  journal  = {American Journal of Agricultural Economics},
  title    = {Maximum likelihood estimation of the standard commodity storage model: Evidence from sugar prices},
  year     = {2015},
  number   = {1},
  pages    = {122--136},
  volume   = {97},
  doi      = {10.1093/ajae/aau068},
}

@Article{Basa16,
  author    = {Suleyman Basak and Anna Pavlova},
  title     = {A model of financialization of commodities},
  year      = {2016},
  volume    = {71},
  number    = {4},
  pages     = {1511--1556},
  doi       = {10.1111/jofi.12408},
  url       = {http://dx.doi.org/10.1111/jofi.12408},
  journal   = {The Journal of Finance},
  publisher = {Wiley-Blackwell},
}

@Article{Fama87,
  author    = {Fama, Eugene F. and French, Kenneth R.},
  title     = {Commodity futures prices: Some evidence on forecast power, premiums, and the theory of storage},
  year      = {1987},
  volume    = {60},
  number    = {1},
  pages     = {55--73},
  issn      = {00219398},
  url       = {http://www.jstor.org/stable/2352947},
  journal   = {The Journal of Business},
  publisher = {The University of Chicago Press},
}

@Article{Koop96,
  author    = {Koop, Gary and Pesaran, M. Hashem and Potter, Simon M.},
  title     = {Impulse response analysis in nonlinear multivariate models},
  year      = {1996},
  volume    = {74},
  number    = {1},
  pages     = {119--147},
  issn      = {0304-4076},
  doi       = {10.1016/0304-4076(95)01753-4},
  url       = {http://www.sciencedirect.com/science/article/B6VC0-3XDS2R2-6/2/f8b1713c81765453e1a5e9a52593b52e},
  journal   = {Journal of Econometrics},
}

@article{Rout00,
  author = {Routledge, Bryan R. and Seppi, Duane J. and Spatt, Chester S.},
  journal = {The Journal of Finance},
  title = {Equilibrium Forward Curves for Commodities},
  year = {2000},
  issn = {0022-1082},
  number = {3},
  pages = {1297--1338},
  volume = {55},
  copyright = {Copyright © 2000 American Finance Association},
  jstor_articletype = {primary_article},
  jstor_formatteddate = {Jun., 2000},
  publisher = {Blackwell Publishing for the American Finance Association},
  url = {http://www.jstor.org/stable/222453},
}

@article{Casa05a,
  author = {Casassus, Jaime and Collin-Dufresne, Pierre},
  journal = {The Journal of Finance},
  title = {Stochastic Convenience Yield Implied from Commodity Futures and Interest Rates},
  year = {2005},
  issn = {1540-6261},
  number = {5},
  pages = {2283--2331},
  volume = {60},
  copyright = {© 2005 the American Finance Association},
  doi = {10.1111/j.1540-6261.2005.00799.x},
  language = {en},
  url = {https://onlinelibrary.wiley.com/doi/abs/10.1111/j.1540-6261.2005.00799.x},
  urldate = {2024-01-09},
}

@techreport{Boda21,
  author = {Vincent Bodart and François Courtoy and Erica Perego},
  institution = {CEPII},
  title = {{World Interest Rates and Macroeconomic Adjustments in Developing Commodity Producing Countries}},
  year = {2021},
  number = {2021-01},
  type = {Working Paper},
  url = {https://ideas.repec.org/p/cii/cepidt/2021-01.html},
}

@article{Tume16,
  author = {Tumen, Semih and Unalmis, Deren and Unalmis, Ibrahim and Unsal, D. Filiz},
  journal = {Energy Economics},
  title = {Taxing fossil fuels under speculative storage},
  year = {2016},
  issn = {0140-9883},
  pages = {64--75},
  volume = {53},
  doi = {10.1016/j.eneco.2014.12.017},
  language = {en},
  series = {Energy {Markets}},
  url = {https://www.sciencedirect.com/science/article/pii/S014098831400334X},
  urldate = {2023-01-05},
}

@Article{Gibs90,
  author   = {Gibson, Rajna and Schwartz, Eduardo S.},
  journal  = {The Journal of Finance},
  title    = {Stochastic {Convenience} {Yield} and the {Pricing} of {Oil} {Contingent} {Claims}},
  year     = {1990},
  issn     = {1540-6261},
  number   = {3},
  pages    = {959--976},
  volume   = {45},
  doi      = {10.1111/j.1540-6261.1990.tb05114.x},
}

@Article{Schw97,
  author   = {Schwartz, Eduardo S.},
  journal  = {The Journal of Finance},
  title    = {The Stochastic Behavior of Commodity Prices: Implications for Valuation and Hedging},
  year     = {1997},
  number   = {3},
  pages    = {923--973},
  volume   = {52},
  doi      = {10.1111/j.1540-6261.1997.tb02721.x},
}

@TechReport{Schm14a,
  author = {Stephanie Schmitt-Grohé and Mart\'in Uribe},
  title  = {Finite-State Approximation of {VAR} Processes: A Simulation Approach},
  year   = {2014},
  type   = {note},
  url    = {https://www.columbia.edu/~mu2166/tpm/tpm.html},
  note   = {Available from \url{https://www.columbia.edu/~mu2166/tpm/tpm.html}},
}

@Article{Cox81,
  author    = {Cox, John C. and Ingersoll, Jonathan E. and Ross, Stephen A.},
  journal   = {Journal of Financial Economics},
  title     = {The relation between forward prices and futures prices},
  year      = {1981},
  issn      = {0304-405X},
  number    = {4},
  pages     = {321--346},
  volume    = {9},
  doi       = {10.1016/0304-405X(81)90002-7},
  publisher = {Elsevier BV},
}

@Article{Gard76,
  author              = {Gardner, Bruce L.},
  journal             = {American Journal of Agricultural Economics},
  title               = {Futures Prices in Supply Analysis},
  year                = {1976},
  issn                = {0002-9092},
  number              = {1},
  pages               = {81--84},
  volume              = {58},
  copyright           = {Copyright © 1976 Agricultural & Applied Economics Association},
  doi                 = {10.2307/1238581},
  url                 = {http://www.jstor.org/stable/1238581},
}

@Article{Goue15,
  author   = {Gouel, Christophe and Legrand, Nicolas},
  journal  = {Journal of Applied Econometrics},
  title    = {Estimating the Competitive Storage Model with Trending Commodity Prices},
  year     = {2017},
  month    = jun,
  number   = {4},
  pages    = {744--763},
  volume   = {32},
  doi      = {10.1002/jae.2553},
}

\end{document}